\documentclass[letterpaper,11pt]{article}
\usepackage{microtype}
\usepackage[letterpaper,margin=1in]{geometry}
\usepackage[numbers,sort&compress]{natbib}
\usepackage{xcolor}
\usepackage{amsmath}
\usepackage{amssymb}
\usepackage{mathrsfs}
\usepackage{amsthm}
\usepackage{enumitem}
\usepackage{textgreek}
\usepackage{graphicx}
\usepackage{framed}
\usepackage[small]{caption}
\usepackage{booktabs}
\usepackage{tikz-cd}
\usepackage[unicode]{hyperref}
\usepackage{doi}
\usepackage[textsize=tiny]{todonotes}
\usepackage[framemethod=tikz]{mdframed}
\usepackage{thm-restate}
\usepackage{mathtools}

\theoremstyle{plain}
\newtheorem{theorem}{Theorem}
\newtheorem{lemma}[theorem]{Lemma}

\newtheorem{corollary}[theorem]{Corollary}

\newtheorem{observation}[theorem]{Observation}

\theoremstyle{definition}
\newtheorem{definition}[theorem]{Definition}

\theoremstyle{remark}

\newcommand{\namedref}[2]{\hyperref[#2]{#1~\ref*{#2}}}

\newcommand{\s}{\mspace{1mu}}
\newcommand{\mybox}[1]{\mspace{2mu}{\setlength{\fboxsep}{1.5pt}\color{lightgray}\boxed{\color{black}\scriptstyle #1}}\mspace{2mu}}
\newcommand{\A}{\mathsf{A}}
\newcommand{\B}{\mathsf{B}}
\renewcommand{\C}{\mathsf{C}}
\renewcommand{\L}{\mathsf{L}}
\renewcommand{\c}{\mathsf{c}}
\renewcommand{\a}{\mathsf{a}}
\renewcommand{\b}{\mathsf{b}}
\renewcommand{\u}{\mathsf{u}}
\newcommand{\y}{\mathsf{y}}
\newcommand{\g}{\mathsf{g}}
\newcommand{\w}{\mathsf{w}}
\newcommand{\Z}{\mathsf{Z}}
\renewcommand{\U}{\mathsf{U}}
\newcommand{\W}{\mathsf{W}}
\renewcommand{\S}{\mathsf{S}}
\newcommand{\M}{\mathsf{M}}
\renewcommand{\P}{\mathsf{P}}
\renewcommand{\O}{\mathsf{O}}
\newcommand{\I}{\mathsf{I}}
\newcommand{\X}{\mathsf{X}}
\newcommand{\bM}{\mybox{\M}}

\newcommand{\bO}{\mybox{\O}}
\newcommand{\bU}{\mybox{\U}}

\newcommand{\bMU}{\mybox{\M\U}}
\newcommand{\bPU}{\mybox{\P\U}}

\newcommand{\bI}{\mybox{\I}}
\newcommand{\bIO}{\mybox{\I\O}}
\renewcommand{\S}{\mathsf{S}}
\newcommand{\labels}{\mathsf{s}}
\newcommand{\dis}{\ensuremath{\operatorname{disj}}}

\DeclareMathOperator{\re}{\mathcal R}
\DeclareMathOperator{\rere}{\overline{\mathcal R}}
\DeclareMathOperator{\size}{\operatorname{size}}

\newcommand{\nodeconst}{\ensuremath{\mathcal{N}}}
\newcommand{\edgeconst}{\ensuremath{\mathcal{E}}}
\newcommand{\gen}[1]{\langle #1 \rangle}

\DeclareMathOperator{\poly}{poly}

\newenvironment{myabstract}
{\list{}{\listparindent 1.5em%
		\itemindent    \listparindent
		\leftmargin    1cm
		\rightmargin   1cm
		\parsep        0pt}%
	\item\relax}
{\endlist}

\newenvironment{mycover}
{\list{}{\listparindent 0pt
		\itemindent    \listparindent
		\leftmargin    1cm
		\rightmargin   1cm
		\parsep        0pt}%
	\raggedright
	\item\relax}
{\endlist}

\newcommand{\myaff}[1]{\,$\cdot$\, {\small #1}\par\smallskip}

\hypersetup{
	colorlinks=true,
	linkcolor=black,
	citecolor=black,
	filecolor=black,
	urlcolor=[rgb]{0,0.1,0.5},
	pdftitle={Distributed Lower Bounds for Ruling Sets},
	pdfauthor={Alkida Balliu, Sebastian Brandt, Dennis Olivetti}
}

\begin{document}

\begin{mycover}
	{\huge\bfseries\boldmath Distributed Lower Bounds for Ruling Sets \par}
	\bigskip
	\bigskip
	\bigskip
	
	\textbf{Alkida Balliu}
	\myaff{University of Freiburg}
	
	\textbf{Sebastian Brandt}
	\myaff{ETH Zurich}
	
	\textbf{Dennis Olivetti}
	\myaff{University of Freiburg}
\end{mycover}
\bigskip

\begin{myabstract}
Given a graph $G=(V,E)$, an $(\alpha,\beta)$-ruling set is a subset $S\subseteq V$ such that the distance between any two vertices in $S$ is at least $\alpha$, and the distance between any vertex in $V$ and the closest vertex in $S$ is at most $\beta$.
We present lower bounds for distributedly computing ruling sets.

More precisely, for the problem of computing a $(2,\beta)$-ruling set (and hence also any $(\alpha,\beta)$-ruling set with $\alpha >2$) in the LOCAL model of distributed computing, we show the following, where $n$ denotes the number of vertices, $\Delta$ the maximum degree, and $c$ is some universal constant independent of $n$ and $\Delta$.
\begin{itemize}
	\item Any deterministic algorithm requires $\Omega\left(\min\left\{ \frac{\log\Delta}{\beta\log\log\Delta},\log_\Delta n\right\}\right)$ rounds, for all $\beta\le c \cdot\min\left\{\sqrt{\frac{\log\Delta}{\log\log\Delta}},\log_\Delta n\right\}$. By optimizing $\Delta$, this implies a deterministic lower bound of $\Omega\left(\sqrt{\frac{\log n}{\beta\log\log n}}\right)$ for all $\beta\le c\,\sqrt[3]{\frac{\log n}{\log\log n}}$.
	\item Any randomized algorithm requires $\Omega\left(\min\left\{\frac{\log\Delta}{\beta\log\log\Delta},\log_\Delta\log n\right\}\right)$ rounds, for all $\beta\le c\cdot\min\left\{\sqrt{\frac{\log\Delta}{\log\log\Delta}},\log_\Delta\log n\right\}$. By optimizing $\Delta$, this implies a randomized lower bound of $\Omega\left(\sqrt{\frac{\log \log n}{\beta\log\log\log n}}\right)$ for all $\beta\le c\,\sqrt[3]{\frac{\log\log n}{\log\log\log n}}$.
\end{itemize}
For $\beta > 1$, this improves on the previously best lower bound of $\Omega(\log^* n)$ rounds that follows from the 30-year-old bounds of Linial [FOCS'87] and Naor [J.Disc.Math.'91] (resp.\ $\Omega(1)$ rounds if $\beta \in \omega(\log^* n)$).
For $\beta = 1$, i.e., for the problem of computing a maximal independent set (which is nothing else than a $(2,1)$-ruling set), our results improve on the previously best lower bound of $\Omega(\log^* n)$ \emph{on trees}, as our bounds already hold on trees.
For maximal independent set \emph{on general graphs}, a deterministic lower bound of $\Omega(\min\{\Delta, \log_{\Delta} n\})$ and a randomized lower bound of $\Omega(\min\{\Delta, \log_{\Delta} \log n\})$ were already known due to Balliu et al.\ [FOCS'19].
\end{myabstract}

\thispagestyle{empty}
\setcounter{page}{0}
\newpage
\tableofcontents
\newpage

\section{Introduction}
In this work, we study the problems of finding maximal independent sets (MIS) and ruling sets in the LOCAL model of distributed computing.
In the LOCAL model, each node of the input graph is considered as a computational device and each edge as a communication link.
Computation proceeds in synchronous rounds, where in each round each node can send a message of arbitrary size to each neighbor and then, after the messages arrive, perform some local computation.
Each node has to terminate at some point and then output its local part of the global solution, i.e., whether it is in the MIS (resp.\ ruling set) or not.
For a more detailed introduction to the LOCAL model, we refer the reader to Section~\ref{subsec:model}.

\paragraph{MIS}
The problem of finding an MIS in a given graph is one of the most central and well-studied problems in the LOCAL model.
Already in the '80s, the very first papers of the area \cite{KarpW85, Luby1985, Alon1986, Linial1992, Naor1991} gave first upper and lower bounds for the complexity of computing an MIS, and since then there has been an abundance of papers (e.g., \cite{Awerbuch89, panconesi96decomposition, BarenboimE10, SchneiderW10, LenzenW11, Barenboim2016, Barenboim2013, barenboim14distributed, ghaffari16improved, Rozhon2020,GGR2020}) studying the problem and variants thereof.
A major open question was whether an MIS can be computed deterministically in a polylogarithmic number of rounds (see, e.g., \cite{Linial1992}, or Open Problem 11.2 in the book by Barenboim and Elkin \cite{Barenboim2013})---this question was finally answered in the affirmative in a very recent breakthrough by Rozho\v n and Ghaffari \cite{Rozhon2020} on network decompositions.
In contrast, if randomization is allowed, already more than 30 years ago, Luby \cite{Luby1985} and Alon, Babai, and Itai \cite{Alon1986} presented $O(\log n)$-round algorithms for solving MIS, where $n$ denotes the number of nodes of the input graph. This is still the best randomized upper bound known if the complexity is expressed solely as a function of $n$.

On the lower bound side, the $\Omega(\log^* n)$-round bound from the '80s and early '90s by Linial \cite{Linial1992} and Naor \cite{Naor1991} was the state of the art, until Kuhn, Moscibroda, and Wattenhofer (KMW) \cite{Kuhn2004} proved in 2004 that there is no algorithm computing an MIS in $t = f(\Delta) + g(n)$ rounds (even allowing randomization) if $f(\Delta) \in o(\log \Delta / \log \log \Delta)$ and $g(n) \in o(\sqrt{\log n / \log \log n})$.
Here, $\log^* ()$ denotes the iterated logarithm and $\Delta$ the maximum node degree.
Finally, last year, the KMW bounds were improved and complemented by Balliu et al.\ \cite{Balliu2019} who showed that $f(\Delta) + g(n)$ rounds are not sufficient for deterministic algorithms if $f(\Delta) \in o(\Delta)$ and $g(n) \in o(\log n / \log \log n)$, and not sufficient for randomized algorithms if $f(\Delta) \in o(\Delta)$ and $g(n) \in o(\log \log n / \log \log \log n)$.
Due to an $O(\Delta + \log^* n)$-round upper bound by Barenboim, Elkin, and Kuhn \cite{barenboim14distributed}, the linear dependency on $\Delta$ is tight.

While the above bounds imply that the complexity of MIS on general graphs must lie in the polylogarithmic (in $n$) regime, the situation on trees is far less clear.
Both the KMW lower bounds and the lower bounds by Balliu et al.\ are achieved by first proving the same bounds for the problem of finding a maximal matching\footnote{The KMW lower bound is actually proved for an even easier problem, that is, for finding a $\poly(\log n)$ approximation of a minimum vertex cover.} and then obtaining the MIS bounds as an immediate corollary due to the fact that maximal matching on general graphs is essentially the same problem as MIS on line graphs.
As the line graph of any graph with $\Delta \geq 3$ contains a cycle (of length $3$), both lower bounds are not applicable on trees; in fact, as there seems to be no way around line graphs in order to transform the maximal matching bounds to MIS, there is little hope that the proofs can be adapted to work on trees.
Hence, on trees, the state of the art is given by the $\Omega(\log^* n)$-round lower bounds by Linial and Naor, exhibiting a large gap to the best known deterministic upper bound of $O(\log n / \log\log n)$ rounds on trees by Barenboim and Elkin~\cite{BarenboimE10}.
This suggests the following question.

\vspace{2pt}
\begin{mdframed}[backgroundcolor=gray!20, topline=false, rightline=false, leftline=false, bottomline=false] 
	\textbf{Question 1}
	
	\noindent Is polylogarithmic time needed for deterministically computing an MIS on trees or is there a (much) faster algorithm?
\end{mdframed}
\vspace{2pt}

\paragraph{Ruling sets}
Ruling sets are a generalization of maximal independent sets.
Let $\alpha \geq 2$, $\beta \geq 1$ be integers.
An $(\alpha,\beta)$-ruling set $S$ is a subset of the nodes of the input graph such that the distance between any two nodes from $S$ is at least $\alpha$ and any node not contained in $S$ has a distance of at most $\beta$ to the closest node in $S$.
An MIS is a $(2,1)$-ruling set.
We observe that an $(\alpha,\beta)$-ruling set is also an $(\alpha',\beta')$-ruling set for any $\alpha' \leq \alpha$ and $\beta' \geq \beta$, hence finding the latter is at least as easy as finding the former.
In particular, the problem of finding a $(2,\beta)$-ruling set for some $\beta > 1$ is at least as easy as the problem of finding an MIS.
Moreover, as our goal is to prove lower bounds, we can safely restrict attention to $\alpha = 2$ without affecting the generality of our results.

Due to their relation to MIS (but also as interesting combinatorial objects of their own), ruling sets have been a natural object of interest in the LOCAL model and are well-studied (see, e.g., \cite{Awerbuch89, SchneiderW10symmetry, Gfeller07, BishtKP13, Barenboim2016, ghaffari16improved}).
In particular, the computation of ruling sets often constitutes a useful subroutine in the computation of other objects, such as maximal matching \cite{Barenboim2016}, maximal independent set \cite{ghaffari16improved}, or distributed coloring \cite{GhaffariHKM18, ChangLP18}. This is not a surprise: also the computation of an MIS is an important step in many algorithms, and it is quite natural to replace this step by the computation of a $(2, \beta)$-ruling set for some $\beta > 1$, if the latter suffices and can be computed faster.
Hence, from the perspective of applications, a lower bound for MIS that also applies to such ruling sets can be considered as substantially more robust than a lower bound that cannot be extended to ruling sets.

Unfortunately, there is a simple argument why the existing lower bounds for MIS by KMW and Balliu et al.\ cannot be extended to $(2, \beta)$-ruling sets: as mentioned before, those lower bounds are achieved on line graphs; however, on line graphs already a $(2,2)$-ruling set can be found in $O(\log^* n)$ rounds as shown by Kuhn, Maus, and Weidner \cite{KuhnMW18}.
The best lower bound for $(2, \beta)$-ruling sets follows again from the lower bounds by Linial and Naor for MIS, and stands at $\Omega(\log^* n)$, both on trees and general graphs, up to some $\beta \in \Theta(\log^* n)$.
For $\beta \in \omega(\log^* n)$, no non-constant lower bound is known. 
In contrast, for up to polylogarithmic\footnote{As long as $\beta$ is not too close to $n$, also no subpolylogarithmic upper bounds are known for larger $\beta$, but the (at most) polylogarithmic regime is arguably the most interesting; for instance, we are not aware of any algorithms that make use of $(2, \beta)$-ruling sets where $\beta$ is superpolylogarithmic.} $\beta$, the best known upper bound (expressed solely as a function of $n$) for computing a $(2,\beta)$-ruling set is polylogarithmic in $n$ \cite{Awerbuch89,SchneiderW10symmetry,GGR2020}.

\vspace{2pt}
\begin{mdframed}[backgroundcolor=gray!20, topline=false, rightline=false, leftline=false, bottomline=false] 
	\textbf{Question 2}
	
	\noindent Is polylogarithmic time needed for deterministically computing a $(2, \beta)$-ruling set (for up to polylogarithmic $\beta$) or is there a (much) faster algorithm?
\end{mdframed}
\vspace{2pt}

\paragraph{Round elimination}
Traditionally, proving lower bounds in the LOCAL model has been a challenging task.
Until 2015, to the best of our knowledge, only about a handful of (non-trivial, non-global) lower bounds were known \cite{Linial1992, Naor1991, Kuhn2004, ChungPS14, GoosHS14, NaorS95}, with the only lower bound (as a function of $n$) beyond $\Omega(\log^* n)$ being the KMW lower bound.
A major obstacle seemed to be the lack of techniques that could be used to obtain (improved) lower bounds.

In 2016, things changed when it was discovered that a technique used in the proof for Linial's $\Omega(\log^* n)$-round lower bound is more widely applicable: Brandt et al.\ \cite{Brandt2016} used the technique, now known under the name \emph{round elimination}, to prove lower bounds for the Lov\'asz Local Lemma (LLL), sinkless orientation (as a special case of the LLL) and $\Delta$-coloring.
Since then, round elimination has been used to prove lower bounds for a variety of problems \cite{chang18complexity, BalliuHOS19, Brandt2019, Balliu2019, trulytight, binary}.

In 2019, Brandt \cite{Brandt2019} showed that round elimination can, in principle, be applied to (almost) any problem that is locally checkable\footnote{For a definition, see Section~\ref{subsec:problems}.}, by providing a so-called \emph{automatic} version of round elimination, which, roughly speaking, is a blueprint for obtaining a lower bound via round elimination in which the problem of interest can be inserted.
Unfortunately, for most problems, a crucial step in the general blueprint is (perhaps far) beyond the reach of current techniques, which is the reason why we have not seen a flurry of new lower bounds in the past year.
By using additional techniques inside this framework, a number of new lower bounds have been achieved \cite{Balliu2019, trulytight, binary}, but the framework itself is still far from being well-understood.
As such, we believe that obtaining a better understanding of (automatic) round elimination is one of the most promising research directions in the LOCAL model currently available and crucial for the design of new lower bounds.

Informally, the general idea of round elimination is as follows.
In order to prove a lower bound for some problem $\Pi_0$ of interest, we want to find a sequence of problems
\[
	\Pi_0 \rightarrow \Pi_1 \rightarrow \Pi_2 \rightarrow \dots
\]
such that for any two consecutive problems $\Pi_i, \Pi_{i+1}$, we have $T_{i+1} \leq T_i - 1$ whenever $T_i > 0$, where $T_j$ denotes the complexity of problem $\Pi_j$ for any $j$.
In other words, $\Pi_{i+1}$ is at least one round faster solvable than $\Pi_i$ as long as $\Pi_i$ is not $0$-round solvable, which we will call the \emph{round elimination property}.
Now all that is necessary for proving a lower bound of $T$ for problem $\Pi_0$ is to show that problem $\Pi_{T-1}$ is not $0$-round solvable, or equivalently, that the first $0$-round solvable problem in the sequence has index at least $T$. In fact, if $\Pi_{T-1}$ requires at least $1$ round, then by the round elimination property $\Pi_{T-2}$ requires at least $2$ rounds, $\Pi_{T-3}$ requires at least $3$ rounds, and so on.

Automatic round elimination explicitly generates such a sequence of problems for any locally checkable problem $\Pi_0$, by repeatedly applying a fixed process that takes some locally checkable problem $\Pi_i$ as input and returns $\Pi_{i+1}$.
The main issue with the obtained sequence is that the descriptions of the problems in the sequence usually become very complicated already for small indices; without applying any additional techniques, already the size of the problem description grows roughly doubly exponential \emph{for each subsequent problem}.
Hence, it is not surprising that the crucial step of determining the first $0$-round solvable problem $\Pi_j$ in the sequence cannot be performed (in general) with the currently available techniques.
Moreover, even if one could keep the problem description sizes reasonably small, no general method how to find the desired problem $\Pi_j$ is known.\footnote{Note that it is usually easy to check for a given problem whether it can be solved in $0$ rounds; the difficulty lies in first obtaining a concise (parameterized) description of the problems in the sequence.}

Nevertheless, when studying a specific problem $\Pi_0$, it seems reasonable to try to make the problems in the sequence easier to understand.
All currently known lower bound proofs via automatic round elimination follow the idea of modifying the problems in the sequence in a way that preserves the round elimination property while simplifying the problem descriptions, as suggested in \cite{Brandt2019}.
The proofs can be grouped into two categories, depending on the chosen modification.
\begin{enumerate}
	\item There exists a constant $c$ such that each problem in the sequence can be described\footnote{The description is required to be in a certain standardized form. For details, we refer to Section \ref{subsec:problems}.} by using at most $c$ output labels. Examples are \cite{Balliu2019, trulytight, binary}.
	
	\item The size of the problem description grows doubly exponentially when going from $\Pi_i$ to $\Pi_{i+1}$, for all $i$. Examples are \cite{BalliuHOS19, Brandt2019}.
\end{enumerate}

The idea of the second approach is to simplify the \emph{structure}\footnote{For instance, the simplification could consist in transforming a problem with complicated constraints using a large number of output labels into a (much easier to understand) coloring problem with a large number of colors.} of the descriptions of the problems in the sequence, but roughly preserve the \emph{size} of the descriptions.
The lower bound is achieved by showing that as long as the description size of a problem in the sequence is in $o(n)$ (or $(o(\Delta))$), the problem is not $0$-round solvable.
Hence, this approach only yields lower bounds of $\Omega(\log^* n)$ (resp.\ $\Omega(\log^* \Delta)$).

In contrast, the first approach can yield higher lower bounds, but requires finding a sequence of problems that can be described with a constant number of labels.
Considering that to obtain a \emph{good} lower bound we also must make sure that we do not reach a $0$-round solvable problem too fast, for many problems such a sequence might simply not exist.
In fact, characterizing the set of problems (or at least interesting subsets thereof) that admit such a sequence is an interesting open problem mentioned in \cite{trulytight}.
For instance, while we do not have a proof, we do not believe that for MIS such a sequence yielding a polylogarithmic lower bound exists.
This discussion raises the following question.

\vspace{2pt}
\begin{mdframed}[backgroundcolor=gray!20, topline=false, rightline=false, leftline=false, bottomline=false] 
	\textbf{Question 3}
	
	\noindent How can we design a problem sequence satisfying the round elimination property that yields a better lower bound than $\Omega(\log^* n)$ without restricting the problem descriptions to a constant number of labels?
\end{mdframed}
\vspace{2pt}

\subsection{Our results}
We prove the following result for deterministic algorithms.

\begin{restatable}{theorem}{detlb}\label{thm:detlb}	
In the LOCAL model, any deterministic algorithm that solves the $(2,\beta)$-ruling set problem requires $\Omega\left(\min \left\{  \frac{\log \Delta}{\beta \log \log \Delta}  ,  \log_\Delta n \right\} \right)$ rounds, for all $\beta \le c \cdot \min\left\{ \sqrt{\frac{\log \Delta}{\log \log \Delta}}  ,  \log_\Delta n  \right\}$, for some constant $c$ independent of $n$ and $\Delta$. 
\end{restatable}
By setting $\Delta \coloneqq 2^{\sqrt{\beta \log n \log \log n}}$, we maximize our lower bound as a function of $n$, thereby obtaining the following corollary.
\begin{restatable}{corollary}{detcor}\label{cor:detlb}
	In the LOCAL model, any deterministic algorithm that solves the $(2,\beta)$-ruling set problem requires  $\Omega\left(\sqrt{\frac{\log n}{\beta \log \log n}}\right)$ rounds, for all $\beta \le c \, \sqrt[3]{\frac{\log n}{\log \log n}}$, for some constant $c$ independent of $n$ and $\Delta$. 
\end{restatable}
This settles Question 2 for all $\beta \le c \, \sqrt[3]{\frac{\log n}{\log \log n}}$.
As any $(\alpha, \beta)$-ruling set is also a $(2, \beta)$-ruling set for all $\alpha > 2$, Theorem~\ref{thm:detlb} also holds for $(\alpha, \beta)$-ruling sets.
Moreover, since the given lower bounds already hold on trees, we obtain the following corollary, by setting $\beta=1$.

\begin{corollary}\label{cor:detmislb}
	In the LOCAL model, any deterministic algorithm that solves MIS on trees requires  $\Omega\left(\sqrt{\frac{\log n}{ \log \log n}}\right)$ rounds. 
\end{corollary}
This settles Question 1.
Corollaries~\ref{cor:detlb} and~\ref{cor:detmislb} provide the first polylogarithmic lower bounds for ruling sets, and for MIS on trees.
Due to an $O(\log n / \log \log n)$-round deterministic upper bound for MIS on trees by Barenboim and Elkin~\cite{BarenboimE10}, and  a polylogarithmic deterministic upper bound for $(2,\beta)$-ruling sets on general graphs following from the work by Ghaffari et al.~\cite{GGR2020}, the only remaining question for the given range of $\beta$ is the exponent in the polylog.

For randomized algorithms, we prove the following.

\begin{restatable}{theorem}{randlb}\label{thm:randlb}
		In the LOCAL model, any randomized algorithm that solves the $(2,\beta)$-ruling set problem w.h.p.\footnote{As usual, we say that an algorithm solves a problem with high probability if the global success probability is at least $1 - 1/n$.}\ requires $\Omega\left(\min \left\{  \frac{\log \Delta}{\beta \log \log \Delta}  ,  \log_\Delta \log n \right\} \right)$ rounds, for all $\beta \le c \cdot \min\left\{ \sqrt{\frac{\log \Delta}{\log \log \Delta}}  ,  \log_\Delta \log n  \right\}$, for some constant $c$ independent of $n$ and $\Delta$. 
\end{restatable}
By setting $\Delta \coloneqq 2^{\sqrt{\beta \log \log n \log \log \log n}}$, we maximize our lower bound as a function of $n$, thereby obtaining the following corollary.
\begin{restatable}{corollary}{randcor}\label{cor:randlb}
	In the LOCAL model, any randomized algorithm that solves the $(2,\beta)$-ruling set problem w.h.p.\ requires $\Omega\left(\sqrt{\frac{\log \log n}{\beta \log \log \log n}}\right)$ rounds, for all $\beta \le c \, \sqrt[3]{ \frac{\log \log n}{\log \log \log n}}$, for some constant $c$ independent of $n$ and $\Delta$. 
\end{restatable}

Again, this bound already holds on trees and we obtain the following corollary for MIS.

\begin{corollary}\label{cor:randmislb}
	In the LOCAL model, any randomized algorithm that solves MIS on trees w.h.p.\ requires $\Omega\left(\sqrt{\frac{\log \log n}{ \log \log \log n}}\right)$ rounds.
\end{corollary}

Note that Theorem~\ref{thm:randlb} implies that there is no randomized algorithm that solves the $(2,\beta)$-ruling set problem w.h.p.\ in $t = f(\Delta) + g(n)$ rounds if $f(\Delta) \in o\left(\frac{\log \Delta}{\beta \log \log \Delta}\right)$ and $g(n) \in o\left(\sqrt{\frac{\log \log n}{\beta \log \log \log n}}\right)$.
Hence, we obtain that the $O(\log \Delta + \log \log n / \log \log \log n)$-round randomized upper bound for MIS (and hence also $(2, \beta)$-ruling set) on trees by Ghaffari \cite{ghaffari16improved} cannot be improved substantially in both $\Delta$ and $n$ simultaneously, for any indicated $\beta$.
Furthermore, Corollary~\ref{cor:randmislb} provides the first progress on Open Problem 10.15 from the book by Barenboim and Elkin \cite{Barenboim2013} (on the lower bound side), asking for the randomized complexity of MIS on trees.

Our results are achieved by designing a sequence of problems with the round elimination property for $(2,\beta)$-ruling sets, where the number of used labels is non-constant.
More precisely, our problem sequence will satisfy that the number of labels used in the description of problem $\Pi_i$ is in $\Theta(i^\beta / (\beta!))$.
In particular, for the special case of MIS, the number of used labels grows linearly.
Hence, our construction of the problem sequence provides an answer to Question 3.

\subsection{Our techniques}
In order to successfully apply the round elimination technique, two main ingredients are required.
The first is \emph{finding} a good problem family: we need to define some family $\{\Pi_{i\ge0}\}$ such that the sequence $\Pi_0 \rightarrow \Pi_1 \rightarrow \dots$ satisfies the round elimination property and $\Pi_0$ is the problem for which we want to prove a lower bound.
The second ingredient is \emph{proving} that the defined sequence indeed satisfies the desired property.

While the second ingredient is technically involved, the conceptually crucial part is the first one, designing a good sequence of problems.
Usually, when applying the round elimination technique, finding the right problem family involves some guessing.\footnote{In rare cases the sequence suggests itself, e.g., for sinkless orientation \cite{Brandt2016} the sequence is obtained by setting $\Pi_0 = \Pi_1 = \dots$}
For instance, in \cite{Balliu2019} the problem family was found by trying to make each subsequent problem in the sequence \emph{look very similar} to the previous one while using the same output labels in the description (see \cite[Section 3.7]{Balliu2019}).
In the case where each problem in the family can be described using a constant number of labels, there is even very recent software available, written by Olivetti \cite{Olivetti2019}, that automatically searches the space of potential problems for small $\Delta$.
Unfortunately, for the MIS problem (and for ruling sets) this approach fails, suggesting that a constant number of labels is not sufficient.
Instead, we propose a more explicit and perhaps surprising approach to find the desired problem family, by first proving an \emph{upper bound} for the problem of interest such that the proof can be ``represented" via a similar sequence of problems.

As explained in \cite{fraigniaud16local, Brandt2019}, the round elimination technique can also be used to find upper bounds:
Instead of finding a problem sequence with the round elimination property, i.e., with the property that $T_{i+1} \leq T_i - 1$, the idea is to find a problem sequence with the property that $T_{i+1} \geq T_i - 1$.
This ensures that the index $j$ of the first $0$-round solvable problem $\Pi_j$ in the sequence (if such a problem exists) is an upper bound for the complexity of $\Pi_0$.
Accordingly, we will call a sequence satisfying $T_{i+1} \leq T_i - 1$ a \emph{lower bound sequence} and a sequence satisfying $T_{i+1} \geq T_i - 1$ an \emph{upper bound sequence}.
We note that the \emph{automatic sequence} provided by \emph{automatic} round elimination is both a lower and an upper bound sequence since there we have $T_{i+1} = T_i -1$; in fact, it can be seen as the \emph{tightest} sequence with the property $T_{i+1} \geq T_i - 1$.
In the following, we will use this automatic sequence to informally describe the intuition behind our approach.

\paragraph{Intuition behind our approach} In the round elimination framework, each problem is described via a list of ``allowed" configurations that specify which local output label configurations around a node or on an edge are considered correct. 
As mentioned before, in the automatic sequence these descriptions grow very fast.
On the other hand, due to the nature of $0$-round algorithms, it seems to be the case that in the first (or more generally, any) $0$-round solvable problem $\Pi_j$ only very few of those allowed configurations are actually required for the correctness of a given $0$-round algorithm.
In other words, $\Pi_j$ would still remain $0$-round solvable if we removed a large number of the allowed configurations; moreover, the remaining part of the problem usually has an intuitive interpretation.
Assuming that the previous problems in the automatic sequence behave similarly, we obtain the following intuition for each problem $\Pi_i$ with $i \leq j$:
\begin{itemize}
	\item[(1)] There is some small part of the problem description that has some intuitive meaning and is relevant for solving the problem in $j - i$ rounds, and
	\item[(2)] there are additional allowed configurations that seem to be an artifact of the automatic process that generates the sequence. 
\end{itemize}
Intuitively, Part (1) can be thought of as the \emph{essence} of the problem, and we argue that the information encoded therein should suffice to prove lower bounds.
Hence, we would like to restrict attention to Part (1).

If we had a concise description of problem $\Pi_j$ and the complete automatic sequence leading to $\Pi_j$, it would be straightforward to extract Part (1) of each problem and thereby obtain a comparably simple sequence $\Pi_0^* \rightarrow \Pi_1^* \rightarrow \dots$ of problems.
However, there are two issues: first, we do not have feasible access to $\Pi_j$ and the preceding sequence (otherwise we would be done), and second, for technical reasons, the obtained sequence is an upper bound sequence, but not a lower bound sequence (in general), i.e., even if we had such access, the fact that $T_{i+1} \leq T_i - 1$ is not satisfied prevents us from using the sequence in a \emph{lower bound} proof.

To solve the second issue, we make use of so-called \emph{wildcards}, a notion introduced in \cite{Balliu2019}.
We show for the case of MIS and ruling sets that, perhaps surprisingly, adding a sufficient number of wildcards to the allowed configurations in the problems from $\Pi_0^* \rightarrow \Pi_1^* \rightarrow \dots$ turns the upper bound sequence into a lower bound sequence that is ``tight enough" to yield a polylogarithmic lower bound.

Our solution to the first issue is to try to design an upper bound sequence $\Pi'_0 \rightarrow \Pi'_1 \rightarrow \dots$ that is as close to the desired sequence $\Pi_0^* \rightarrow \Pi_1^* \rightarrow \dots$ as possible, and then work with problem family $\{\Pi'_{i\ge0}\}$ instead of $\{\Pi_{i\ge0}^*\}$.
As the latter sequence is unknown, our guideline for designing $\{\Pi'_{i\ge0}\}$ will be \emph{simplicity}, following the above intuition that Part (1) of each $\Pi_i$ (i.e., $\Pi_i^*$) is small and intuitive.
A key idea in the design will be to introduce a \emph{coloring component} into the MIS and ruling set problems.
Roughly speaking, the purpose of this coloring component is that, with enough care, we can make sure that only the coloring part of the problem description grows when we go from $\Pi'_i$ to $\Pi'_{i+1}$, while the MIS (resp.\ ruling set) part remains unchanged.
This allows us to keep the structure of the problems in the sequence comparably simple, which in turn allows us to determine at which point in the sequence the problems become $0$-round solvable.

Essentially, our approach reduces the task of proving lower bounds to proving upper bounds, which usually is considered to be an easier task.\footnote{While the current literature uses round elimination primarily to prove lower bounds, this statement arguably also holds for lower/upper bounds \emph{via round elimination}. One main reason is that to make a problem given in the form specified by round elimination \emph{harder} (a technique instrumental for the design of upper bound sequences), we can simply discard allowed configurations, while to make a problem \emph{easier} (instrumental for lower bound sequences), more complicated operations have to be used.}
However, the designed algorithm should also have a ``simple representation" as an upper bound sequence, and this does not seem to be the case for existing ruling set algorithms.
Hence, we will design a new, genuinely different ruling set algorithm that gives state-of-the-art upper bounds in terms of $\Delta$ (which is the relevant dependency for the round elimination technique, from a technical perspective) and yields a simple upper bound sequence.

Figure \ref{fig:re-lb-ub} depicts the high level idea of what happens when using the round elimination technique to prove upper and lower bounds by doing simplifications.
\begin{figure}[h]
	\centering
	\includegraphics[width=0.9\textwidth]{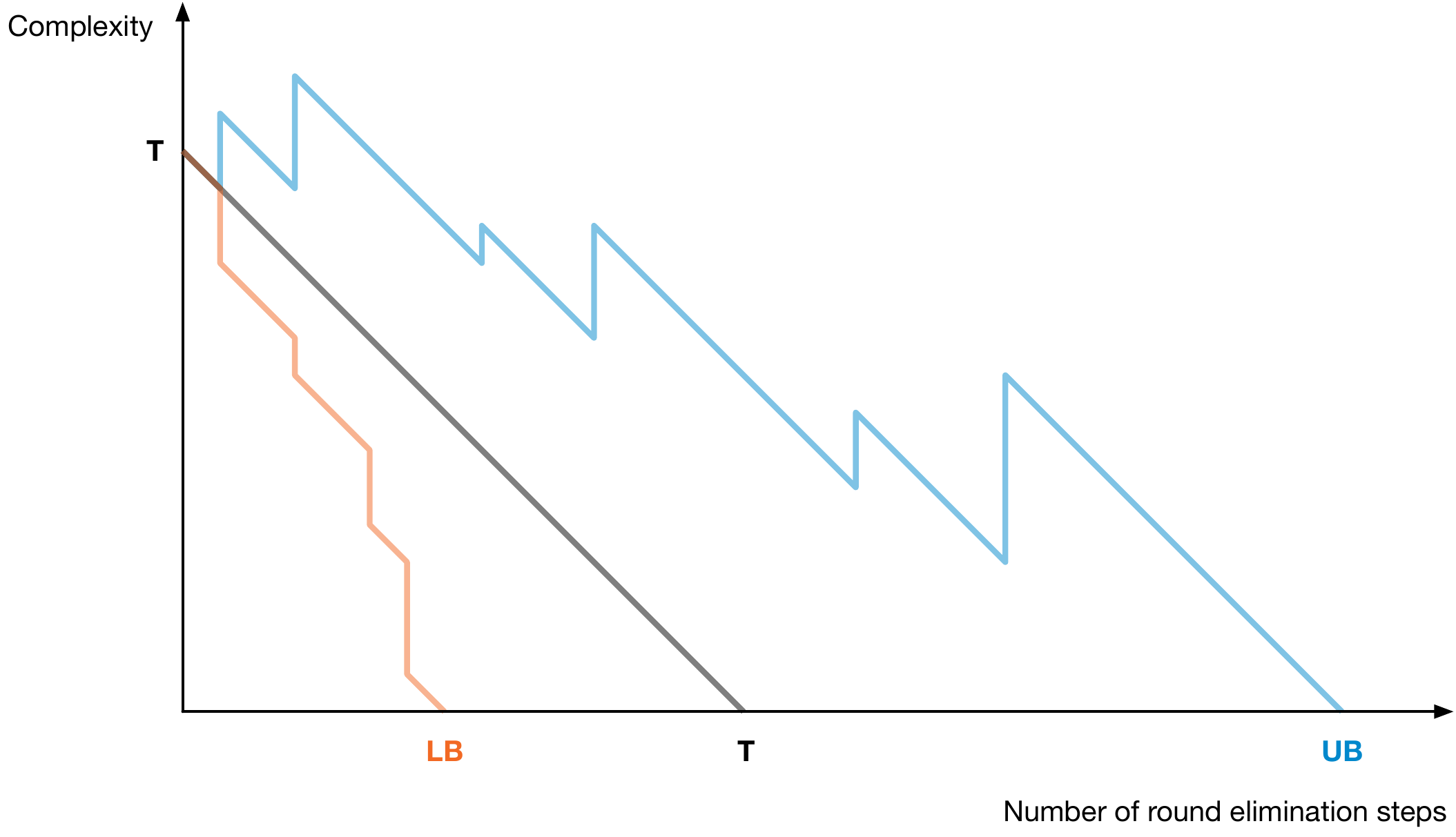}
	\caption{$T$ is the unknown complexity of some problem $\Pi$. By directly applying the round elimination technique to $\Pi$ we obtain a sequence of problems, each one being exactly one round easier than the previous one, and after $T$ steps we reach a $0$-round solvable problem. The problems in this sequence all lie on the grey line. Unfortunately, it is often practically not feasible to compute such a sequence. In order to prove lower bounds, we can try to relax the obtained problems to problems with simpler descriptions, and in some cases the simplified problems may become strictly easier. This is depicted in the orange lower bound sequence, and the obtained lower bound is given by the value of the horizontal axis where the orange line intersects it. Similarly, we may lose precision also in an upper bound sequence, depicted in blue, where we simplify problems in a manner that makes them potentially harder.}
	\label{fig:re-lb-ub}
\end{figure}

\paragraph{Approach}
To summarize, our approach works as follows.
First, we prove an upper bound for finding a $(2,\beta)$-ruling set (of which MIS is a special case) that can be represented by a comparably simple upper bound sequence.
To this end, we consider the initial problem $\Pi_0$ of the sequence as ``$(2,\beta)$-ruling set with some coloring component" and then introduce more and more colors into the problem over the course of the sequence, in a certain hierarchical manner.
Second, we insert (an increasing number of) wildcards into the problems in our sequence, and prove that this turns the upper bound sequence into a lower bound sequence that yields a polylogarithmic lower bound. 

While the individual parts of our approach are technically challenging, the approach itself is surprisingly simple.
Hence, we believe that this general approach does not only work for MIS and ruling sets but should also be applicable to other problems; however, as it involves, e.g., finding an upper bound proof that can be described well via a sequence of problems, obtaining new bounds using this approach is not automatic.
Moreover, we think that the idea of introducing a coloring component into problems that do not seem to have any particular relation to coloring should be more widely applicable; one intuitive reason is that, similar to wildcards, it gives a relatively simple way to represent \emph{progress} towards $0$-round solvability in the sequence, which seems like a necessary ingredient for designing a lower or upper bound sequence (which we can feasibly infer bounds from).

\subsection{Further discussion of related work}
\paragraph{MIS} 
The maximal independent set problem has been widely studied in the LOCAL model. Barenboim et al.\ showed that, if we also consider the dependency in $\Delta$, MIS can be solved in $O(\log^2\Delta) + 2^{O(\sqrt{\log\log n})}$ rounds~\cite{Barenboim2016}. Ghaffari improved this running time to $O(\log\Delta) + 2^{O(\sqrt{\log\log n})}$ \cite{ghaffari16improved}. The MIS problem has been studied also in specific classes of graphs \cite{SchneiderW10,BarenboimE10,BarenboimE11,Barenboim2016}. For example, for computing MIS on trees with randomized algorithms, Lenzen and Wattenhofer showed an $O(\sqrt{\log n}\log\log n)$-round algorithm \cite{LenzenW11,Barenboim2016}. This was later improved by Barenboim et al.\ to $O(\sqrt{\log n\log\log n})$ \cite{Barenboim2016}, and then further improved to $O(\sqrt{\log n})$ by Ghaffari \cite{ghaffari16improved}. Barenboim et al.\ also showed that MIS on trees can be solved in $O(\log\Delta\log\log\Delta + \log\log n/\log\log\log n)$ rounds \cite{Barenboim2016}. Ghaffari later improved this bound to $O(\log\Delta + \log\log n/\log\log\log n)$ rounds~\cite{ghaffari16improved}.

%Ghaffari studied MIS also in the CONGEST\footnote{The CONGEST model is the same as the LOCAL model with the difference that in CONGEST the size of the messages is bounded by $O(\log n)$ bits. We refer the reader to Section \ref{subsec:model} for more details on these models.} model, giving a randomized algorithm with a running time of $\min \{ O(\log \Delta \log \log n) + 2^{O(\sqrt{\log\log n \log \log \log n})}, \log \Delta \cdot 2^{O(\sqrt{\log \log n})} \}$ rounds \cite{Ghaffari19congest}. This was later improved to $O(\log \Delta \cdot \sqrt{\log \log n}) + 2^{O(\sqrt{\log\log n})}$ rounds by Ghaffari and Portmann~\cite{ghaffariPortman19}.

While all the above algorithms are randomized, Panconesi and Srinivasan provided a deterministic algorithm for solving MIS in $2^{O(\sqrt{\log n})}$ rounds~\cite{panconesi96decomposition}. Later, Barenboim, Elkin and Kuhn showed an $O(\Delta + \log^* n)$-round algorithm \cite{barenboim14distributed}. Very recently, Rozho\v n and Ghaffari proved that MIS can be solved deterministically in $\poly(\log n)$ rounds~\cite{Rozhon2020}. Meanwhile, the exponent of the polylog has been improved by Ghaffari et al.~\cite{GGR2020}. The MIS problem has been studied also in the CONGEST\footnote{The CONGEST model is the same as the LOCAL model with the difference that in CONGEST the size of the messages is bounded by $O(\log n)$ bits. We refer the reader to Section \ref{subsec:model} for more details on these models.} model (e.g., see \cite{Ghaffari19congest, ghaffariPortman19}).

\paragraph{Ruling sets}
Ruling sets have been introduced by Awerbuch et al.\ \cite{Awerbuch89}, where the authors showed how to construct $(\alpha, O(\alpha \log n))$-ruling sets in $O(\alpha \log n)$ deterministic rounds in the LOCAL model. Since then, there have been several works in this direction both in the deterministic and randomized setting. % and both in the LOCAL and CONGEST models of distributed computing. 
%In fact, 
As far as deterministic algorithms are concerned, Schneider, Elkin, and Wattenhofer showed how to get $(2,\beta)$-ruling sets in $O(\beta \Delta^{2/\beta} + \log^* n)$ rounds in the LOCAL model \cite{SEW13}. It is then easy to obtain an $(\alpha, (\alpha-1)\beta)$-ruling set of a graph $G$ in the LOCAL model by just computing a $(2, \beta)$-ruling set on the power graph $G^{\alpha - 1}$.

%Notice that, in the LOCAL model, it is possible to get an $(\alpha, (\alpha-1)\beta)$-ruling set of a graph $G$ by just computing a $(2, \beta)$-ruling set on the power graph $G^{\alpha - 1}$. %This reasoning does not directly apply to the CONGEST model, where the size of the messages is bounded by $O(\log n)$ bits. However, the algorithm of Awerbuch et al.\ can be modified to work in the CONGEST model. In fact, Henzinger, Krinninger, and Nanongkai sketched the arguments that show how to adapt it and get a CONGEST algorithm that gives $(\alpha, O(\alpha \log n))$-ruling sets in $O(\alpha \log n)$ rounds \cite{HKN16}. Later on, Kuhn, Maus, and Weidner gave a formal proof of these arguments \cite{KuhnMW18}. Also, the same authors showed how to obtain $(\alpha, (\alpha-1)\lceil \log_B n \rceil)$-ruling sets ($B\ge2$) in $O(\alpha B \log_B n)$ rounds. As a corollary, they get the same trade offs as in \cite{SEW13} and obtain a $(2, \beta)$-ruling set (for $\beta>2$) in $O(\beta \Delta^{2/\beta} + \log^* n)$ rounds for the CONGEST model.

If randomness is allowed, Gfeller and Vicari showed how to compute a version of $(1, O(\log\log\Delta))$-ruling sets where each node in the ruling set is allowed to have at most $O(\log^5 n)$ neighbors also in the ruling set, in $O(\log\log\Delta)$ rounds~\cite{Gfeller07}, and by then applying the algorithm of \cite{SEW13} on the graph induced by selected nodes, we can obtain an algorithm for $(2,\log \log n)$-ruling sets running in $O(\log \log n)$ time.
Kothapalli and Pemmaraju showed how to compute $(2,2)$-ruling sets in $O\left(\frac{\log \Delta}{(\log n)^\varepsilon} + (\log n)^{1/2+\varepsilon}\right)$ rounds, for any $\varepsilon > 0$ \cite{KP12}. One year later, Bisht, Kothapalli, and Pemmaraju provided a sparsifying procedure that can be used, together with some MIS algorithm, to obtain $(2,\beta)$-ruling sets (in a runtime that depends on the respective MIS algorithm)~\cite{BishtKP13}. For instance, by combining this sparsifying procedure with the MIS algorithm by Barenboim et al.~\cite{Barenboim2016}, a $(2,\beta)$-ruling set can be computed in $O(\beta\log^{1/(\beta-1/2)} \Delta) + 2^{O(\sqrt{\log\log n})}$ rounds. By using the improved MIS algorithm by Ghaffari  \cite{ghaffari16improved} instead, we obtain a runtime of $O(\beta\log^{1/\beta} \Delta) + 2^{O(\sqrt{\log\log n})}$ rounds, which can in turn be improved to $O(\beta\log^{1/\beta} \Delta) + \poly(\log\log n)$ rounds by making use of the $\poly(\log n)$-round network decomposition algorithm by Rozho\v n and Ghaffari \cite{Rozhon2020}.
%Lastly, Pai et al.\ studied randomized ruling sets in the CONGEST model. They showed how to compute $(2,3)$-ruling sets in $O(\log n/ \log\log n)$ rounds, and $(2,2)$-ruling sets in $O(\log\Delta (\log n)^{1/2 + \varepsilon} + \varepsilon \log n \log\log n)$ rounds \cite{PPPR017}.
Ruling sets have been investigated also in the more restrictive CONGEST model (e.g., see \cite{HKN16, KuhnMW18, PPPR017}).

\section{Background}\label{sec:preliminaries}

\subsection{Model}\label{subsec:model}

\paragraph{The LOCAL model}
The model of computation used in this paper is the widely studied LOCAL model of distributed computing \cite{Peleg2000}. In this model, each node of the input graph has a unique identifier from $1$ to $\poly n$, and the computation proceeds in synchronous rounds. At each round, each node can send a message of arbitrary size to each neighbor, and, after receiving the messages from its neighbors, perform some local computation of arbitrary complexity. In the LOCAL model, each node knows initially its unique identifier and its degree. As commonly done in this context, we also assume that each node knows the number of nodes $n$ in the graph (or a polynomial upper bound of it) and the maximum degree $\Delta$.
Clearly, this can make the task of proving lower bounds only harder.
Each node executes the same algorithm (which is what we call a distributed algorithm), and each node has to terminate at some point and then output its local part of the global solution, e.g., in the case of MIS whether the node is in the MIS or not.
The runtime of such a distributed algorithm is the number of synchronous rounds until the last node terminates.
In the randomized version of the LOCAL model, each node additionally has access to a stream of private random bits.
We will study Monte Carlo algorithms that solve the desired problem with high probability, that is, the global success probability must be at least $1-1/n$.

Another well-studied model in the area of distributed computing is the CONGEST model \cite{Peleg2000}, which is defined as the LOCAL model with the only difference that the size of each message sent between the nodes is restricted to $O(\log n)$ bits.
As the CONGEST model is strictly weaker than the LOCAL model, our lower bounds hold also in the CONGEST model.

\paragraph{The Port Numbering model}
Our results hold in the LOCAL model of distributed computing, however, for technical reasons we pass through the Port Numbering (PN) model, in the sense that we first show how to obtain our results in the PN model, and then lift them to the LOCAL model. The PN model is a variant of the LOCAL model where nodes do not have identifiers, but each node $v$ has an internal ordering of its incident edges given by an arbitrary assignment of (pairwise distinct) so-called \emph{port numbers} from $1$ to $\deg(v)$ to the edges.
This model is also synchronous, and, as in the LOCAL model, the size of the messages and the computational power of each node is not bounded. In the randomized version of the PN model, each node has access to a stream of private random bits and we require that randomized algorithms succeed with high probability.

To be able to apply the round elimination framework, we also need that \emph{edges} have port numbers; in other words, we assume that an orientation of the edges is given.
However, this is just a technical detail that does not have any effect on our argumentation, and as such we will ignore it in the following.\footnote{For the interested reader, we note that the edge orientations are (only) required in the proof of the round elimination theorem \cite[Theorem 1, arxiv version]{Brandt2019}, i.e., in the proof of the statement asserting that a problem $\Pi_1$ constructed in a fixed deterministic manner from a given problem $\Pi_0$ has a complexity of precisely $1$ round less than $\Pi_0$. Conveniently, given the theorem, to prove a lower bound, only the mentioned construction is relevant, allowing us to ignore, e.g., the edge orientations needed for the theorem to hold.}
Note that, in the LOCAL model, such an edge orientation can be obtained from the unique identifiers in one round; therefore also the presented upper bounds do not change asymptotically if we assume that an edge orientation is given.
 
\subsection{Problems}\label{subsec:problems}

In the round elimination framework a problem is characterized by an alphabet $\Sigma$ of labels, a \emph{node constraint} \nodeconst{} and an \emph{edge constraint} \edgeconst.
We will only consider problems defined on $\Delta$-regular graphs in this formalism, since, as we will later see, this is enough for our purposes.
The node constraint \nodeconst{} is a collection of words of length $\Delta$ over the alphabet $\Sigma$, and the edge constraint \edgeconst{} is a collection of words of length $2$ over $\Sigma$.
The same label can appear several times in a word and the order of the elements that compose a word does not matter, hence each word technically is a multiset.
We call a word in \nodeconst{} a \emph{node configuration} and a word in \edgeconst{} an \emph{edge configuration}.

Let $G=(V,E)$ be our input graph and let $A=\{ (v,e)\in V\times E~|~v\in e\}$ be the set that contains all pairs $(\mbox{node}, \mbox{incident edge})$.
The output for a problem in this formalism is given by a labeling of each $(v,e)\in A$ with one element from $\Sigma$.
Put otherwise, each node has to output an element of the set $\Sigma$ on each incident edge.
We say that such an output is \emph{correct} if it satisfies \nodeconst{} and \edgeconst, i.e., for each node $v' \in V$, the collection of $\Delta$ output labels assigned to the $(v,e) \in A$ with $v = v'$ is a node configuration listed in $\nodeconst$, and for each edge $e' \in E$, the two output labels assigned to the $(v,e)$ with $e = e'$ is an edge configuration listed in $\edgeconst$.

We use regular expressions to represent (collections of) node and edges configurations.
For example, the expression $\P\s\O^{\Delta-1}$ describes a node configuration that consists of exactly one label $\P$ and $\Delta-1$ labels $\O$.
Similarly, the expression $\M\s[\P\O]$ describes a collection of edge configurations that consists of one label $\M$ and the other label can be either $\P$ or $\O$, i.e., $\M\s[\P\O] = \{ \M\P, \M\O \}$.
We call a part of an expression such as $[\P\O]$, where we have a choice between different labels, a \emph{disjunction}.
While technically an expression containing a disjunction describes a set of configurations, we will use the term \emph{configuration} also for such an expression, for simplicity.
In order to explicitly specify that the expression contains a disjunction, we will use the term \emph{condensed configuration}.
Moreover, we will say that a configuration is \emph{contained in} a condensed configuration if we can obtain the former from the latter by picking a choice in each disjunction.

With a few exceptions, all problems from a large class of problems of interest in the LOCAL model, so-called \emph{locally checkable} problems, can be described in this formalism.
A locally checkable problem is simply a problem for which the correctness of a solution can be verified by checking whether the $O(1)$-hop neighborhood of each node is locally correct.
For technical reasons, locally checkable problems whose definitions involve small cycles (such as determining for each node whether it is contained in a triangle) cannot be described in the above formalism. 
For example, consider the triangle-detection problem, and consider a graph $G$ that is a triangle. Assume for a contradiction that there is an alphabet $\Sigma$, and node and edge constraints, satisfying that a graph can be labeled correctly if and only if it contains a triangle. Consider a valid labeling for $G$.
We can construct a $6$-cycle $H$ satisfying that each node and edge configuration that appears in $H$ also appears in $G$ (that is, $H$ is a lift of $G$). Hence, the labeling is valid, even if $H$ does not contain any triangle,  contradicting the correctness of the labeling.
Hence, for simplicity, in the remainder of the paper we will use the term ``locally checkable" for (locally checkable) problems that are not of this kind.

In the following we present two examples highlighting how we arrive at the description of a problem in the new formalism.
In Section~\ref{sec:equivalence}, we will show more formally that the given descriptions capture the MIS and ruling set problems.

\paragraph{Example: MIS}
Let us see, for example, how we can describe the MIS problem in this formalism.
We define $\Sigma=\{\M,\P,\O \}$.
We will use the node constraint to represent whether a node is in the independent set or not.
Nodes that are in the independent set must output the label $\M$ (as in ``in the MIS'') on all incident edges.
For nodes that are not in the independent set, we have to make sure that at least one neighbor is in the independet set.
To this end, we require that nodes that are not in the independent set \emph{point} to a neighbor that is in the independent set, thereby ensuring maximality.
In other words, these nodes must output a label $\P$ (as in ``pointer'') on exactly one incident edge and the label $\O$ (as in ``other'') on all the other $\Delta - 1$ incident edges.
Now the edge constraint must guarantee that no two neighbors are in the MIS, hence $\M\s\M\notin \edgeconst$, and that a pointer points to a node that is in the MIS, hence $\P\s\M\in \edgeconst$, but $\P\s\P\notin \edgeconst$, and $\P\s\O\notin \edgeconst$.
In order to capture the situation where a node not in the MIS has several neighbors in the MIS, we must allow $\M\s\O\in \edgeconst$.
Also, since two nodes not in the MIS may be neighbors, $\O\s\O\in \edgeconst$.
This leads to the following formal definition of the node and edge constraint.

\begin{equation*}
\begin{aligned}
\begin{split}
\nodeconst\text{:} \\ 
& \quad\M^{\Delta} \\
& \quad\P\s\O^{\Delta-1} 
\end{split}
\qquad
\begin{split}
\edgeconst\text{:} \\
& \quad \M\s[\P\O] \\
& \quad \O\s\O
\end{split}
\end{aligned}
\end{equation*}

\paragraph{Example: (2,2)-ruling set} 
In order to encode the $(2,2)$-ruling set problem we need to use a larger set of labels compared to the one used for the MIS problem. Let $\Sigma=\{\M,\P_1,\P_2,\O_1,\O_2\}$. Intuitively, similarly as before, the $\M$ label can be seen as the ``I am in the ruling set" label, while the labels $\P_1$ and $\P_2$ are ``pointer'' labels that are used to point to nodes in the ruling set and to nodes that are at distance $1$ from a node in the ruling set. Notice that, as a $(2,1)$-ruling set (i.e., MIS) solves the $(2,2)$-ruling set problem, the encoding of the $(2,2)$-ruling set problem will contain the node and edge configurations of the MIS problem.
For instance, a node in the ruling set will output $\M^{\Delta}$. Nodes at distance $1$ from a node in the ruling set may output either $\P_1\s\O_1^{\Delta-1}$ or $\P_2\s\O_2^{\Delta-1}$, but those at distance $2$ must output $\P_2\s\O_2^{\Delta-1}$. On the edge side, we must guarantee that, for any pair of nodes in the ruling set, they do not share an edge, hence $\M\s\M\notin \edgeconst$. Also, a pointer of type $1$ must point to a node in the ruling set, while a pointer of type $2$ must point to a node at distance at most $1$ from a node in the ruling set, hence $\M\s[\P_1\P_2]\in \edgeconst$ and $\O_1\s\P_2\in \edgeconst$. On the other hand, we want to forbid bad pointing. In fact, nodes at distance $1$ from a node in the ruling set must not be able to point to a node that is not in the ruling set, hence $\P_1\s[\O_1\O_2\P_1\P_2]\notin \edgeconst$. Also, nodes at distance $2$ from a node in the ruling set must not point to another node that is at distance $2$ as well, hence $\P_2\s[\O_2\P_2]\notin \edgeconst$. More precisely, the $(2,2)$-ruling set problem can be encoded in the formalism as follows.

\begin{equation}
\begin{aligned}
\begin{split}
\nodeconst\text{:} \\
& \quad\M^{\Delta} \\
& \quad\P_1\s\O_1^{\Delta-1} \\
& \quad\P_2\s\O_2^{\Delta-1} 
\end{split}
\qquad
\begin{split}
\edgeconst\text{:} \\
& \quad \M\s[\P_1\O_1\P_2] \\
& \quad \O_1\s[\O_1\O_2\P_2] \\
& \quad \O_2\s\O_2
\end{split}
\end{aligned}
\label{eq:22rs}
\end{equation}

\subsection{Round elimination}\label{sec:resec}

In our proofs, we will use the result of \cite[Theorem 4.3]{Brandt2019}, that is at the core of the round elimination technique. On a high level, this theorem says that, on $\Delta$-regular high-girth graphs, given a locally checkable problem $\Pi$ with time complexity $T$, there exists a locally checkable problem $\Pi''$ with time complexity $T-1$. The procedure of showing this theorem goes through an intermediate problem, that we call $\Pi'$. Given $\Pi$, Brandt \cite{Brandt2019} shows how to construct first $\Pi'$ and then $\Pi''$. We will formally define these problems and then we will see an example where we compute $\Pi'$ and $\Pi''$ starting from a specific problem $\Pi$. Let $\Sigma_{\Pi}$, $\nodeconst_{\Pi}$, and $\edgeconst_{\Pi}$ be the alphabet of labels, the node constraint, and the edge constraint for problem $\Pi$, respectively.

\paragraph{Problem $\Pi'$} In order to define problem $\Pi'$, we must define the alphabet $\Sigma_{\Pi'}$, the node constraint $\nodeconst_{\Pi'}$, and the edge constraint $\edgeconst_{\Pi'}$.
\begin{itemize}
	\item $\Sigma_{\Pi'}$: The set of labels for $\Pi'$ is the set of all non-empty subsets of $\Sigma_{\Pi}$, i.e., $\Sigma_{\Pi'}=2^{\Sigma_{\Pi}} \setminus \{\{\}\}$.
	
	\item $\edgeconst_{\Pi'}$: We construct the edge constraint in the following way. Consider a configuration $\A_1\s \A_2$, where $\A_1,\A_2\in \Sigma_{\Pi'}$, such that, for all $(\a_1, \a_2) \in \A_1 \times \A_2$, it holds that $\a_1\s \a_2\in \edgeconst_{\Pi}$ (notice that, by construction of $\Sigma_{\Pi'}$, it holds that $\a_1, \a_2 \in \Sigma_{\Pi}$). Let $\mathcal{A}$ be the collection of all such configurations.
We call a configuration $\A_1\s \A_2 \in \mathcal{A}$ \emph{non-maximal} if there exists another configuration $\A'_1\s \A'_2\in \mathcal{A}$ such that $\A_i \subseteq \A'_i$ for all $i \in \{ 1, 2 \}$, and $\A_i \subsetneq \A'_i$ for at least one $i \in \{ 1, 2 \}$.
In other words, if we have a configuration $\A'_1\s \A'_2\in \mathcal{A}$ that is obtained from $\A_1\s \A_2$ by adding at least one element to at least one of $\A_1$ and $\A_2$, then we say that $\A_1\s \A_2$ is non-maximal. We delete all non-maximal configurations from $\mathcal{S}$, and what remains is our set $\edgeconst_{\Pi'}$ of configurations.
	
	\item $\nodeconst_{\Pi'}$: Consider a configuration $\B_1\s \B_2\s\ldots\s \B_\Delta$ where $\B_i\in\Sigma_{\Pi'}$ for all $i\in\{1,\dotsc,\Delta\}$, such that there exists a tuple $ (\b_1,\dotsc,\b_\Delta) \in \B_1 \times \dotsc \times \B_\Delta$ such that $\b_1\s \b_2\s\dotsc\s \b_\Delta \in \nodeconst_{\Pi}$. Let $\mathcal{B}$ be the collection of all such configurations. We delete from the set $\mathcal{B}$ all configurations that contain some set $\B_i$ that does not appear in any configuration in $\edgeconst_{\Pi'}$. The modified set $\mathcal{B}$ is our set $\nodeconst_{\Pi'}$.
\end{itemize}
For simplicity, we can (and will) assume that all labels that occur neither in $\edgeconst_{\Pi'}$, nor in $\nodeconst_{\Pi'}$, are also removed from $\Sigma_{\Pi'}$.

\paragraph{Problem $\Pi''$}
Similarly as before, we need to define the alphabet $\Sigma_{\Pi''}$, the node constraint $\nodeconst_{\Pi''}$, and the edge constraint $\edgeconst_{\Pi''}$. 

\begin{itemize}
	\item $\Sigma_{\Pi''}$: The set of labels for $\Pi''$ is the set of all non-empty subsets of $\Sigma_{\Pi'}$, i.e., $\Sigma_{\Pi''}=2^{\Sigma_{\Pi'}} \setminus \{\{\}\}$.
	
	\item $\nodeconst_{\Pi''}$: The node constraint is constructed as follows. Consider a configuration $\B_1\s \B_2\s\ldots\s \B_\Delta$ where $\B_i\in\Sigma_{\Pi''}$ for all $i\in\{1,\dotsc,\Delta\}$, such that for all $(\b_1,\dotsc,\b_\Delta) \in \B_1 \times \dotsc \times \B_\Delta$ it holds that $\b_1\s \b_2\s\dotsc \b_\Delta\in\nodeconst_{\Pi'}$. Let $\mathcal{B}$ be the collection of all such configurations. We delete from $\mathcal{B}$ all non-maximal configurations, i.e., all those configurations $\B_1\s\ldots\s \B_\Delta$ such that there exists some other configuration $\B'_1\s\ldots\s \B'_\Delta$ that is obtained from the former by adding at least one element to at least one of the $\B_i$ sets. After performing these deletions, we set $\nodeconst_{\Pi''} = \mathcal{B}$.
	
	\item $\edgeconst_{\Pi''}$: Consider a configuration $\A_1\s \A_2$, where $\A_1,\A_2\in \Sigma_{\Pi''}$, such that there exists a pair $(\a_1, \a_2) \in \A_1 \times \A_2$ such that $\a_1\s \a_2\in\edgeconst_{\Pi'}$. Let $\mathcal{A}$ be the collection of all such configurations. We delete from the set $\mathcal{A}$ all configurations that contain some set $\A_1$ or $\A_2$ that does not appear in any configuration in $\nodeconst_{\Pi''}$, then we set $\edgeconst_{\Pi''}=\mathcal{A}$.
\end{itemize}
Again, we can (and will) assume that all labels that occur neither in $\nodeconst_{\Pi''}$, nor in $\edgeconst_{\Pi''}$, are also removed from $\Sigma_{\Pi''}$.

As $\Pi'$ is uniquely defined by $\Pi$, we can define a function $\re(\cdot)$ that takes $\Pi$ as input and returns $\Pi'$.
Similarly, as $\Pi''$ is uniquely defined by $\Pi'$, we can define a function $\rere(\cdot)$ that takes $\Pi'$ as input and returns $\Pi''$.
With these definitions, we have $\Pi'' = \rere(\re(\Pi))$.
Note that $\rere(\cdot)$ can take any problem as input that is of the form specified by round elimination---it is not necessary that the input problem has been obtained by applying $\re(\cdot)$ to some problem.

Now \cite[Theorem 4.3]{Brandt2019} provides the following relation between a problem $\Pi$ and $\rere(\re(\Pi))$ that provides the fundament for automatic round elimination.
For technical reasons, the theorem itself only holds in the port numbering model, but we will show later how to lift the obtained bounds to the LOCAL model.

\begin{theorem}[\cite{Brandt2019}, rephrased]\label{thm:sebastien}
	Let $T > 0$. Consider a class $\mathcal G$ of graphs\footnote{Technically, the class of graphs has to satisfy a certain property, called $t$-independence in \cite{Brandt2019}, but since it is straightforward to check that our considered class of $\Delta$-regular high-girth graphs satisfies this property, we omit this detail.} with girth at least $2 T+2$, and some locally checkable problem $\Pi$. 
	Then, there exists an algorithm that solves problem $\Pi$ on $\mathcal G$ in $T$ rounds if and only if there exists an algorithm that solves problem $\rere(\re(\Pi))$ in $T-1$ rounds.
\end{theorem}

In more technical detail, for any pair $(n, \Delta)$, Theorem~\ref{thm:sebastien} holds for graph classes $\mathcal G = \mathcal G(n, \Delta)$ consisting of $n$-node graphs with maximum degree $\Delta$ and girth at least $T = T(n, \Delta) > 0$.
However, for simplicity, we will usually omit the dependency on $n$ and $\Delta$.
We note that Theorem~\ref{thm:sebastien} also holds if we add a proper input vertex coloring to the setting.
Moreover, we will assume that the input graphs satisfy the given girth requirement whenever we apply Theorem~\ref{thm:sebastien}.
In Section~\ref{sec:liftlocal}, we will see how this requirement affects the obtained bounds.

An interesting fact that we have not seen mentioned in \cite{Brandt2019} (or any other work) is that the equivalence breaks only \emph{in one direction} when we go from high-girth graphs to general graphs:
it is straightforward to go through the proof of \cite[Theorem 4.3]{Brandt2019} and check that even on general graphs, $\Pi$ can be solved in $1$ round given a solution to $\rere(\re(\Pi))$.\footnote{Roughly speaking, the output of a node in a correct solution for $\rere(\re(\Pi))$ is a collection of sets of sets of output labels for $\Pi$, and a node can infer a correct solution for $\Pi$ from it by collecting the outputs of the adjacent nodes and then choosing output labels from the seen sets in a certain manner. As the topology of the graph does not enter the argumentation, the obtained $1$-round transformation holds on general graphs.}
In other words, $\rere(\re(\Pi))$ is \emph{at most} one round faster solvable than $\Pi$.
Hence, any upper bound achieved via automatic round elimination holds on general graphs, both in the port numbering model and the LOCAL model (as the latter is a stronger model).
In particular, this is true for our upper bounds for ruling sets.

\paragraph{Example: sinkless orientation}
Let $\Pi$ be the sinkless orientation problem, where the goal is to consistently orient edges such that no node is a sink. In this example, we will see how to encode sinkless orientation in the round elimination framework, and we will see what the problems $\re(\Pi)$ and $\rere(\re(\Pi))$ look like. 

The sinkless orientation problem can be encoded using two labels. So, let the set of labels be $\Sigma_\Pi=\{\I,\O\}$. If a node outputs label $\I$ in one the endpoint of one of the incident edges, it can be interpreted as that edge being incoming. Similarly, if the label is $\O$, that would indicate an outgoing edge. Hence, on the node side, we want that each node has the label $\O$ on at least one of its incident edges. On the edge side, we want each edge to be consistently oriented, hence if in one endpoint it has the label $\I$, in the other endpoint there must be the label $\O$, and vice versa. More precisely, our problem $\Pi$ is the following.

\begin{equation*}
\begin{aligned}
\nodeconst_\Pi&\text{:}\quad \O\s [\I\O]^{\Delta-1} \\
\edgeconst_\Pi&\text{:}\quad \I\s\O
\end{aligned}
\end{equation*}

Let $\Pi'=\re(\Pi)$. By definition, $\Sigma_{\Pi'}=2^{\Sigma_{\Pi}}=\{\{\I\}, \{\O\}, \{\I,\O\} \}$. Next we should define the edge constraint, where we want all configurations of the form $\S_1\s \S_2$ such that, for any choice in $\S_1$ and for any choice in $\S_2$ we obtain a configuration in $\edgeconst_{\Pi}$. Also, we want to eliminate all non-maximal configurations. Before going to that, for simplicity of the presentation, in order to avoid writing set of sets, let us rename the labels of $\Sigma_{\Pi'}$ in the following way: $\{\I\} \mapsto \bI$, $\{\O\} \mapsto \bO$, and $\{\I,\O\} \mapsto \bIO$. Now we can define the edge constraint. We must satisfy the universal quantification specified in the definition of $\edgeconst_{\Pi'}$, which means that we must forbid configurations that may result in $\I\s\I$ or $\O\s\O$, hence $\edgeconst_{\Pi'}: \bI\s\bO$. The node constrant must satisfy an existential and all configurations must not use labels that do not appear in $\edgeconst_{\Pi'}$. In other words, we want to be able to pick at least one $\O$, hence something like $[\bO\bIO]\s[\bI\bO\bIO]^{\Delta-1}$ would do, but since $\bIO$ does not appear in $\edgeconst_{\Pi'}$, we get $\nodeconst_{\Pi'}: \bO\s[\bI\bO]^{\Delta-1}$. So, problem $\Pi'=\re(\Pi)$ is the following.

\begin{equation*}
\begin{aligned}
\edgeconst_{\Pi'}&\text{:}\quad \bI\s\bO \\
\nodeconst_{\Pi'}&\text{:}\quad \bO\s[\bI\bO]^{\Delta-1} 
\end{aligned}
\end{equation*}

Now we are ready to define problem $\Pi''=\rere(\re(\Pi))$ which, by Theorem \ref{thm:sebastien}, we know that is exactly one round faster solvable than the sinkless orientation problem. By definition $\Sigma_{\Pi''}=2^{\Sigma_{\Pi'}}=\{\{\bI\}, \{\bO\}, \{\bI,\bO\}\}$. Again, in order to avoid writing set of sets, let us rename the labels of $\Sigma_{\Pi''}$ as follows: $\{\bO\} \mapsto \O$, $\{\bI\} \mapsto \I'$ (as in ``the non-maximal set that contains the $\bI$ label''), and $\{\bI,\bO\} \mapsto \I$. We must first define the node constraint, that must satisfy a universal quantifier. We want to avoid that there is the label $\bI$ in each of the $\Delta$ positions, since in that case the configuration $\bI^\Delta$ would be possible, but it is not allowed in $\nodeconst_{\Pi'}$. The configuration that satisfies this condition is $\O\s[\I\I'\O]^{\Delta -1}$, and after removing the non-maximal sets, we have $\nodeconst_{\Pi''}: \O\s\I^{\Delta -1}$. For the edge constraint we must satisfy an existential, hence on one side we can have all labels that contain $\bI$, while on the other all labels that contain $\bO$. The configuration that satisfies this is $[\I\O][\I'\I]$, but since $\I'$ is not used in the set $\nodeconst_{\Pi''}$, we have that $\edgeconst_{\Pi''}: \I[\I\O]$. Hence, the problem $\Pi''=\rere(\re(\Pi))$ that is exactly one round faster solvable that the sinkless orientation one is:

\begin{equation}
\begin{aligned}
\nodeconst_{\Pi''}&\text{:}\quad \O\s\I^{\Delta -1}  \\
\edgeconst_{\Pi''}&\text{:}\quad \I\s[\I\O] 
\end{aligned}
\end{equation}

\paragraph{Relations between labels}
For computing $\re(\Pi)$ or $\rere(\Pi)$, given some problem $\Pi$, it will be very useful to relate the labels used in $\Pi$ to each other according to their ``usefulness" in satisfying the edge constraint $\edgeconst_\Pi$ (resp.\ the node constraint $\nodeconst_{\Pi}$).

Let $\A$ and $\B$ be labels from $\Sigma_\Pi$ with the following property: for each edge configuration in $\edgeconst_\Pi$ containing $\A$, replacing one occurrence of $\A$ in that configuration by $\B$ again results in a configuration in $\edgeconst_\Pi$.
Then we say that $\B$ is \emph{at least as strong as} $\A$ \emph{according to $\edgeconst_\Pi$} and, equivalently that $\A$ is \emph{at least as weak as} $\B$ according to $\edgeconst_\Pi$.
We may omit the reference constraint if it is clear from context.
Moreover, if $\B$ is at least as strong as $\A$, but $\A$ is not at least as strong as $\B$, we say that $\B$ is \emph{stronger than} $\A$, and $\A$ is \emph{weaker than} $\B$. For example, consider the aforementioned problem $\Pi''=\rere(\re(\Pi))$ where $\Pi$ is sinkless orientation. Recall that the edge constraints are $\I[\I\O]$. We can say that label $\I$ is stronger than label $\O$ (and equivalently, label $\O$ is weaker than label $\I$), since, for each edge configuration, replacing one occurrence of $\O$ with $\I$ results in a configuration that is still allowed.
We also define the analogous notions for node constraints.

It is helpful to illustrate the strengths of labels via diagrams.
The \emph{edge diagram} of a problem $\Pi$ is a directed graph where the nodes are the labels in $\Sigma_\Pi$ and we have an edge from some label $\A$ to some label $\B$ if $\B \neq \A$, $\B$ is at least as strong as $\A$, and there exists no label $\Z \in \Sigma_\Pi$ such that $\B$ is stronger than $\Z$ and $\Z$ is stronger than $\A$, all according to $\edgeconst_\Pi$.
The latter condition simply ensures that we only illustrate ``irreducible" strength relations, i.e., none that can be decomposed into ``smaller" strength relations.
We define the \emph{node diagram} of a problem $\Pi$ analogously, by considering the strengths of labels according to $\nodeconst_\Pi$.
For examples of such diagrams, see Figures~\ref{fig:mis} and~\ref{fig:Re-mis}.

Note that the definition of strength implies that the diagrams do not contain directed cycles of length greater than $2$, and cycles of length $2$ appear exactly between all pairs of labels of equal strength.
In particular, if there are no pairs $(\A, \B)$ of labels such that $\A$ is stronger than $\B$ and vice versa, the respective diagram will be a directed acyclic (not necessarily connected) graph.

\paragraph{Additional Notation}
While the above notions are already known from \cite{Brandt2019, Balliu2019, trulytight}, we now introduce some useful new notation.
For a set $\{ \A_1, \dots, \A_p \} \subseteq \Sigma_\Pi$ of labels, we denote by $\gen{\A_1, \dots, \A_p}$ the set of all labels from $\Sigma_\Pi$ that are at least as strong as at least one of the $\A_i$.
In other words, we can read $\gen{\A_1, \dots, \A_p}$ off of the respective diagram by collecting each $\A_i$ together with all its successors.
Technically, for the definition of $\gen{}$, we need to specify whether the label strengths are considered w.r.t.\ $\nodeconst_\Pi$ or $\edgeconst_\Pi$.
However, whenever we consider $\gen{}$, the labels that $\gen{}$ takes as arguments will either come from (the alphabet of) a problem that we (are about to) apply the function $\re(\cdot)$ to, or a problem that we apply the function $\rere(\cdot)$ to.
In the former case, we will always consider $\gen{}$ w.r.t.\ the \emph{edge} constraint of the considered problem, and in the latter case w.r.t\ the \emph{node} constraint.
In particular, in the context of computing $\rere(\re(\Pi))$ for some problem $\Pi$, we will consider expressions such as $\gen{\gen{\A}}$ (where $\A \in \Sigma_\Pi$), which represents a set of sets of labels from $\Sigma_\Pi$; here the inner $\gen{}$ is taken w.r.t.\ the edge constraint of $\Pi$, and the outer $\gen{}$ w.r.t.\ the node constraint of $\re(\Pi)$.

\paragraph{Example} Let $\Pi$ be the maximal independent set problem, which we can express in the round elimination formalism as follows.

\begin{equation*}
\begin{aligned}
\nodeconst_\Pi&\text{:}\quad \M^{\Delta} \\
 &\text{ }\quad \P\s\U^{\Delta-1}\\
\edgeconst_\Pi&\text{:}\quad \M\s[\P\U]\\
 &\text{ }\quad \U\s\U
\end{aligned}
\end{equation*}

A node in the MIS outputs $\M^{\Delta}$. Otherwise, if a node is not in the MIS it must output $\P$ on one incident edge and $\U$ on all the others. The edge constraint implies that a node in the MIS can accept a pointer label $\P$ or label $\U$. Also, since $\P$ is only compatible with $\M$, each node not in the MIS can use label $\P$ only for pointing to a neighbor in the MIS. Moreover $\U\s\U$ is allowed since two nodes not in the MIS may be neighbors. The relations between the strengths of the labels in $\Sigma_\Pi$ is shown in the edge diagram of $\Pi$, given in Figure \ref{fig:mis}. Expressions in $\gen{}$ notation can be easily read from the edge diagram; for instance, we have $\gen{\M} = \{\M\}$, $\gen{\P} = \{\P, \U\}$, $\gen{\U} = \{\U\}$ $\gen{\M,\P} = \{\M, \P, \U\}$.

\begin{figure}[h]
	\centering
	\includegraphics[width=0.17\textwidth]{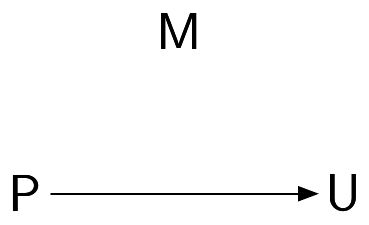}
	\caption{Relations between the strengths of the labels of the MIS problem: label $\U$ is stronger than $\P$, while there is no relation between label $\M$ and labels $\P$ or $\U$.}
	\label{fig:mis}
\end{figure}

Now, let $\Pi' = \re(\Pi)$, and consider the following mapping: $\{\U\} \mapsto \bU$, $\{\M\} \mapsto \bM$, $\{\M,\U\} \mapsto \bMU$, $\{\P,\U\} \mapsto \bPU$. The edge and node constraint of $\Pi'$ are as follows.

\begin{equation*}
\begin{aligned}
\edgeconst_{\Pi'}&\text{:}\quad \bU\s\bMU \\
&\text{ }\quad \bM\s\bPU\\
\nodeconst_{\Pi'}&\text{:}\quad [\bM\bMU]^\Delta\\
&\text{ }\quad \bPU\s[\bU\bMU\bPU]^{\Delta-1}
\end{aligned}
\end{equation*}
The node diagram of $\Pi'$, representing the relations of the strengths of the labels in $\Sigma_{\Pi'}$, is depicted in Figure \ref{fig:Re-mis}. Regarding the $\gen\gen{}$ notation, we obtain, for instance, that $\gen{\gen{\M}} = \gen{\bM} = \{\bM,\bMU\} = \{\{\M\},\{\M,\U\}\}$, $\gen{\gen{\U}} = \gen{\bU} = \{\bU,\bMU,\bPU\} = \{\{\U\},\{\M,\U\}, \{\P,\U\}\}$, $\gen{\gen{\M,\U}} = \gen{\bMU} = \{\bMU\} = \{\{\M,\U\}\}$, $\gen{\gen{\P}} = \gen{\bPU} = \{\bPU\} = \{\{\P,\U\}\}$.

\begin{figure}[h]
	\centering
	\includegraphics[width=0.17\textwidth]{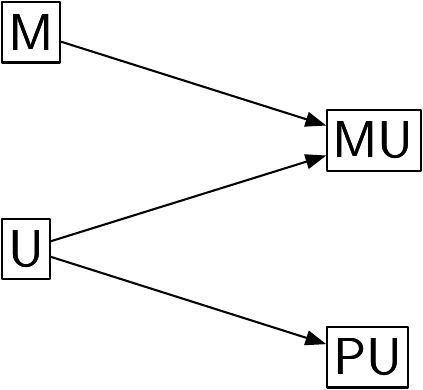}
	\caption{Relations between the strengths of the labels of problem $\re(\Pi)$, where $\Pi$ is the MIS problem; the diagram shows that label $\bMU$ is stronger than both labels $\bM$ and $\bU$, also label $\bPU$ is stronger than label $\bU$.}
	\label{fig:Re-mis}
\end{figure}

We call a set $S = \{ \A_1, \dots, \A_p \} \subseteq \Sigma_\Pi$ \emph{right-closed} if $S = \gen{\A_1, \dots, \A_p}$.
In other words, $S$ is right-closed if and only if for each label $\A_i$ contained in $S$ also all successors of $\A_i$ in the respective diagram are contained in $S$.
The definitions of $\re(\cdot)$ and $\rere(\cdot)$, in particular the removal of non-maximal configurations in the definitions, imply the following observation.

\begin{observation}\label{obs:rcs}
	Consider an arbitrary collection of labels $\A_1, \dots, \A_p \in \Sigma_\Pi$.
	If $\{ \A_1, \dots, \A_p \} \in \Sigma_{\re(\Pi)}$, then the set $\{ \A_1, \dots, \A_p \}$ is right-closed (w.r.t.\ $\edgeconst_\Pi$).
	If $\{ \A_1, \dots, \A_p \} \in \Sigma_{\rere(\Pi)}$, then the set $\{ \A_1, \dots, \A_p \}$ is right-closed (w.r.t.\ $\nodeconst_\Pi$).
\end{observation}
\begin{proof}
	For reasons of symmetry, we only need to prove the first statement.
	Let $\S = \{ \A_1, \dots, \A_p \}$ and assume that $\S \in \Sigma_{\re(\Pi)}$.
	Then there must be an edge configuration in $\edgeconst_{\re(\Pi)}$ containing $\S$, by the definition of $\re(\cdot)$.
	Consider an arbitrary label $\B \in \Sigma_\Pi$ that is at least as strong as at least one $\A_i$ w.r.t.\  $\edgeconst_\Pi$.
	By the definition of strength, and the definition of $\edgeconst_{\re(\Pi)}$ (or $\edgeconst_{\Pi'}$), adding label $\B$ to set $\S$ in the considered edge configuration results in a configuration that is still contained in $\edgeconst_{\re(\Pi)}$.
	Since $\edgeconst_{\re(\Pi)}$ does not contain any non-maximal configurations, this implies that $\B$ was already contained in $\S$, i.e., $\B = \A_i$ for some $i$.
	It follows that $\S$ is right-closed (w.r.t.\ $\edgeconst_\Pi$).
\end{proof}

Observation~\ref{obs:rcs} enables us to prove the following observation.

\begin{observation}\label{obs:subsetarrow}
	Let $\U, \W \in  \Sigma_{\re(\Pi)}$ be two sets satisfying $\U \subseteq \W$.
	Then $\W$ is at least as strong as $\U$ according to $\nodeconst_{\re(\Pi)}$.
	In particular, for any label $\A \in \Sigma_\Pi$ such that $\gen{\A} \in \Sigma_{\re(\Pi)}$, every set $\X \in \Sigma_{\re(\Pi)}$ containing $\A$ is contained in $\gen{\gen{\A}}$.
	
	Analogous statements hold for $\rere(\cdot)$ instead of $\re(\cdot)$.
\end{observation}
\begin{proof}
	For reasons of symmetry, we only need to prove the statements for $\re(\cdot)$.
	The definition of $\re(\cdot)$ immediately implies that replacing $\U$ by $\W$ in any configuration contained in $\nodeconst_{\re(\Pi)}$ results in a configuration that is also contained in $\nodeconst_{\re(\Pi)}$.
	Hence, $\W$ is at least as strong as $\U$ according to $\nodeconst_{\re(\Pi)}$.
	
	Now, let $\A$ and $\X$ be as described in the lemma.
	By Observation~\ref{obs:rcs}, the set $\X$ is right-closed w.r.t.\ $\edgeconst_{\Pi}$, which, by the definition of $\gen{}$, implies $\gen{\A} \subseteq \X$, since $\X$ contains $\A$.
	It follows that $\X$ is at least as strong as $\gen{\A}$ according to $\nodeconst_{\re(\Pi)}$, by the first part of Observation~\ref{obs:subsetarrow}.
	Hence, $\X \in \gen{\gen{\A}}$.	
\end{proof}

Moreover, for a set $S = \{ \A_1, \dots, \A_p \} \subseteq \Sigma_\Pi$ of labels, we denote by $\dis(S)$ the disjunction $[\A_1 \dots \A_p]$.
For instance, $\dis(\gen{\A})$ is the disjunction of all labels that are at least as strong as $\A$.

\paragraph{Generalizing to non-regular graphs}
As mentioned before, in this paper we will restrict attention to regular graphs.
Since we are proving lower bounds, this does not affect the generality of our results; however, for the upper bound we prove along the way, some additional step is required to lift the bound to general graphs.
In its full generality, the round elimination framework can also be applied to non-regular graphs, and the arguments in our upper bound would essentially remain the same; however, describing the framework formally is somewhat cumbersome.
Hence, we will choose a different route to show that our upper bound holds on general graphs: we will present a ``human-understandable" version of the algorithm obtained by round elimination for which it will be easy to check that its correctness is not affected by having nodes of different degrees.

\subsection{Roadmap}
We will start, in Section \ref{sec:problems}, by defining a family of problems $\Pi_{\Delta,\beta}(v,x)$, for which we will later show how it relates to the $(2,\beta)$-ruling set problem. The parameter $v = [v_0, \ldots, v_\beta]$ is a list of non-negative numbers, that can be interpreted as a number of colors. Intuitively, the problem $\Pi_{\Delta,\beta}(v,x)$ can be solved in $0$ rounds if we are given some vertex coloring with $\size(v) \coloneqq \sum_{i=0}^{\beta} v_i$ colors. The parameter $x$ is some relaxation parameter: we will allow nodes to violate edge constraints on at most $x$ of their incident edges.

In Sections \ref{sec:firstspeedup}, \ref{sec:ub}, and \ref{sec:lb}, we will use the round elimination theorem to relate problems of this family. In Section \ref{sec:firstspeedup}, we will compute the problem that we obtain by applying our operator $\re(\cdot)$ to $\Pi_{\Delta,\beta}(v,x)$.
In Section \ref{sec:ub} we will prove upper bounds for the $(2,\beta)$-ruling set problem. We will consider a subset of the problems of the family, that is, those where parameter $x$ is set to be $0$. We will first show that $\rere(\Pi'_{\Delta,\beta}(v,0))$ is at least as easy as some other problem of the family, that is $\Pi_{\Delta,\beta}(v',0)$, where $v'$ is the inclusive prefix sum of $v$ (i.e., $v'_i = \sum_{j\le i} v_j$). The round elimination theorem will imply that, given a solution for $\Pi_{\Delta,\beta}(v',0)$, we can obtain a solution for $\Pi_{\Delta,\beta}(v,0)$ in at most one round of communication. We will finally combine multiple steps of such reasoning to obtain upper bounds: we will show how parameter $v$ evolves over multiple steps. Crucially, a solution for $\Pi_{\Delta,\beta}([1,0,\ldots,0],0)$ will directly imply a solution for the $(2,\beta)$-ruling set problem, and by repeatedly applying the round elimination theorem we will obtain some problem $\Pi_{\Delta,\beta}(v',0)$ where $\size(v')$ is at least as large as the number of colors in the given vertex coloring. We will first prove an upper bound on the number of steps required to obtain such a problem, thereby giving an upper bound on the time complexity of the algorithm. Then, we will provide a human-understandable version of the round-elimination-generated algorithm, in order to argue that this algorithm does not only work on regular graphs, but on all graphs.

In Section \ref{sec:lb}, we will prove lower bounds for the $(2,\beta)$-ruling set problem. The main idea here will be to show that, by increasing parameter $x$, we can essentially relate the problems of the family in the same way as we do for the upper bounds. That is, we can get the same evolution of parameter $v$ as in the upper bound, at the price of increasing parameter $x$. Essentially, this will allow us to use the ideas obtained from the upper bound to get a lower bound. We will show in Section \ref{sec:liftlocal} how to lift the obtained lower bounds from the port numbering model to the LOCAL model.

\section{The problem family}\label{sec:problems}
\subsection{Problem definition}
In this section, we define a family of problems $\Pi_{\Delta,\beta}(v,x)$, that we will use to prove lower and upper bounds for the $(2,\beta)$-ruling set problem on graphs of maximum degree $\Delta$. The parameter $v = [v_0,\dotsc,v_\beta]$ is a list of non-negative integers, and the parameter $x$ satisfies $0\le x \le \Delta$ (while proving upper bounds, we will actually only consider the case where $x=0$). Intuitively, $v$ represents a list of color \emph{groups}, where each $v_i$ represents the number of colors in that group, while $x$ represents some relaxation parameter we will refer to as the number of wildcards. As we will see, if we start from a problem in this family, and we increase the value of $x$, or we increase the value of $v_i$ for some $i$, we will get a problem that is at least as easy as the one we started from. More precisely, given a solution for the starting problem, we can use it to solve the new problem in $0$ rounds of communication.

The high-level idea of the construction of the problem family is that we have \emph{colors} and \emph{pointers}, and nodes can either output a color (satisfying the usual constraints of the vertex coloring problem), or a pointer. Moreover, we have $\beta+1$ \emph{groups} called group $0$ to group $\beta$, and each color and each pointer belongs to exactly one of these groups. More precisely, there are exactly $v_i$ colors in group $i$, and there is exactly one pointer in each group except group $0$ (which contains no pointer). A pointer can only point to a node outputting a pointer, or a color, of a lower group. An example of a correct solution is given in Figure \ref{fig:problem-family}. Moreover, each node can label at most $x$ of its incident edges with a so-called \emph{wildcard}. If an edge is labeled with a wildcard by one of its endpoints, the resulting output label pair on the edge is correct by definition (i.e., it is an edge configuration listed in the edge constraint) regardless of the label the other endpoint outputs on the edge.

\begin{figure}[h]
	\centering
	\includegraphics[width=0.7\textwidth]{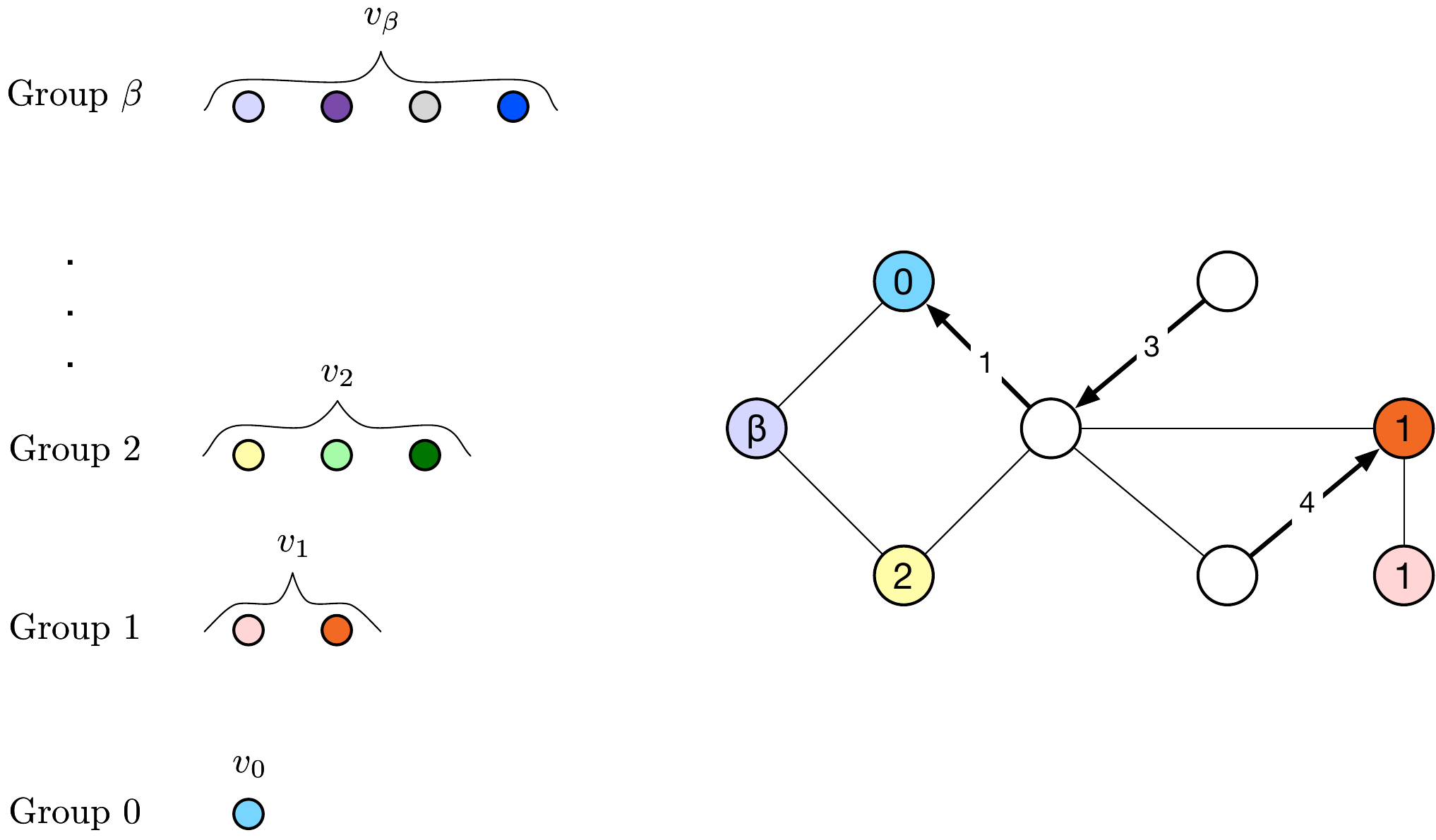}
	\caption{An example of a problem with parameters $[v_0, v_1, \dotsc, v_\beta]$ and $x=0$. Each node in the graph outputs either a color of some group, or a pointer pointing to a color or a pointer of a strictly smaller group. Neighboring nodes are not allowed to output the same color, but they can output colors belonging to the same group.}
	\label{fig:problem-family}
\end{figure}

The $(2,\beta)$-ruling set problem is the special case where we allow only $1$ color, i.e., $v_0=1$ and $v_i=0$ for all $i>0$, $\beta$ pointers, and no wildcards, i.e., $x=0$. In fact, the nodes in the ruling set will be exactly the nodes that output the color (note that since the ruling set nodes form an independent set, the coloring constraints are satisfied), and we allow the other nodes to point using pointers of different groups, depending on the distance they have from a node in the ruling set. We will later show that, while a solution for $\Pi_{\Delta,\beta}([1,0,\ldots,0],0)$ can be converted in $0$ rounds to a solution for the $(2,\beta)$-ruling set problem, we may need up to $\beta$ rounds to do the converse (and we will have to take this in consideration later when determining the actual lower bounds).

We define $\size(v) = v_0 + \dots + v_\beta$.
If we increase the number of colors in $\Pi_{\Delta,\beta}(v,x)$, i.e., if we increase $\size(v)$, the problem becomes easier: once we reach, for example, the case where $\size(v) = \Omega(\Delta^2)$, we have a problem that can be solved in $O(\log^* n)$ rounds in the LOCAL model, since in this model a graph can be colored in $O(\log^* n)$ rounds with $O(\Delta^2)$ colors \cite{Linial1992}. Also, by letting parameter $x$ grow we get an easier problem: in the extreme case of $x=\Delta$ we have a problem that is $0$-round solvable, since we can output wildcards everywhere.

\paragraph{Labels}We now formally define the set $\Sigma_{\Delta,\beta}(v,x)$ of labels of the problem $\Pi_{\Delta,\beta}(v,x)$. Let $\Sigma_{\Delta,\beta}(v,x) = \mathcal{P} \cup \mathcal{C}\cup \mathcal{X}$, where
\begin{itemize}
	\item $\mathcal{P}=\{ \A_i, \B_i ~|~ 1 \le i \le \beta \}$,
	\item $\mathcal{C}=\{ \C_{i,j} ~|~ 0 \le i \le \beta, 1 \le j \le v_i \}$,
	\item $\mathcal{X} = \{\X\}$ if $x>0$ and $\mathcal{X} = \{\}$ if $x=0$ (that is, if $x=0$ there is no label $\X$ in the set $\Sigma$).
\end{itemize}
These labels can be interpreted as follows:
\begin{itemize}
	\item The label $\X$ is a wildcard. Nodes write it on an edge to mark that edge as \emph{``don't care''}.
	\item The label $\A_i$ is a pointer, and the label $\B_i$ can be used to ``accept" pointers (of higher groups) that are output by neighboring nodes on connecting edges.
	\item The label $\C_{i,j}$ is the $j$-th color of group $i$.
\end{itemize}

\paragraph{Node constraint} We now define the node constraint $\nodeconst_{\Delta,\beta}(v,x)$, i.e., the set of allowed node configurations. The set $\nodeconst_{\Delta,\beta}(v,x)$ contains the following:
\begin{itemize}
	\item $\c^{\Delta-x} \s \X^x$, for each $\c \in \mathcal{C}$. That is, nodes output some color $\c \in \mathcal{C}$, marking $x$ incident edges as ``don't care''.
	\item $\A_i \s \B_i^{\Delta-1}$, for each $1 \le i \le \beta$. That is, nodes can output a pointer $\A_i$ on one incident edge. All other incident edges are marked as $\B_i$. We will see, when defining the edge constraint, that this will allow to accept pointers of higher groups. Intuitively, a node outputting this configuration must be at distance at most $i$ from a node outputting a color (of some group $< i$). 
\end{itemize}

\paragraph{Edge constraint}  We now define the edge constraint $\edgeconst_{\Delta,\beta}(v,x)$. It contains the following edge configurations:
\begin{itemize}
	\item $\C_{i,j} \s \C_{i',j'}$ if $(i,j) \neq(i',j')$, for each $1 \le i,i' \le \beta$, $1 \le j \le v_i$, $1 \le j' \le v_{i'}$. That is, all colors are compatible with all other colors (except themselves).
	\item $\B_i \s \B_j$, for each $1 \le i,j \le \beta$. That is, all $\B$ labels are compatible with all other $\B$ labels (including themselves).
	\item $\A_j \s \B_i$, for each $1 \le i < j \le \beta$. That is, pointers can point to non-colored nodes of lower groups.
	\item $\B_i \s \C_{i',j}$, for each $1 \le i,i' \le \beta$, $1 \le j \le v_{i'}$. That is, all $\B$ labels are compatible with all colors.
	\item $\A_i \s \C_{i',j}$, for each $1 \le i' < i \le \beta$, $1 \le j \le v_{i'}$. That is,  pointers can point to colored nodes of lower groups.
	\item $\X \s \L$, for each $\L \in \Sigma$, if $x > 0$. That is, the wildcard $\X$ is compatible with all labels.
\end{itemize}

\subsection{From the problem family to ruling sets, and vice versa}\label{sec:equivalence}
We now discuss the relation between $\Pi_{\Delta,\beta}([1,0,\dotsc,0],0)$ and the $(2,\beta)$-ruling set problem. We argue that a solution for $\Pi_{\Delta,\beta}([1,0,\dotsc,0],0)$ can be turned in $0$ rounds into a solution for the $(2,\beta)$-ruling set problem, and that a solution for the $(2,\beta)$-ruling set problem can be turned in $\beta$ rounds into a solution for $\Pi_{\Delta,\beta}([1,0,\dotsc,0],0)$. In Section \ref{sec:lb} we will use this relation to transform a lower bound of $T$ rounds for $\Pi_{\Delta,\beta}([1,0,\dotsc,0],0)$ into a lower bound of $(T-\beta)$ rounds for the $(2,\beta)$-ruling set problem.

Let us start by showing how to turn a solution for $(2,\beta)$-ruling set into a solution for the $\Pi_{\Delta,\beta}([1,0,\dotsc,0],0)$ problem. Given a solution for the $(2,\beta)$-ruling set problem, proceed as follows. Nodes in the ruling set output $\C_{0,1}$ on each incident edge. Each node $v$ can find in $\beta$ rounds the closest node of the ruling set (breaking ties arbitrarily); let this distance be $d_v$, satisfying $1 \le d_v \le \beta$. Node $v$ outputs $\A_{d_v}$ on the incident edge contained in the shortest path to this closest ruling set node, and $\B_{d_v}$ on all the other incident edges. The node constraint of $\Pi_{\Delta,\beta}([1,0,\dotsc,0],0)$ is clearly satisfied. Moreover, by construction, no neighboring nodes are outputting $\C_{0,1}$ on the same edge, and since other nodes use their distance to the closest ruling set node to output pointers, also the edge constraint is satisfied.

Consider now a solution for the problem $\Pi_{\Delta,\beta}([1,0,\dotsc,0],0)$. Nodes are either labeled with the color $\C_{0,1}$, or with one of the $\beta$ configurations that contain a pointer. We put exactly the colored nodes in the ruling set. Since the configuration $\C_{0,1} \s \C_{0,1}$ is not contained in the edge constraint, the colored nodes form an independent set. Also, since the constraints of $\Pi_{\Delta,\beta}([1,0,\dotsc,0],0)$ guarantee that each node that outputs $\A_i \s \B_i^{\Delta-1}$ has a colored neighbor, or a neighbor that outputs $\A_j \s \B_j^{\Delta-1}$ with $j < i$, nodes that are not in the independent set are at distance at most $\beta$ from a node in the independent set.

\subsection{The idea behind this problem family}
While the definition of the problem family $\Pi_{\Delta,\beta}(v,x)$ may seem arbitrary, we argue that there is a \emph{natural} way to obtain it, at least for the case $x=0$, that is the following:
\begin{itemize}
	\item Start from a problem of the family (at the beginning, this means to start from the $(2,\beta)$-ruling set problem).
	\item Apply the round elimination theorem.
	\item Note that in the obtained problem there are some allowed configurations that directly correspond to the original $\A_i \s \B_i^{\Delta-1}$ allowed configurations. Keep these configurations.
	\item Note that in the obtained problem there are some allowed configurations that directly correspond to a coloring problem (configurations of the form $\C^\Delta$ such that label $\C$ is compatible with all the labels of the configurations of the same form, except itself). Keep these configurations.
	\item Discard everything else.
\end{itemize}
Essentially what we need to do is to keep the part of the problem that has some \emph{intuitive} meaning (that is, colors and pointers), and discard everything else.

In the upper bound section, we will prove how the color groups evolve at each step. Intuitively, by applying the round elimination theorem to the problem $\Pi_{\Delta,\beta}(v,0)$ (and by discarding some allowed configurations, thus by making the problem harder), we obtain the problem $\Pi_{\Delta,\beta}(v',0)$, where $v'$ is the inclusive prefix sum list of $v$. For example, the $(2,2)$-ruling set problem is equivalent to $\Pi_{\Delta,2}([1,0,0],0)$, and by applying the round elimination theorem we get a problem that is not harder than $\Pi_{\Delta,2}([1,1,1],0)$, and by repeating the same procedure we get $\Pi_{\Delta,2}([1,2,3],0)$, and then we get $\Pi_{\Delta,2}([1,3,6],0)$, and so on. This gives a quadratic growth in the number of colors, and we thus get an algorithm that, given some $c$ coloring can solve the $(2,2)$-ruling set problem in $O(\sqrt{c})$ rounds. By generalizing the same reasoning to $(2,\beta)$-ruling sets, we get an algorithm that, given a $c$ coloring, solves the problem in $O(\beta c^{1/\beta})$ rounds, matching the current state-of-the-art algorithm w.r.t.\ dependency on $\Delta$, $c$ and $\beta$ \cite{SEW13}. While an algorithm obtained in the specific round elimination framework we use only works on regular graphs, we will show that the algorithm that we obtain actually works in any graph.

In the lower bound section, we will show that, by increasing parameter $x$ at each step, we can prove that the color groups evolve in the same way as in the upper bound, and that we can thus prove a \emph{lower bound} using a problem family suggested by the upper bound. In particular, we will show that all the non-intuitive allowed configurations can be relaxed to the intuitive ones, if we allow some slack on them.

\subsection{The edge diagram}\label{sec:edgediag}
We now show the structure of the edge diagram of our $\Pi_{\Delta,\beta}(v,x)$ problems. Knowing such structure will be helpful in the following sections. In particular, as previously discussed in Section \ref{sec:preliminaries}, when defining $\re(\Pi_{\Delta,\beta}(v,x))$ we will only have to consider right-closed subsets of labels with regard to this diagram (see Observation \ref{obs:rcs}). The following relations between the labels of $\Pi_{\Delta,\beta}(v,x)$ derive directly from the definition of $\edgeconst_{\Delta,\beta}(v,x)$.
\begin{itemize}
	\item $\B_i \le \B_j$, if $j < i$.
	\item $\A_i \le \A_j$, if $i < j$.
	\item $\A_i \le \C_{i',j}$, if $1 \le i \le i'$.
	\item $\C_{i',j} \le \B_i$, if $1 \le i \le i'$.
	\item $\L \le \X$ for all $\L \in \Sigma$. 
\end{itemize}
Notice that this also implies $\A_i \le \B_j$ for all $1 \le i,j \le \beta$. An example of the diagram for $\Pi_{\Delta,3}([1,2,3,4],1)$ is shown in Figure \ref{fig:diag}. 

\begin{figure}
	\centering
	\includegraphics[width=0.7\textwidth]{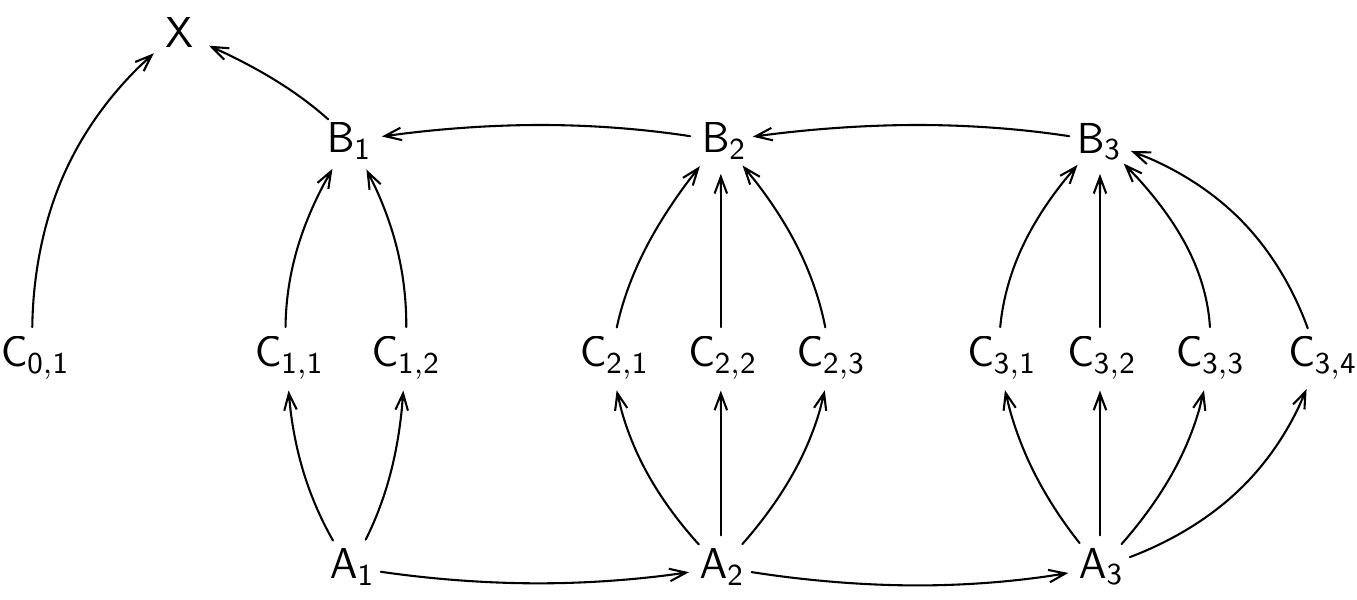}
	\caption{Edge diagram of $\Pi_{\Delta,3}([1,2,3,4],1)$}
	\label{fig:diag}
\end{figure}

\section{The intermediate problems}\label{sec:firstspeedup}
In Section \ref{sec:equivalence} we formally introduced a family of problems $\Pi_{\Delta,\beta}(v,x)$ by defining $\Sigma_{\Delta,\beta}(v,x)$, $\nodeconst_{\Delta,\beta}(v,x)$ and $\edgeconst_{\Delta,\beta}(v,x)$. In this section, we compute $\re(\Pi_{\Delta,\beta}(v,x))$, i.e., we compute the family of problems that we get by applying the function $\re(\cdot)$ to $\Pi_{\Delta,\beta}(v,x)$. In other words, we will compute the set of labels $\Sigma'_{\Delta,\beta}(v,x)$, the node constraint $\nodeconst'_{\Delta,\beta}(v,x)$ and the edge constraint $\edgeconst'_{\Delta,\beta}(v,x)$ of $\re(\Pi_{\Delta,\beta}(v,x))$. In this section, we will always assume that $x < \Delta$ (or, where indicated, even $x \leq \Delta - 2$).

\paragraph{Labels} By the definition of $\re(\cdot)$, the set of labels of $\re(\Pi_{\Delta,\beta}(v,x))$ is the set of non-empty subsets of the set $\Sigma_{\Delta,\beta}(v,x)$, that is, $\Sigma'_{\Delta,\beta}(v,x)=2^{\Sigma_{\Delta,\beta}(v,x)}$. In other words, $\Sigma'_{\Delta,\beta}(v,x)=2^{ \mathcal{P}\s \cup\s \mathcal{C}\s\cup\s \mathcal{X}}$, where
\begin{itemize}
	\item $\mathcal{P}=\{ \A_i, \B_i ~|~ 1 \le i \le \beta \}$ is the set of pointers,
	\item $\mathcal{C}=\{ \C_{i,j} ~|~ 0 \le i \le \beta, 1 \le j \le v_i \}$ is the set of colors, and
	\item $\mathcal{X} = \{\X\}$ if $x>0$, and $\mathcal{X} = \{\}$ if $x=0$, is the set of wildcards.
\end{itemize}

\subsection{Edge constraint}\label{sec:edgeconstraint}
Given a set of colors $\mathscr{C} \subseteq \mathcal{C}$, let $g(\mathscr{C})$ be the largest index $k$ such that $\C_{k,\ell} \in \mathscr{C}$ for some $\ell$ (if $\mathscr{C}$ is empty, let $g(\mathscr{C}) = -1$). In other words, $g(\mathscr{C})$ is the highest group of all colors contained in $\mathscr{C}$.
Consider all the possible pairs $(\mathscr{C},i)$ where $\mathscr{C} \subseteq \mathcal{C}$ and $0 \le i \le \beta$. A pair $(\mathscr{C},i)$ is \emph{good} if and only if $i \ge g(\mathscr{C})$. If $x=0$, we additionally require that, in order for a pair to be good, it must be different from $(\emptyset,0)$. Essentially, good pairs represent all ways to combine subsets of colors and group indices such that the index is at least as large as the highest color group appearing in the set. Let $\S_1((\mathscr{C},i)) = \mathscr{C} \cup \{\B_j ~|~ 1 \le j \le i\}\} \cup \mathcal{X}$. Let $\S_2((\mathscr{C},i)) = (\mathcal{C} \setminus \mathscr{C}) \cup \{\B_j ~|~ 1 \le j \le \beta\} \cup \{\A_j ~|~ i < j \le \beta\} \cup \mathcal{X}$.

\begin{lemma}\label{lemma:edgeconst'}
	The edge constraint of $\re(\Pi_{\Delta,\beta}(v,x))$ is \[\edgeconst'_{\Delta,\beta}(v,x) = \{ \S_1((\mathscr{C},i)) \s \S_2((\mathscr{C},i)) ~|~ (\mathscr{C},i) \text{ is a good pair}\}.\]
\end{lemma}
\begin{proof}
Let us start with some observations. First of all, recall that the $\re(\cdot)$ operator requires $\edgeconst'_{\Delta,\beta}(v,x)$ to contain all (and only) pairs of sets $\S_1 \s \S_2$ that satisfy that for all $\labels_1 \in \S_1$ and for all $\labels_2 \in \S_2$, $\labels_1 \s \labels_2$ is in $\edgeconst_{\Delta,\beta}(v,x)$. Also, recall that we can discard all \emph{non-maximal} pairs, and that pairs are equivalent up to reordering. Moreover, recall that we do not need to consider all possible subsets in $ \Sigma'_{\Delta,\beta}(v,x)$, but only \emph{right-closed subsets} with respect to the edge diagram of $\Pi_{\Delta,\beta}(v,x)$ (see Observation \ref{obs:rcs}). Essentially, in $\edgeconst'_{\Delta,\beta}(v,x)$ we must have all possible pairs $\S_1 \s \S_2$, where $\S_1$ is a right-closed subset, and $\S_2$ is the intersection of all sets of labels compatible with each $\labels_1 \in \S_1$ (note that also the resulting set $\S_2$ must be right-closed). 

We consider all cases where $\S_1$ does not contain any label $\A_i$. Since no $\A$-type label is compatible (with regard to $\edgeconst_{\Delta,\beta}(v,x)$) with any other $\A$-type label, it is not possible to have some configuration $\S_1 \s \S_2$ that contains $\A_i \in \S_1$ and $\A_j \in \S_2$, for any $i$ and $j$. Thus, for all valid configurations, either $\S_1$ or $\S_2$ does not contain any $\A$-type label, and this implies that by only considering the case where $\S_1$ does not contain any $\A$-type label, we cover all cases (up to symmetry). By the definition of the edge diagram of $\Pi_{\Delta,\beta}(v,x)$, right-closed subsets $\S$ of $\Sigma'_{\Delta,\beta}(v,x)$ that do not contain any label $\A_i$ are of the following form: we have a subset of colors, and if a color of group $i$ is in $\S$, all $\B_j$ satisfying $j \le i$ are also present. Also, additional $\B_j$ may be in $\S$, and if $\B_j$ is present, all $\B_{j'}$ satisfying $j' \le j$ must also be there. Finally, the label $\X$ is also present, if $x\neq 0$. Notice that there is a one-to-one correspondence between all good pairs and all right-closed subsets not containing any label $\A_i$ (the case distinction on the value of $x$ ensures that we are not considering the empty set). In fact, since $i \ge g(\mathscr{C})$, when creating a set we put at least all $\B$-type labels with index between $1$ and the maximum color group appearing in $\mathscr{C}$, and by increasing $i$ we put additional $\B$-type labels.

For each good pair $(\mathscr{C},i)$, we add the configuration $\S_1((\mathscr{C},i)) \s \S_2((\mathscr{C},i))$ to $\edgeconst'_{\Delta,\beta}(v,x)$. We need to prove that $\S_2 \coloneqq \S_2((\mathscr{C},i))$ contains all and only the labels that are edge compatible with all the labels in $\S_1 \coloneqq \S_1((\mathscr{C},i))$. First, note that a color cannot appear in both $\S_1$ and $\S_2$, since a color is not compatible with itself in $\edgeconst_{\Delta,\beta}(v,x)$. Hence, since $\S_2$ contains $\mathcal{C} \setminus \mathscr{C}$, colors added to $\S_2$ are all valid, and no color can be added. Then, all $\B_i$ are present in $\S_2$, thus, trivially, we cannot add more $\B$-type labels to $\S_2$, and since each $\B$-type label is edge compatible with all other $\B$-type labels and with all colors, the configurations in $\edgeconst_{\Delta,\beta}(v,x)$ are not violated. The same holds for the label $\X$, that we add to $\S_2$ if present in $\Sigma'_{\Delta,\beta}(v,x)$. Note that $\X$ is compatible with any label, so the configurations in $\edgeconst_{\Delta,\beta}(v,x)$ are trivially not violated. The last remaining case to analyze is the $\A$-type labels: if labels $\{\B_1, \ldots, \B_i\}$ are present in $\S_1$ we added $\{\A_{i+1},\ldots,\A_\beta\}$ to $\S_2$. Since $\A_i$ is not compatible with $\B_i$ with regard to $\edgeconst_{\Delta,\beta}(v,x)$, this implies that we cannot add more $\A$-type labels to $\S_2$. Also, note that the presence of a color of group $j$ in $\S_1$ implies that $i \ge j$, and thus the presence of $\B_j$, and since $\A_i$ is edge compatible with all colors of groups strictly less than $i$, the $\A$-type labels added to $\S_2$ do not violate the $\edgeconst_{\Delta,\beta}(v,x)$ configurations.	
\end{proof}

\subsection{Properties}
Before computing the node constraint of $\re(\Pi_{\Delta, \beta}(v,x))$ in Section~\ref{sec:noco}, we will first collect two facts about problem $\re( \Pi_{\Delta,\beta}(v,x) )$ that we can derive from the description of the edge constraint in Section~\ref{sec:edgeconstraint} and will be useful later.

\begin{lemma}\label{lem:restrong}
	Consider two sets $\U, \W \in \Sigma'_{\Delta,\beta}(v,x)$, and assume that $x + 2 \leq \Delta$.
	Then $\W$ is at least as strong as $\U$ according to $\nodeconst'_{\Delta,\beta}(v,x)$ if and only if $\U \subseteq \W$.
\end{lemma}
\begin{proof}
	If $\U \subseteq \W$, then $\W$ is at least as strong as $\U$ according to $\nodeconst'_{\Delta,\beta}(v,x)$, by Observation~\ref{obs:subsetarrow}.
	For the other direction, assume that $\U \nsubseteq \W$, and let $\u \in \Sigma_{\Delta,\beta}(v,x)$ be a label contained in $\U \setminus \W$.
	We want to show that $\W$ is not at least as strong as $\U$ according to $\nodeconst'_{\Delta,\beta}(v,x)$.
	For a contradiction assume that $\W$ is at least as strong as $\U$.
	Consider some configuration $\u \s \y_2 \s \y_3 \s \dots \s \y_\Delta \in \nodeconst_{\Delta,\beta}(v,x)$.
	We first show that the configuration $\U \s \gen{\y_2} \s \gen{\y_3} \s \dots \s \gen{\y_\Delta}$ is contained in $\nodeconst'_{\Delta,\beta}(v,x)$.
	Recalling the definition of $\re(\cdot)$, we see that, since $\y_j \in \gen{\y_j}$, for all $2 \leq j \leq \Delta$, and $\u \in \U$, the only case in which the configuration might not be contained in $\nodeconst'_{\Delta,\beta}(v,x)$ is that one of the labels in the configuration is not contained in any configuration in $\edgeconst'_{\Delta,\beta}(v,x)$.
	Hence, for our first step it suffices to show that each of the $\gen{\y_j}$, and also $U$, is contained in some configuration in $\edgeconst'_{\Delta,\beta}(v,x)$.
	As $\U \in \Sigma'_{\Delta,\beta}(v,x)$ by definition, $\U$ is contained in such a configuration.
	The analogous statement for the $\gen{\y_j}$ follows from Lemma~\ref{lemma:edgeconst'} and the fact that for any label $\L \in \Sigma_{\Delta,\beta}(v,x)$, there exists some good pair $(\mathscr{C},i)$ such that $\gen{\L} = \S_1((\mathscr{C},i))$ or $\gen{\L} = \S_2((\mathscr{C},i))$.
	To see the latter, observe that
	\begin{align*}
		\gen{\A_i} &= \S_2((\mathscr{C},i-1)) &\text{ where $\mathscr{C} = \{ \C_{k,j} \mid 0 \leq k \leq i-1, 1 \leq j \leq v_k\}$,}\\
		\gen{\B_i} &= \S_1((\mathscr{C},i)) &\text{ where $\mathscr{C} = \emptyset$,}\\
		\gen{\C_{i,j}} &= \S_1((\mathscr{C},i)) &\text{ where $\mathscr{C} = \{ \C_{i,j} \}$, and}\\
		\gen{\X} &= \S_1((\emptyset,0)) &\text{ if $x > 0$.}
	\end{align*}
	It follows that the configuration $\U \s \gen{\y_2} \s \gen{\y_3} \s \dots \s \gen{\y_\Delta}$ is contained in $\nodeconst'_{\Delta,\beta}(v,x)$.

	Since $\W$ is at least as strong as $\U$, we obtain that also $\W \s \gen{\y_2} \s \gen{\y_3} \s \dots \s \gen{\y_\Delta} \in \nodeconst'_{\Delta,\beta}(v,x)$.
	By the definition of $\nodeconst'_{\Delta,\beta}(v,x)$, there is a configuration $\w \s \y'_2 \s \y'_3 \s \dots \s \y'_\Delta \in \nodeconst_{\Delta,\beta}(v,x)$ such that  $(\w, \y'_2, \y'_3, \dots, \y'_\Delta) \in \W \times \gen{\y_2} \times \gen{\y_3} \times \dots \times \gen{\y_\Delta}$.
	Since $\W$ is right-closed by Observation~\ref{obs:rcs}, the fact that $\u$ is not contained in $\W$ implies that also any label that is at least as weak as $\u$ according to $\edgeconst_{\Delta,\beta}(v,x)$ is not contained in $\W$; thus, $\w$ is not at least as weak as $\u$ according to $\edgeconst_{\Delta,\beta}(v,x)$.
	Moreover, as $\y'_j \in \gen{\y_j}$ for any $2 \leq j \leq \Delta$, we obtain the following picture:
	there are two configurations $\mathcal U = \u \s \y_2 \s \y_3 \s \dots \s \y_\Delta$ and $\mathcal W = \w \s \y'_2 \s \y'_3 \s \dots \s \y'_\Delta$ in $\nodeconst_{\Delta,\beta}(v,x)$ such that $\w$ is not at least as weak as $\u$, and $\y'_j$ is at least as strong as $\y_j$, for all $2 \leq j \leq \Delta$.
	Now it is straightforward to check that there are no two configurations in $\nodeconst_{\Delta,\beta}(v,x)$ with these properties, by going through all possible pairs of configurations.
	To this end, recall the strength relations of the labels in $\Sigma_{\Delta,\beta}(v,x)$ given in Section~\ref{sec:edgediag} (in particular, Figure~\ref{fig:diag}), and assume for a contradiction that such configurations $\mathcal U, \mathcal W$ exist (recall that two configurations that are identical up to reordering of the contained labels are considered as the same configuration).

	Consider first the case that $\mathcal W = \C_{i,j}^{\Delta-x} \s \X^x$ for some $0 \leq i \leq \beta$, $1 \leq j \leq v_i$, and $0 \leq x \leq \Delta - 2$.
	Since $\Delta-x \geq 2$, there is some index $2 \leq k \leq \Delta$ such that $\y'_k = \C_{i,j}$, which implies that $\mathcal U = \mathcal W$ or $\mathcal U = \A_\ell \s \B_\ell^{\Delta-1}$ for some $1 \leq \ell \leq i$, as otherwise $\y'_k$ cannot be at least as strong as $\y_k$.
	If $\mathcal U = \mathcal W$, then we have $\u = \C_{i,j}$ and $\w = \X$, as otherwise $\w$ is at least as weak as $\u$.
	It follows that there is some index $2 \leq k' \leq \Delta$ such that $\y_{k'} = \X$ and $\y'_{k'} =  \C_{i,j}$, yielding a contradiction to the fact that $\y'_{k'}$ is at least as strong as $\y_{k'}$.
	If $\mathcal U = \A_\ell \s \B_\ell^{\Delta-1}$ for some $1 \leq \ell \leq i$, then we have $\u = \B_\ell$ and $\w = \C_{i,j}$, or there is some index $2 \leq k' \leq \Delta$ such that $\y_{k'} = \B_\ell$ and $\y'_{k'} =  \C_{i,j}$, since $\Delta-x \geq 2$.
	In both cases, we obtain a contradiction, since $\C_{i,j}$ is weaker than $\B_\ell$.

	Now, consider the other case, i.e., that $\mathcal W = \A_\ell \s \B_\ell^{\Delta-1}$ for some $1 \leq \ell \leq \beta$.
	If $\mathcal U = \C_{i,j}^{\Delta-x} \s \X^x$ for some $0 \leq i \leq \beta$, $1 \leq j \leq v_i$, and $0 \leq x \leq \Delta - 2$, then we have $i < \ell$, as otherwise $\A_\ell$ is weaker than any label in $\mathcal U$, which would yield a contradiction no matter whether $\w = \A_\ell$ or $\y'_k = \A_\ell$ for some $2 \leq k \leq \Delta$.
	But since $\Delta-x \geq 2$ implies that there is some $2 \leq k \leq \Delta$ such that $\y_k = \C_{i,j}$ and $\y'_k = \B_\ell$, the case $i < \ell$ also yields a contradiction, as $\B_\ell$ is not at least as strong as $\C_{i,j}$ if $i < \ell$.
	If $\mathcal U = \A_{\ell'} \s \B_{\ell'}^{\Delta-1}$ for some $1 \leq \ell' \leq \beta$, then we see that $\ell' \leq \ell$, as otherwise, again, $\A_\ell$ is weaker than any label in $\mathcal U$.
	If $\ell' < \ell$, then $\B_{\ell'}$, which is contained in $\mathcal U$, is stronger than any label in $\mathcal W$, leading to a contradiction no matter whether $\u = \B_{\ell'}$ or $\y_k = \B_{\ell'}$ for some $2 \leq k \leq \Delta$.
	If $\ell' = \ell$, then we have $\u = \A_\ell$ and $\w = \B_\ell$, as otherwise $\w$ is at least as weak as $\u$.
	But then it follows that there is some index $2 \leq k \leq \Delta$ such that $\y_{k} = \B_\ell$ and $\y'_{k} =  \A_\ell$, yielding a contradiction to the fact that $\y'_k$ is at least as strong as $\y_k$.
\end{proof}

\begin{corollary}\label{cor:strongflip}
	Let $\U_1, \U_2 \in \Sigma_{\Delta,\beta}(v,x)$ be two labels such that $\U_2$ is stronger than $\U_1$ according to $\edgeconst_{\Delta,\beta}(v,x)$, and assume that $x + 2 \leq \Delta$.
	Then $\gen{\U_1}$ is stronger than $\gen{\U_2}$ according to $\nodeconst'_{\Delta,\beta}(v,x)$.
\end{corollary}
\begin{proof}
	Recall that, by definition, $\gen{\U_1}$ and $\gen{\U_2}$ contain exactly those labels from $\Sigma_{\Delta,\beta}(v,x)$ that are at least as strong as $\U_1$ and $\U_2$, respectively.
	Hence, the fact that $\U_2$ is stronger than $\U_1$ implies that $\gen{\U_2} \subseteq \gen{\U_1}$ and $\gen{\U_1} \nsubseteq \gen{\U_2}$.
	Now, applying Lemma~\ref{lem:restrong} yields the corollary.
\end{proof}

\subsection{Node constraint}\label{sec:noco}
We now compute the node constraint of $\re(\Pi_{\Delta, \beta}(v,x))$.
\begin{lemma}\label{lemma:nodeconst'}
	Let $x + 2 \leq \Delta$. The node constraint $\nodeconst'_{\Delta,\beta}(v,x)$ of $\re(\Pi_{\Delta, \beta}(v,x))$ is the collection of the following (condensed) configurations:
\begin{itemize}
	\item For each color $\C_{i,j}\in\mathcal{C}$, where $0\le i\le \beta$ and $1\le  j \le v_i$, 
	\[
	\dis(\gen{\gen{\C_{i,j}}})^{\Delta-x} \s \dis(\gen{\gen{\X}})^{x} \enspace.
	\]
	
	\item For each $1\le i\le \beta$ 
	\[
	\dis(\gen{\gen{\A_i}})\s \dis(\gen{\gen{\B_i}})^{\Delta-1} \enspace.
	\]
\end{itemize}
\end{lemma}

\begin{proof}
First of all, note that the definition of $\gen{\gen{L}}$, for some label $L$ of $\Pi_{\Delta, \beta}(v,x)$, depends on the strength of the labels of $\re(\Pi_{\Delta, \beta}(v,x))$, which in turn depends on the node constraint $\nodeconst'_{\Delta,\beta}(v,x)$, which we are currently defining by using the $\gen{\gen{\cdot}}$ notation. Notice that such a recursive definition is not an issue: by Lemma \ref{lemma:edgeconst'} we know what are the labels of $\re(\Pi_{\Delta, \beta}(v,x))$, and by Lemma \ref{lem:restrong} we know that the strength relation of these labels is given exactly by set inclusion. Hence we already know enough about the strength of the labels of $\re(\Pi_{\Delta, \beta}(v,x))$ even before formally defining $\nodeconst'_{\Delta,\beta}(v,x)$, and this allows us to use the $\gen{\gen{\cdot}}$ notation to define them.

The above lemma says that $\nodeconst'_{\Delta,\beta}(v,x)$ is given by the union, over all configurations $\L_1\s\dotsc\s\L_\Delta \in\s \nodeconst_{\Delta,\beta}(v,x)$, of the configurations $\dis(\gen{\gen{\L_1}})\s \ldots\s \dis(\gen{\gen{\L_\Delta}})$. Recall that the $\re(\cdot)$ operator requires that $\nodeconst'_{\Delta,\beta}(v,x)$ contains all (and only) configurations $\S_1 \s \ldots \s \S_\Delta$ that satisfy that there exists a choice $(\labels_1, \ldots, \labels_\Delta) \in \S_1 \times \ldots \times \S_\Delta$ such that $\labels_1\s \ldots\s \labels_\Delta$ is in $\nodeconst_{\Delta,\beta}(v,x)$. Also, recall that tuples are equivalent up to reordering. We argue that $\nodeconst'_{\Delta,\beta}(v,x)$ can be obtained as follows. Start from $\nodeconst' = \{\}$. For each configuration $\L_1\s\ldots\s\L_\Delta \in \nodeconst_{\Delta,\beta}(v,x)$ add to $\nodeconst'$  all the configurations that can be obtained from the condensed configuration $\dis(\gen{\gen{\L_1}})\s \ldots\s \dis(\gen{\gen{\L_\Delta}})$. We now prove that the obtained set  $\nodeconst'$ is equivalent to $\nodeconst'_{\Delta,\beta}(v,x)$. 

By Lemma \ref{lem:restrong}, and by definition of $\gen{}$, $\gen{\gen{\L_i}}$ contains all and only the sets containing $\L_i$, hence, $\nodeconst'$ satisfies the requirements of the existential quantifier. We now prove that $\nodeconst'$ is maximal, in the sense that we cannot add any new valid configuration $\S'_1\s \ldots\s \S'_\Delta$ to $\nodeconst'$. Assume for a contradiction that $\S'_1\s \ldots\s \S'_\Delta$ is a valid maximal configuration not contained in $\nodeconst'$. There must exist a choice $(\labels'_1, \ldots, \labels'_\Delta) \in \S'_1 \times \ldots \times \S'_\Delta$ such that $\labels'_1 \s \ldots\s \labels'_\Delta$ is in $\nodeconst_{\Delta,\beta}(v,x)$.  Note that, by construction, $\nodeconst'$ contains  $\dis(\gen{\gen{\labels'_1}}) \s \ldots\s \dis(\gen{\gen{\labels'_\Delta}})$, and by definition of $\gen{\gen{\labels'_i}}$ and Observation \ref{obs:subsetarrow} we have that $\S'_i \in \gen{\gen{\labels'_i}}$, for all $1 \le i \le \Delta$. Hence, $\S'_1\s \ldots\s \S'_\Delta$ is present in $\nodeconst'$, contradicting the assumption. Hence the constructed set $\nodeconst'$ is equal to $\nodeconst'_{\Delta,\beta}(v,x)$.
\end{proof}

\section{Upper bound}\label{sec:ub}
In this section we prove upper bounds for the $(2,\beta)$-ruling set problem. While an upper bound is not necessary to prove the main results of our work, perhaps surprisingly, it will serve the purpose of giving some intuition behind the definition of the problem family that we use to prove lower bounds.
We will first prove that
 $\Pi_{\Delta,\beta}(v',0)$ is at least as hard as $\rere(\re(\Pi_{\Delta,\beta}(v,0)))$, where $v'$ is the inclusive prefix sum list of $v$ (i.e., $v'_i = \sum_{j\le i} v_j$). That is, we can apply the round elimination theorem on $\Pi_{\Delta,\beta}(v,0)$ to get a problem that can be solved in (at most) $1$ round given a solution for $\Pi_{\Delta,\beta}(v',0)$. Hence, we will prove the following lemma.
 \begin{lemma}\label{lem:ub_secondspeedup}
 	The problem $\rere(\re(\Pi_{\Delta,\beta}(v,0)))$ can be solved in $0$ rounds given a solution for $\Pi_{\Delta,\beta}(v',0)$, where $v'$ is the inclusive prefix sum list of $v$. Hence, given a solution for $\Pi_{\Delta,\beta}(v',0)$ we can solve $\Pi_{\Delta,\beta}(v,0)$ in at most $1$ round.
 \end{lemma}

We will then analyze the whole problem family in order to provide an upper bound for the $(2,\beta)$-ruling set problem. In particular, we will analyze how the number of colors evolves over time, and we will prove that the time required to compute a $(2,\beta)$-ruling set is at most the minimum $t$ such that ${\beta+t \choose \beta} \ge c$, if nodes are initially labeled with some $c$-vertex coloring. In particular, this implies that a $(2,\beta)$-ruling set can be found in $O(\beta\, c^{1/\beta})$ rounds, and that a $(2,\beta\, c^{1/\beta})$-ruling set can be found in $\beta$ rounds, for all $\beta \le c$. While this upper bound does match but not improve the current state of the art, we will later show how this family, essentially obtained while proving upper bounds, can be turned into a lower bound by increasing parameter $x$ (recall the definition of the problem family $\Pi_{\Delta,\beta}(v,x)$ in Section \ref{sec:problems}). Hence, we will prove the following lemma.
\begin{lemma}\label{lem:ub_time}
	The time required to solve the $(2,\beta)$-ruling set problem in the port numbering model given a $c$-vertex coloring is at most the minimum $t$ such that ${\beta + t \choose \beta} \ge c$. In particular, the $(2,\beta)$-ruling set problem can be solved in $t \le \beta \, c^{1/\beta}$ rounds. Also, the $(2, \beta \, c^{1/\beta})$-ruling set problem can be solved in at most $\beta$ rounds.
\end{lemma}

Interestingly, the strategy that we will use to prove that $\Pi_{\Delta,\beta}(v',0)$ can be used to solve $\rere(\re((\Pi_{\Delta,\beta}(v,0)))$ shows that, by just blindly applying the round elimination theorem and discarding everything that has no \emph{intuitive} meaning, we can obtain algorithms that are able to compete with the current state of the art.

\subsection{Proof of Lemma \ref{lem:ub_secondspeedup}}
For simplicity, let us define $\Pi'_{\Delta,\beta}(v,0) \coloneqq \re(\Pi_{\Delta,\beta}(v,0))$ and $\Pi''_{\Delta,\beta}(v,0) \coloneqq \rere(\re(\Pi_{\Delta,\beta}(v,0))) =\rere(\Pi'_{\Delta,\beta}(v,0))$. We want to understand $\Pi''_{\Delta,\beta}(v,0)$, but it seems highly non trivial to show the exact form for $\Pi''_{\Delta,\beta}(v,0)$ using the round elimination technique. Instead, starting from $\Pi'_{\Delta,\beta}(v,0)$, we prove that some specific configurations are present in the node constraint $\nodeconst''_{\Delta,\beta}(v,0)$ of $\Pi''_{\Delta,\beta}(v,0)$. This is enough for our purposes, since, even if there are more configurations that we do not consider, it means that we are only making the problem harder. We will then show that we can rename labels appearing in $\Pi''_{\Delta,\beta}(v,0)$ such that the collection of configurations that we consider matches the node constraint definition of $\Pi_{\Delta,\beta}(v',0)$. We will also show that the edge constraint $\edgeconst''_{\Delta,\beta}(v,0)$ of $\Pi''_{\Delta,\beta}(v,0)$ matches the edge constraint $\edgeconst_{\Delta,\beta}(v,0)$ of $\Pi_{\Delta,\beta}(v,0)$. This will imply that, by applying Theorem \ref{thm:sebastien} on $\Pi_{\Delta,\beta}(v,0)$, we get a problem where we can discard some allowed configurations, and rename the obtained sets of sets, such that we get problem $\Pi_{\Delta,\beta}(v',0)$, and thus that a solution for $\Pi_{\Delta,\beta}(v',0)$ can be transformed in $0$ rounds to a solution for $\rere(\re(\Pi_{\Delta,\beta}(v,0)))$. Hence, and by Theorem \ref{thm:sebastien} this will imply that, given a solution for $\Pi_{\Delta,\beta}(v',0)$ we can solve $\Pi_{\Delta,\beta}(v,0)$ in at most $1$ round of communication.

\paragraph{Node constraint}
We start by showing that, for each possible pair $(\C_{i,j},i')$ such that $\C_{i,j}$ is a color of group $i$ of the original problem and $i \le i' \le \beta$, $\nodeconst''_{\Delta,\beta}(v,0)$ contains some allowed configuration that corresponds to a color of group $i'$ of the new problem.  In other words, we show that the colors of the new problem are generated by all possible pairs composed of a color of the original problem and a color group, where the color group is at least as large as the group of the original color. See Figure \ref{fig:color-generation} for an example.

\begin{figure}[h]
	\centering
	\includegraphics[width=0.75\textwidth]{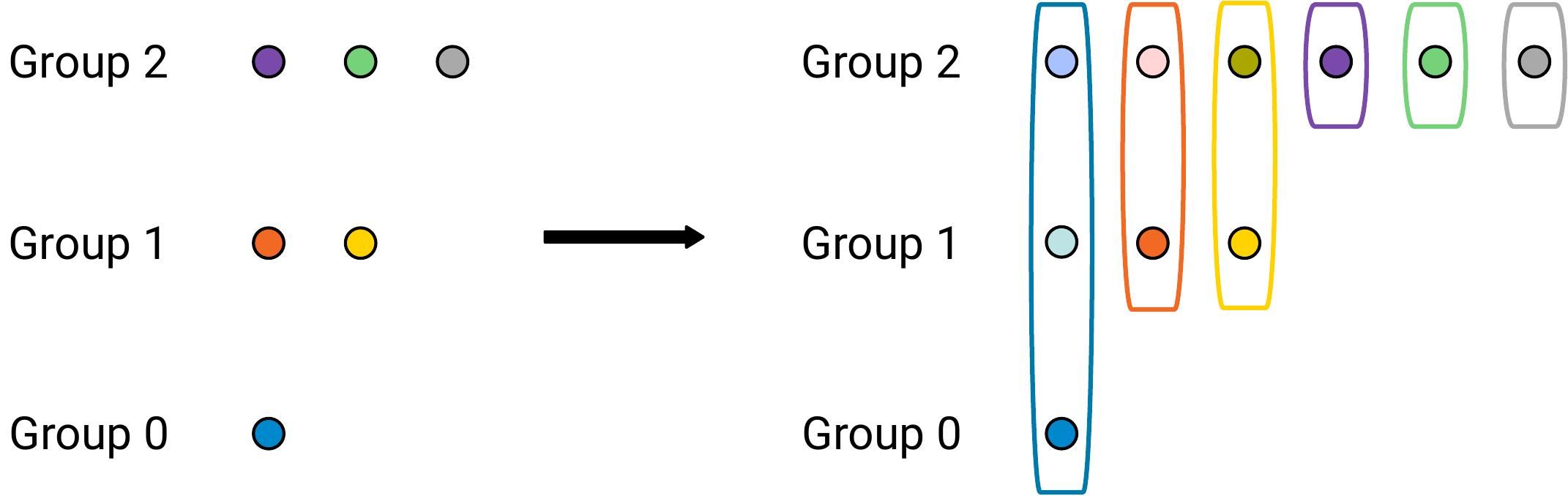}
	\caption{An example of how colors in $\Pi''$ are generated from colors in $\Pi$. In the depicted case, the colors of $\Pi$ can be described by the vector $[1,2,3]$, and each color of group $i$ in $\Pi$ generates a color for every group $j \ge i$ in $\Pi''$; hence the colors in $\Pi''$ can be described by the vector $[1,3,6]$. For example, the orange color in $\Pi$, in group $1$, generates two colors (orange and pink) in groups $1$ and $2$ in $\Pi''$.}
	\label{fig:color-generation}
\end{figure}

Consider an arbitrary choice of $i,j,i'$ as described above, and the set $\C$ of sets defined as $\gen{\gen{ \C_{i,j} , \B_{i'} },\gen{ \A_{i'}}}$ if $i'>0$ and as $\gen{\gen{ \C_{i,j} }}$ if $i'=0$. We show that $\C^\Delta$ is an allowed configuration of $\nodeconst''_{\Delta,\beta}(v,0)$, that is, any choice of sets $(\labels_1,\ldots,\labels_\Delta) \in \C^\Delta$ is an allowed configuration in the node constraint of $\Pi'_{\Delta,\beta}(v,0)$. Any choice satisfies (by construction) the following: each chosen set either contains both $\C_{i,j}$ and $\B_{i'}$ (or just $\C_{i,j}$ if $i'=0$), or it contains $\A_{i'}$ (condition allowed only if $i'>0$). 
If \emph{all} choices fall in the first case, hence also if $i'=0$, then this configuration is contained in $\nodeconst'_{\Delta,\beta}(v,0)$, since all choices over $\dis(\gen{\gen{\C_{i,j}}})^\Delta$  are contained in $\nodeconst'_{\Delta,\beta}(v,0)$.Thus let us now consider the case where at least one choice contains $\A_{i'}$, for $i'>0$. Other choices either also contain $\A_{i'}$, or they contain $\B_{i'}$. Since $\B_{i'}$ is stronger than $\A_{i'}$ with regard to the edge diagram of $\Pi_{\Delta,\beta}(v,0)$, the presence of $\A_{i'}$ in a set implies the presence of $\B_{i'}$ (see Observation \ref{obs:rcs}). Hence, we have a set containing $\A_{i'}$ and all other sets containing $\B_{i'}$, which is a configuration present in $\nodeconst'_{\Delta,\beta}(v,0)$, given by the presence of all choices over $\dis(\gen{\gen{ \A_{i'}}}) \s \dis(\gen{\gen{ \B_{i'}}})^{\Delta-1}$.  Notice that, even though we required $i \le i'$, this is not strictly necessary for obtaining allowed configurations. Nevertheless, since all sets containing $\C_{i,j}$ also contain $\B_i$ (by Observation \ref{obs:rcs}, since sets are right-closed and $\B_i$ is stronger than $\C_{i,j}$) and all $\B_j$ for $j < i$, the configuration obtained by using the pair $(\C_{i,j},i')$ would be the same as the one obtained by using $(\C_{i,j},i)$ if $i'<i$.

We now show that also $\gen{\gen{\A_i}} \s \gen{\gen{\B_i}}^{\Delta-1}$ is an allowed configuration of $\nodeconst''_{\Delta,\beta}(v,0)$, for all $1 \le i \le \beta$. Consider an arbitrary choice: it must be a set containing $\A_i$ and all other sets containing $\B_i$, and thus a configuration present in $\nodeconst'_{\Delta,\beta}(v,0)$, given by the presence of all choices over $\dis(\gen{\gen{ \A_{i}}}) \s \dis(\gen{\gen{ \B_{i}}})^{\Delta-1}$.

\paragraph{Renaming}
We show how to rename the obtained sets of sets, such that, under the proposed renaming, we have the desired relation between $\Pi''_{\Delta,\beta}(v,0)$ and $\Pi_{\Delta,\beta}(v',0)$. We consider the set of sets obtained starting from $(\C_{i,j},i')$ as a new color of group $i'$. Then we map the other set of sets as follows: $\gen{\gen{\B_i}} \mapsto \B_i$, and $\gen{\gen{\A_i}} \mapsto \A_i$.
Notice that $v'$ is indeed the inclusive prefix sum of $v$: since we assume that $i \le i'$, in the new group $i'$ we have a number of colors equal to the sum of the number of colors of the old groups $i \le i'$. Also, we have all the $\A_i\B_i^{\Delta-1}$ configurations in the node constraint of $\Pi''_{\Delta,\beta}(v,0)$. This means that $\nodeconst''_{\Delta,\beta}(v,0)$ is the same as $\nodeconst_{\Delta,\beta}(v',0)$ of $\Pi_{\Delta,\beta}(v',0)$, defined in Section \ref{sec:problems}. On the other hand, we want to show that $\edgeconst''_{\Delta,\beta}(v,0)$ is the same as $\edgeconst_{\Delta,\beta}(v',0)$ of $\Pi_{\Delta,\beta}(v',0)$. For that, it is enough to show that that the following statements hold for the edge constraint $\edgeconst''_{\Delta,\beta}(v,0)$.

\paragraph{Edge constraint}
In order to show the desired properties mentioned above, we must show that the configurations of the edge constraint $\edgeconst''_{\Delta,\beta}(v,0)$ are the following.
\begin{enumerate}
	\item All (new) colors are compatible with all other (new) colors (except themselves).
	\item $\gen{\gen{\B_i}}$ is compatible with $\gen{\gen{\B_j}}$, for all $1 \le i,j \le \beta$.
	\item $\gen{\gen{\A_j}}$ is compatible with $\gen{\gen{\B_i}}$, if and only if $i < j$. 
	\item $\gen{\gen{\A_j}}$ is compatible with all colors of the new group $i$, if and only if $i < j$.
\end{enumerate}
We start with color compatibility. Let $\C_{i,j,k}$ be the color (a set of sets) obtained from the pair $(\C_{i,j},k)$. Consider two new colors $\C_{i,j,k}$ and $\C_{i',j',k'}$. We need to show that there exists a set $\labels_1 \in \C_{i,j,k}$ and a set $\labels_2 \in \C_{i',j',k'}$ such that $\labels_1 \s \labels_2$ is in $\edgeconst'_{\Delta,\beta}(v,0)$, if and only if $(i,j,k) \neq (i',j',k')$. Consider the case where $(i,j) \neq (i',j')$, that is, either $i \neq i'$ or $j \neq j'$. We argue that such sets are given by the configuration $\labels_1 \s \labels_2$ added to $\edgeconst'_{\Delta,\beta}(v,0)$ starting from the good pair $(\{\C_{i,j}\},k)$ when defining the problem $\Pi'_{\Delta,\beta}(v,0)$ (recall the definition of good pair given in Section \ref{sec:firstspeedup}). By construction, $\labels_1$ contains both $\C_{i,j}$ and $\B_k$ if $k>0$, or $\C_{i,j}$ if $k=0$. Since $\C_{i,j,k}$ is defined as $\gen{\gen{\C_{i,j} , \B_{k} },\gen{ \A_{k} }}$ if $k>0$ and as $\gen{\gen{ \C_{i,j} }}$ if $k=0$, then $\labels_1 \in \C_{i,j,k}$, thus we can pick such set. Since $(i,j) \neq (i',j')$, and since by construction $\labels_2$ contains all colors not present in $\labels_1$ and all $\B$-type labels, then $\labels_2$ is contained in $\C_{i',j',k'}$ (because $\labels_2$ is in  $\dis( \gen{\gen{ \C_{i',j'} , \B_{k'} } })$ if $k'>0$, or in the disjunction $\dis( \gen{\gen{ \C_{i',j'}  } })$ if $k'=0$, and $\C_{i',j',k'}$ contains all elements of such disjunction), and thus we can choose such set. Let us now consider the case $k \neq k'$. Without loss of generality, consider the case where $k < k'$. Consider again the configuration $\labels_1 \s \labels_2$ added to $\edgeconst'_{\Delta,\beta}(v,0)$ starting from the good pair $(\{\C_{i,j}\},k)$.  The set $\labels_2$ contains $\A_{k'}$, since $k' > k$, and it is thus contained in $\gen { \gen{ \A_{k'} }}$, that is a subset of $\C_{i',j',k'}$, thus we can pick such set. We now show that color $\C_{i,j,k}$ is not compatible with itself. By picking a pair of sets that both contain $\C_{i,j}$ we get a configuration not in $\edgeconst'_{\Delta,\beta}(v,0)$, since a color never appears on both sides of an edge configuration. By picking a pair of sets that both contain $\A_k$ we get a configuration not in $\edgeconst'_{\Delta,\beta}(v,0)$, since $\A$-type labels never appears on both sides. By picking on one side a set that contains $\C_{i,j}$ and on the other side a set that contains $\A_k$, we get a configuration not in $\edgeconst'_{\Delta,\beta}(v,0)$, because the first choice must also contain $\B_k$ (since all $\B$-type labels are stronger than all $\A$-type labels), and all configurations in $\edgeconst'_{\Delta,\beta}(v,0)$ are such that if $\B_k$ is contained on one side, $A_k$ is not contained on the other side.

We now argue about compatibility between  $\gen{\gen{\B_i}}$ and $\gen{\gen{\B_j}}$, for all  $1 \le i,j \le \beta$. Consider the configuration $\labels_1 \s \labels_2$ added to $\edgeconst'_{\Delta,\beta}(v,0)$ starting from the good pair $(\{\},i)$. Recall from Section \ref{sec:edgeconstraint} that, for this specific good pair, $\labels_1$ is defined as $\{\B_1,\ldots,\B_i\}$, and that $\{\B_1,\ldots,\B_\beta\} \subseteq \labels_2$. The set $\labels_1 $ can be picked from $\gen{\gen{\B_i}}$, and since $\labels_2$ contains all $\B$-type labels, and $\gen{\gen{\B_j}}$ contains $\{\B_1,\ldots,\B_j\}$ and hence a subset of $\labels_2$, the set $\labels_2$ can be picked from $\gen{\gen{\B_j}}$. 

Regarding the compatibility between $\gen{\gen{\A_j}}$ and $\gen{\gen{\B_i}}$, notice that, if $i < j$, then the aforementioned set $\labels_2$, by construction, contains also $\A_j$, and thus $\labels_2$ can be picked from $\gen{\gen{\A_j}}$. Moreover, if $i \ge j$, $\gen{\gen{\A_j}}$ is not compatible with $\gen{\gen{\B_i}}$, since any pair of choices must contain $\B_i$ on one side and $\A_j$ on the other, but by construction all configurations in $\edgeconst'_{\Delta,\beta}(v,0)$ satisfy that if $\B_i$ is present on one side, and $j \le i$, then $\A_j$ is not present on the other side.

Let us now prove the last point. If $j \ge k$, $\gen{\gen{\A_k}}$ is not compatible with colors in the new group $j$, since any choice over such colors either contains another $\A$-type label, or $\B_j$, and such configurations never appear in $\edgeconst'_{\Delta,\beta}(v,0)$. If $j < k$, consider color $\C_{i,j,k}$, and consider the configuration $\labels_1 \s \labels_2$ added to $\edgeconst'_{\Delta,\beta}(v,0)$ starting from the good pair $(\{\C_{i,j}\},k)$. By definition of $\C_{i,j,k}$, we know that $\labels_1$ is contained in $\C_{i,j,k}$, and thus we can choose such set. Then, since $k > j$, $\A_k$ is contained in $\labels_2$, and thus $\labels_2$ is contained in $\gen{\gen{\A_j}}$.

We showed that up to the above renaming, $\nodeconst''_{\Delta,\beta}(v,0) = \nodeconst_{\Delta,\beta}(v',0)$ and $\edgeconst''_{\Delta,\beta}(v,0)= \edgeconst_{\Delta,\beta}(v',0)$, which means that $\Pi''_{\Delta, \beta}(v,0)=\rere(\re(\Pi_{\Delta,\beta}(v,0)))$ can be solved in $0$ rounds given a solution for $\Pi_{\Delta,\beta}(v',0)$, proving Lemma \ref{lem:ub_secondspeedup}.

\subsection{Proof of Lemma \ref{lem:ub_time}}
While the algorithm that we implicitly obtain by applying round elimination multiple times only works on $\Delta$-regular graphs, we will later show, by giving a human-readable version of such algorithm, that it can be adapted to work on any graph.

Let $p(v)$ be the inclusive prefix sum of $v$, and let $p^t(v)$ be the function that recursively applies $t$ times $p$ to $v$ (if $t=0$, then $p^t(v) = v$). That is, $p^{t+1}(v)_j = \sum_{i=0}^j p^t(v)_i$, for all $0 \le j \le \beta$.

Let us see what we get by applying Lemma \ref{lem:ub_secondspeedup} multiple times. We start from $\Pi_{\Delta,\beta}(v=[1,0,\ldots,0],0)$, and by applying Lemma \ref{lem:ub_secondspeedup} we get that it can be solved in $1$ round given a solution for $\Pi_{\Delta,\beta}(p(v),0)$. Applying the same reasoning, the latter problem can be solved in $1$ round given a solution for  $\Pi_{\Delta,\beta}(p(p(v)),0)$, and so on. We get that $\Pi_{\Delta,\beta}([1,0,\ldots,0],0)$ can be solved in $t$ rounds given a coloring with $c \le \size(p^t(v))$ colors.
Hence, the time required to solve $\Pi_{\Delta,\beta}([1,0,\ldots,0],0)$, and thus the $(2,\beta)$-ruling set problem, given a $c$-vertex coloring, is upper bounded by the minimum $t$ such that $c \le \size(p^t(v))$, since $\Pi_{\Delta,\beta}(p^t(v),0)$, the problem obtained by applying $t$ times Lemma \ref{lem:ub_secondspeedup}, can be solved by defining a one-to-one mapping between the given coloring and the $\size(p^t(v))$ configurations of $\nodeconst_{\Delta,\beta}(p^t(v),0)$ that do not correspond to pointers.
Let us now give an exact bound on the size of $p^t(v)$, assuming $v = [1,0,\ldots,0]$.
\begin{lemma}\label{lem:colorgrowth}
	For all $0 \le k \le \beta$, and for all $j \ge 1$, $p^j(v)_k = {j+k-1 \choose k}$.
\end{lemma}
\begin{proof}
	For $j = 1$, the definition of $p^j(v)_k$ yields $p(v)_k = 1 + 0 + 0 + \dots + 0 = {k \choose k}$ for any $0 \leq k \leq \beta$, so the claim holds trivially in that case.
	By induction, and using the binomial identity $\sum_{a=0}^b {c+a \choose a} = {c+b+1 \choose b}$ , we obtain that for $j > 1$ and any $0 \leq k \leq \beta$, we have $p^j(v)_k = \sum_{i=0}^k p^{j-1}(v)_i = \sum_{i=0}^k {j-2+i \choose i} = {j+k-1 \choose k}$.
\end{proof}
By applying Lemma \ref{lem:colorgrowth}, we get that $\size(p^t(v)) = p^{t+1}(v)_\beta = {\beta+t \choose \beta}$. This implies that, given some $c$-coloring, the time required to solve $(2,\beta)$-ruling sets is upper bounded by the minimum $t$ such that ${\beta+t \choose \beta} \ge c$. Since, for $t = \beta \, c^{1/\beta}$, we have that ${\beta+t \choose \beta} \ge \left(\frac{t}{\beta}\right)^\beta \ge c$, we obtain that $(2,\beta)$-ruling sets can be found in $\beta \, c^{1/\beta}$ rounds. Also, by using the inequality  ${x+y \choose x} \ge \left(\frac{x}{y}\right)^y$, we obtain that $(2, \beta \, c^{1/\beta})$-ruling sets can be found in at most $\beta$ rounds, since ${\beta c^{1/\beta} + \beta \choose \beta c^{1/\beta}} \ge \left(\frac{\beta c^{1/\beta}}{\beta}\right)^\beta = c$.

\subsection{Intuition behind the algorithm}
While proving Lemma \ref{lem:ub_secondspeedup} we showed that, starting from problem $\Pi_{\Delta,\beta}(v,0)$, we can apply the round elimination theorem and obtain a problem that can be made harder such that the obtained sets of sets can be renamed to obtain problem $\Pi_{\Delta,\beta}(v',0)$, where $v'$ is the inclusive prefix sum of $v$. This implies that, given a solution for $\Pi_{\Delta,\beta}(v',0)$, we can obtain, in $0$ rounds, a solution for $\rere(\re(\Pi_{\Delta,\beta}(v,0)))$. Also, note that each problem $\Pi_{\Delta,\beta}(v,0)$ can be solved given a $c \le \size(v)$ coloring. Moreover, the round elimination theorem implies that by spending $1$ round of communication, given a solution for $\rere(\re(\Pi_{\Delta,\beta}(v,0)))$ nodes can obtain a solution for $\Pi_{\Delta,\beta}(v,0)$. Hence, given a solution for $\Pi_{\Delta,\beta}(v',0)$, we can obtain a solution for $\Pi_{\Delta,\beta}(v,0)$ by spending $1$ round of communication. In order to do that, nodes can compute (offline) the inverse renaming from the labels of the current problem to the labels (sets of sets) of the previous problem. Then, they can spend $1$ round to share their current labels with the neighbors, and choose a label for the previous problem that satisfies the previous node constraint and edge constraint (and the round elimination theorem guarantees that this is possible). Essentially, this means that, given a proof for an upper bound using round elimination, a proper algorithm can be obtained by ``executing the proof in the reverse order'', and a crucial part is that the (inverse) label renaming can be computed without need of coordination.

While the round elimination theorem guarantees us the existence of an algorithm, we can actually extract some human-understandable version of such procedure by recalling how we defined the colors of the new problem starting from the colors of the old problem. In fact, there is a very simple idea behind the round elimination generated algorithm: we have colors and pointers. Each color corresponds to a pair containing a color $c$ and a group $i$ of the previous step. Group numbers give priorities. If no neighbor has the same color $c$ with a better (lower) priority, pick color $c$, otherwise point to that neighbor.  An example of this color reduction step is shown in Figure \ref{fig:color-reduction} (note that this is essentially the inverse of the mapping previously shown in Figure \ref{fig:color-generation}). 

\begin{figure}[h]
	\centering
	\includegraphics[width=0.75\textwidth]{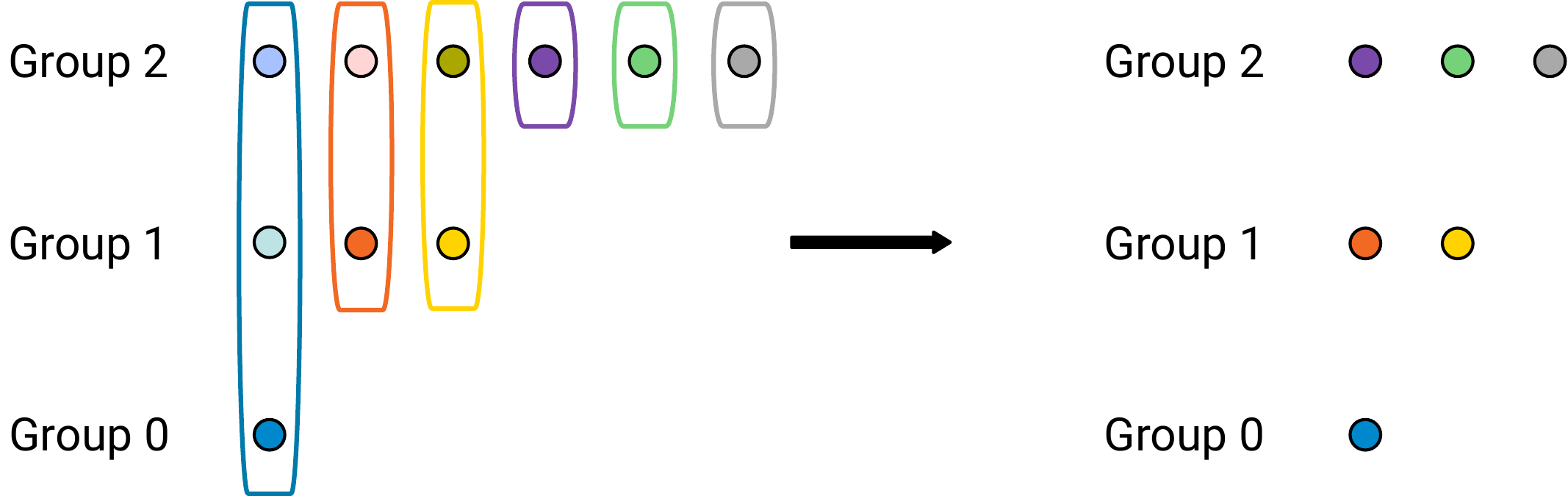}
	\caption{An example of the reduction of the colors. A pink node would either change its color to orange, if no neighbor has that color, or it would uncolor itself and point to an orange neighbor.}
	\label{fig:color-reduction}
\end{figure}

 More formally, let us now see how to transform a solution for $\Pi_{\Delta,\beta}(p^{t}(v),0)$, given a $c \le p^t(v)$ coloring, to a solution for $\Pi_{\Delta,\beta}(v=[1,0,\ldots,0],0)$, in $t$ rounds of communication. At first, nodes start by (offline) computing some mapping from colors to groups such that in each group $k$, at most $p^{t}(v)_k$ colors are present. Thus each color can be seen as $\C_{k,i'}$, for some $0 \le k \le \beta$ and $1 \le i' \le p^{t}(v)_k$. Now, let us see how nodes can convert a solution for $\Pi_{\Delta,\beta}(p^{t'+1}(v),0)$ to a solution for $\Pi_{\Delta,\beta}(p^{t'}(v),0)$ in $1$ round of communication, for each $0 \le t' < t$. Recall that each color of the current problem corresponds to a good pair of the previous problem, and nodes can compute offline such an inverse mapping. Hence each color $\C_{k,i'}$ corresponds to some pair $(\C_{i,j},k)$ satisfying $k \ge i$. The algorithm, for a node having the pair $(\C_{i,j},k)$, does the following:
\begin{enumerate}
	\item Gather the pair $(\c_u,\g_u)$ of each neighbor $u$.
	\item If no neighbor satisfies $\c_u = \C_{i,j} \land \g_u < k$, then set $\C_{i,j}$ as new color.
	\item Otherwise, output $\A_i$ on the port connecting to that neighbor, and $\B_i$ on the others.
\end{enumerate}

The reason why this algorithm works is the following. At each step the number of colors reduces, and at the end there is only one ``color'' that corresponds to nodes in the ruling set. In order to get rid of colors fast, nodes do not wait to know if nodes are in the ruling set in order to point to them. Instead, they are ensured that if they output some pointer with distance $i$, in the next step the pointed node will either output some color of group smaller than $i$, or a pointer with distance smaller than $i$. In case the neighbor outputs a color, in the next steps such node will either output a color of smaller groups or a smaller pointer (this condition is guaranteed by how ``good pairs'' are defined). Notice that the provided algorithm does not need the graph to be $\Delta$-regular, hence we get an algorithm that works in all graphs.

\section{Lower bound}\label{sec:lb}

In this section we prove lower bounds on the time complexity for computing $(2,\beta)$-ruling sets using deterministic algorithms in the port numbering model. We will prove such lower bounds using the round elimination theorem. While it is not difficult to argue that in the port numbering model $(2,\beta)$-ruling sets are not solvable with deterministic algorithms, we will later see in Section \ref{sec:liftlocal} how to convert the lower bounds obtained in this section to lower bounds for the LOCAL model. We will heavily exploit the fact that the lower bounds that we prove in this section show the existence of a sequence of problems $\Pi_0 \rightarrow \Pi_1 \rightarrow \dots \rightarrow \Pi_t$, where $\Pi_t$ is not $0$-round solvable, and the problem family is obtained by applying the round elimination theorem multiple times in a way that preserves the property that the number of labels of each problem is not too large. As we will see in Section \ref{sec:liftlocal}, this is the only property that we need, in order to ensure that, even by allowing randomization, the problems of such a family do not become much easier.

We start by defining the notion of \emph{relaxation}, which allows us to relate different configurations.
\begin{definition}
	We say that a configuration $\A_1 \s \dots \s \A_\Delta$ of sets \emph{can be relaxed to} a configuration $\B_1 \s \dots \s \B_\Delta$ of sets if there exists a bijection $\sigma: \{1, \dots, \Delta\} \to \{1, \dots, \Delta\}$ such that $\A_k \subseteq \B_{\sigma(k)}$ for each $1 \leq k \leq \Delta$.
	For simplicity, we will assume in the following w.l.o.g., that $\sigma(k) = k$ for each $1 \leq k \leq \Delta$, i.e., that the bijection is the identity.
\end{definition}

We now define a node constraint $\mathcal Z_{\Delta, \beta}(v, x)$ that will be useful later. In particular, we will show in Lemma~\ref{lem:industop} that each node configuration of $\rere(\re(\Pi_{\Delta,\beta}(v,x)))$ can be relaxed to some configuration in $\mathcal Z_{\Delta, \beta}(v, x)$.
\begin{definition}
	Recall that $\Sigma_{\Delta,\beta}(v,x)$ denotes the set of all labels used in $\Pi_{\Delta, \beta}(v, x)$, $\Sigma'_{\Delta,\beta}(v,x)$ the set of all labels used in $\re( \Pi_{\Delta,\beta}(v,x) )$, and $\mathcal C$ the set of all colors $\C_{i,j} \in \Sigma_{\Delta,\beta}(v,x)$, and that $\size(v) = v_0 + \dots + v_\beta$.
	We will denote the group of a color $\C \in \mathcal C$ by $g(\C)$, i.e., if $\C = \C_{i,j}$, then $g(\C) = i$ (similar to the already defined notion for sets of colors).
	Moreover, for any $v, x$ such that $\size(v) \cdot (x+1) + 1 \leq \Delta$, we denote by $\mathcal Z_{\Delta, \beta}(v, x)$ the set containing the configurations
	\begin{align*}
		\gen{\gen{\A_i}} \quad & \gen{\gen{\B_i}}^{\Delta-1} \enspace &\text{ for any $1 \leq i \leq \beta$,}\\
		\gen{\gen{\C}}^{\Delta - \size(v) \cdot (x+1)}\quad &\, (\Sigma'_{\Delta,\beta}(v,x))^{\size(v) \cdot (x+1)} \enspace &\text{ for any $\C \in \mathcal C$ satisfying $g(\C) = 0$, and}\\
		\gen{\gen{\C, \B_i}, \gen{\A_i}}^{\Delta - \size(v) \cdot (x+1)}\quad &\, (\Sigma'_{\Delta,\beta}(v,x))^{\size(v) \cdot (x+1)} \enspace &\text{ for any $\C \in \mathcal C$ and  $\max\{ 1, g(\C) \} \leq i \leq \beta$.}
	\end{align*}
	Here, the inner $\gen{}$ are taken w.r.t.\ $\edgeconst_{\Delta, \beta}(v, x)$, and the outer $\gen{}$ w.r.t.\ $\nodeconst'_{\Delta, \beta}(v, x)$.
\end{definition}

The next lemma restricts the space of configurations that can possibly appear in the node constraint of  $\rere(\re( \Pi_{\Delta,\beta}(v,x) ))$. We will later use this lemma inductively in order to restrict this space even further.
\begin{lemma}\label{lem:industep}
	Let $1 \leq i \leq \beta - 1$, and let $v, x$ satisfy $\size(v) \cdot (x+1) + 1 \leq \Delta$.
	Let $\mathcal U = \U_1 \s \dots \s \U_\Delta$ be a node configuration of $\rere(\re( \Pi_{\Delta,\beta}(v,x) ))$ such that
	\begin{enumerate}[label=(\roman*)]
		\item \label{item:two}  $\mathcal U$ cannot be relaxed to any configuration in $\mathcal Z_{\Delta, \beta}(v,x)$, and
		\item \label{item:four} each set contained in $\U_k$ contains $\B_i$, for any $1 \leq k \leq \Delta$. 
	\end{enumerate}
	Then, each set contained in $\U_k$ contains $\B_{i+1}$, for any $1 \leq k \leq \Delta$.
\end{lemma}
\begin{proof}
	We begin by proving two useful claims.

	\underline{Claim 1:} For any color $\C \in \mathcal C$, there are at least $\size(v) \cdot (x+1) + 1$ distinct indices $k$ such that $\U_k$ contains a set that contains neither $\A_i$ nor $\C$.
	
	For a contradiction, suppose that there is a color $\C$ for which the claim is false.
	W.l.o.g., we can assume that $\U_1, \dots, \U_{k'}$ are exactly those $\U_k$ that contain a set that contains neither $\A_i$ nor $\C$, where $k' \leq \size(v) \cdot (x+1)$.
	Now consider an arbitrary set $\U_k$ with $k > k'$, and some arbitrary set $\W \in \U_k$.
	Since $k > k'$, we know that $\W$ must contain $\A_i$ or $\C$.
	As $\W$ also contains $\B_i$ by Property~$\ref{item:four}$ and is right-closed by Observation~\ref{obs:rcs} (since $\W$ is a label/set used in $\re(\Pi_{\Delta, \beta}(v, x))$), it follows that $\W$ is a superset of $\gen{\C, \B_i}$ or a superset of $\gen{\A_i, \B_i}$.
	Since $\B_i$ is stronger than $\A_i$ according to $\edgeconst_{\Delta, \beta}(v, x)$, we have $\gen{\A_i, \B_i} = \gen{\A_i}$; hence, $\W \in \gen{\gen{\C, \B_i}, \gen{\A_i}}$, by Lemma~\ref{lem:restrong}, and we obtain that $\U_k \subseteq \gen{\gen{\C, \B_i}, \gen{\A_i}}$.
	This implies that every $\U_k$ with $k > k'$ is a subset of $\gen{\gen{\C, \B_i}, \gen{\A_i}}$, and since every $\U_k$ is a subset of $\Sigma'_{\Delta,\beta}(v,x)$, and $k' \leq \size(v) \cdot (x+1)$, it follows that $\mathcal U$ can be relaxed to the configuration $\gen{\gen{\C, \B_i}, \gen{\A_i}}^{\Delta - \size(v) \cdot (x+1)} \, (\Sigma'_{\Delta,\beta}(v,x))^{\size(v) \cdot (x+1)}$.
	We observe that, for all $1 \leq j \leq g(\C)$, we have that $\gen{\C, \B_j} = \gen{\C}$ since $\B_j$ is stronger than $\C$ according to $\edgeconst_{\Delta, \beta}(v, x)$.
	Therefore, $\gen{\C, \B_i} = \gen{\C, \B_{\max\{i, g(\C)\}}}$, which implies that the configuration $\gen{\gen{\C, \B_i}, \gen{\A_i}}^{\Delta - \size(v) \cdot (x+1)} \, (\Sigma'_{\Delta,\beta}(v,x))^{\size(v) \cdot (x+1)}$ is contained in $\mathcal Z_{\Delta, \beta}(v, x)$.
	This yields a contradiction to Property~$\ref{item:two}$ and proves the claim.

	\underline{Claim 2:} Each $\U_k$ contains a set that does not contain $\A_i$.
	
	For a contradiction, suppose that the claim is false, and let $k'$ be an index such that any set $\W$ contained in $\U_{k'}$ contains $\A_i$.
	Since any such $\W$ is right-closed by Observation~\ref{obs:rcs} (since $W$ is a label/set used in $\re(\Pi_{\Delta, \beta}(v, x))$), this implies that $\gen{\A_i} \subseteq \W$ for any $\W \in \U_{k'}$.
	Hence, by Lemma~\ref{lem:restrong}, we obtain that, for any $\W \in \U_{k'}$, $\W$ is at least as strong as $\gen{\A_i}$ according to $\nodeconst'_{\Delta, \beta}(v, x)$, and it follows that $\U_{k'}$ is a subset of $\gen{\gen{\A_i}}$.
	With an analogous argumentation, we see that any $\U_k$ is a subset of $\gen{\gen{\B_i}}$ since for any $\U_k$, each set contained in $\U_k$ contains $\B_i$, by Property~$\ref{item:four}$.
	It follows that $\mathcal U$ can be relaxed to the configuration $\gen{\gen{\A_i}} \s \gen{\gen{\B_i}}^{\Delta-1}$ contained in $\mathcal Z_{\Delta, \beta}(v, x)$, yielding a contradiction to Property~$\ref{item:two}$ and proving the claim.

	Now, we are set to prove the lemma.
	For a contradiction, suppose that the lemma does not hold, and, w.l.o.g., let $\W_1 \in \U_1$ be a set that does not contain $\B_{i+1}$.
	We will now pick one set from each $\U_k$ such that the obtained configuration satisfies certain useful properties.
	Consider the partial configuration of sets obtained by picking $\W_1$ from $\U_1$, and picking, for each color $\C \in \mathcal C$, $(x+1)$ sets $\W_k$ from distinct $\U_k$ such that $\W_k$ contains neither $\A_i$ nor $\C$.
	To be precise, the $\size(v) \cdot (x+1) + 1$ sets picked so far are picked from $\size(v) \cdot (x+1) + 1$ $\U_k$ with pairwise distinct index.
	This is possible due to Claim 1 (by simply picking $\W_1$ first, and then picking the remaining $\size(v) \cdot (x+1)$ sets with the desired properties in an arbitrary order).
	We complete this partial configuration to a configuration $\W_1, \s \dots \s \W_\Delta$ containing $\Delta$ sets by picking from each remaining $\U_k$ a set that does not contain $\A_i$.
	This is possible by Claim 2.

	We argue that $\W_1 \s \dots \s \W_\Delta$ is not a node configuration of $\re( \Pi_{\Delta,\beta}(v,x) )$:
	Since, by construction, for each color $\C \in \mathcal C$, there are at least $x+1$ distinct indices $k$ such that $\W_k$ does not contain $\C$, the configuration $\W_1 \s \dots \s \W_\Delta$ is not contained in the condensed configuration $\dis(\gen{\gen{\C}})^{\Delta - x} \s \dis(\gen{\gen{\X}})^x$.
	Moreover, as $\B_{i+1}$ is stronger than $\A_i$ according to $\edgeconst_{\Delta, \beta}(v, x)$, the fact that $\W_1$ does not contain $\B_{i+1}$ implies that $\W_1$ also does not contain $\A_i$, as $\W_1$ is right-closed by Observation~\ref{obs:rcs}.
	Thus, none of the $\W_k$ contains $\A_i$, and we obtain for each $i' \leq i$ that none of the $\W_k$ contains $\A_{i'}$, since $\A_i$ is stronger than $\A_{i'}$ according to $\edgeconst_{\Delta, \beta}(v, x)$ and each $\W_k$ is right-closed.
	It follows that $\W_1 \s \dots \s \W_\Delta$ is not contained in the condensed configuration $\dis(\gen{\gen{\A_{i'}}}) \s \dis(\gen{\gen{\B_{i'}}})^{\Delta-1}$, for any $i' \leq i$.
	Finally, for any $i' > i$, the fact that $\W_1$ does not contain $\B_{i+1}$ implies that $\W_1$ also contains neither $\A_{i'}$ nor $\B_{i'}$ since both are at least as weak as $\B_{i+1}$ according to $\edgeconst_{\Delta, \beta}(v, x)$ and $\W_1$ is right-closed.
	Hence, $\W_1 \s \dots \s \W_\Delta$ is not contained in the condensed configuration $\dis(\gen{\gen{\A_{i'}}}) \s \dis(\gen{\gen{\B_{i'}}})^{\Delta-1}$, for any $i' > i$.
	
	Therefore, $\W_1 \s \dots \s \W_\Delta$ is contained in none of the condensed configurations in $\nodeconst'_{\Delta, \beta}(v, x)$, and we conclude that $\W_1 \s \dots \s \W_\Delta$ is not a node configuration of $\re( \Pi_{\Delta,\beta}(v,x) )$.
	Since we have $\W_k \in \U_k$ for each $1\leq k \leq \Delta$, it follows by the definition of the function $\rere(\cdot)$ that $\U_1 \s \dots \s \U_\Delta$ is not a node configuration of $\rere(\re( \Pi_{\Delta,\beta}(v,x) ))$.
	This yields the desired contradiction and proves the lemma.
\end{proof}

We now prove an analogous statement of the previous lemma, for the case where $i=0$.
\begin{lemma}\label{lem:industart}
	Let $v, x$ satisfy $\size(v) \cdot (x+1) + 1 \leq \Delta$, and let $\mathcal U = \U_1 \s \dots \s \U_\Delta$ be a node configuration of $\rere(\re( \Pi_{\Delta,\beta}(v,x) ))$ such that $\mathcal U$ cannot be relaxed to any configuration in $\mathcal Z_{\Delta, \beta}(v,x)$.
	Then, each set contained in $\U_k$ contains $\B_1$, for any $1 \leq k \leq \Delta$.
\end{lemma}
\begin{proof}
	This proof is essentially a special case of the proof of Lemma~\ref{lem:industep} that only requires a subset of the arguments.
	Given that the approach is analogous, we will only provide the minimally necessary arguments without additional explanations.
	
	\underline{Claim 1:} For any color $\C \in \mathcal C$, there are at least $\size(v) \cdot (x+1) + 1$ distinct indices $k$ such that $\U_k$ contains a set that does not contain $\C$.
	
	For a contradiction, we assume that the claim is false for some color $\C$, and that $\U_1, \dots, \U_{k'}$ are those $\U_k$ that contain a set that does not contain $\C$, where $k' \leq \size(v) \cdot (x+1)$.
	Hence, for any $\U_k$ with $k > k'$, any set contained in $\U_k$ must contain $\C$, which implies that $\U_k \subseteq \gen{\gen{\C}}$, due to the right-closedness of all sets contained in $\U_k$ and Lemma~\ref{lem:restrong}.
	Thus, $\mathcal U$ can be relaxed to $\gen{\gen{\C}}^{\Delta - \size(v) \cdot (x+1)} \, (\Sigma'_{\Delta,\beta}(v,x))^{\size(v) \cdot (x+1)}$.
	If $g(\C) = 0$, the latter configuration is contained in $\mathcal Z_{\Delta, \beta}(v, x)$, yielding a contradiction and proving the claim.
	Hence, assume that $g(\C) \geq 1$.
	Since $\B_{g(\C)}$ is stronger than $\C$ according to $\edgeconst_{\Delta, \beta}(v, x)$, we have $\gen{\C} = \gen{\C, \B_{g(\C)}}$, by the definition of $\gen{}$.
	Furthermore, since $\A_{g(\C)}$ is weaker than $\C$ according to $\edgeconst_{\Delta, \beta}(v, x)$, we have that $\gen{\A_{g(\C)}}$ is stronger than $\gen{\C}$ according to $\nodeconst'_{\Delta, \beta}(v, x)$, by Corollary~\ref{cor:strongflip}.
	Hence, $\gen{\gen{\C}} = \gen{\gen{\C}, \gen{\A_{g(\C)}}} = \gen{\gen{\C, \B_{g(\C)}}, \gen{\A_{g(\C)}}}$.
	It follows that the configuration $\gen{\gen{\C}}^{\Delta - \size(v) \cdot (x+1)} \s (\Sigma'_{\Delta,\beta}(v,x))^{\size(v) \cdot (x+1)}$ is contained in $\mathcal Z_{\Delta, \beta}(v, x)$, which yields a contradiction and proves the claim.

	Now assume for a contradiction that the lemma does not hold, and w.l.o.g., pick one set $\W_k$ from each $\U_k$, such that $\W_1$ does not contain $\B_1$, and for each color $\C$, there are at least $(x+1)$ distinct $k$ such that $\W_k$ does not contain $\C$.
	This is possible due to Claim 1, by an analogous argumentation to the one in the proof of Lemma~\ref{lem:industep}.

	Now, again, the fact that for each color $\C$, at least $x+1$ many $\W_k$ do not contain $\C$ ensures that $\W_1 \s \dots \s \W_\Delta$ is not contained in the condensed configuration $\dis(\gen{\gen{\C}})^{\Delta - x} \s \dis(\gen{\gen{\X}})^x$.
	Also, since for any $1 \leq i \leq \beta$ both $\A_i$ and $\B_i$ are at least as weak as $\B_1$ according to $\edgeconst_{\Delta, \beta}(v, x)$, $\W_1$ does not contain any $\A_i$ or $\B_i$ due to the right-closedness of $\W_1$, which ensures that $\W_1 \s \dots \s \W_\Delta$ is not contained in the condensed configuration $\dis(\gen{\gen{\A_i}}) \s \dis(\gen{\gen{\B_i}})^{\Delta-1}$, for any $1 \leq i \leq \beta$.
	Hence, $\W_1 \s \dots \s \W_\Delta$ is not a node configuration of $\re( \Pi_{\Delta,\beta}(v,x) )$, and we conclude that $\mathcal U$ is not a node configuration of $\rere(\re( \Pi_{\Delta,\beta}(v,x) ))$.
	This yields the desired contradiction and proves the lemma.
\end{proof}

Again, we prove an analogous statement of Lemma \ref{lem:industep}, for the case where $i=\beta$.
\begin{lemma}\label{lem:industop}
	Let $v, x$ satisfy $\size(v) \cdot (x+1) + 1 \leq \Delta$.
	Then any node configuration of $\rere(\re( \Pi_{\Delta,\beta}(v,x) ))$ can be relaxed to some configuration in $\mathcal Z_{\Delta, \beta}(v,x)$.
\end{lemma}
\begin{proof}
	This proof is essentially a special case of the proof of Lemma~\ref{lem:industep}, where $i = \beta$.
	For a contradiction, assume that the lemma is false, and let $\mathcal U = \U_1 \s \dots \s \U_\Delta$ be a node configuration of $\rere(\re( \Pi_{\Delta,\beta}(v,x) ))$ such that $\mathcal U$ cannot be relaxed to any configuration in $\mathcal Z_{\Delta, \beta}(v,x)$.
	By Lemma~\ref{lem:industep} and Lemma~\ref{lem:industart}, we know that for each $1 \leq k \leq \Delta$, any set contained in $\U_k$ contains $\B_\beta$.
	Hence, the premises of Lemma~\ref{lem:industep} are satisfied for $i = \beta$, and by an analogous argumentation to the one in the proof of Lemma~\ref{lem:industep}, we obtain the following: we can pick one set $\W_k$ from each $\U_k$, such that for each color $\C$, there are at least $(x+1)$ distinct $k$ such that $\W_k$ does not contain $\C$, and each $\W_k$ does not contain $\A_\beta$.
	We note that in our case of $i = \beta$, we do not pick a special set $\W_1$, and that the condition $i \leq \beta - 1$ in Lemma~\ref{lem:industep} is not used in the parts of the proof we use here analogously.
	
	Now, similarly as before, the fact that there are at least $(x+1)$ distinct $k$ such that $\W_k$ does not contain $\C$ ensures that $\W_1 \s \dots \s \W_\Delta$ is not contained in the condensed configuration $\dis(\gen{\gen{\C}})^{\Delta - x} \s \dis(\gen{\gen{\X}})^x$, and the fact that $\A_\beta$ is not contained in any $\W_k$ ensures that, for any $1 \leq i \leq \beta$, label $\A_i$ is not contained in any $\W_k$, which in turn implies that $\W_1 \s \dots \s \W_\Delta$ is not contained in the condensed configuration $\dis(\gen{\gen{\A_i}}) \s \dis(\gen{\gen{\B_i}})^{\Delta-1}$, for any $1 \leq i \leq \beta$.
	Again, we obtain a contradiction to the fact that $\mathcal U$ is a node configuration of $\rere(\re( \Pi_{\Delta,\beta}(v,x) ))$.
\end{proof}

We now prove that, given a solution for $\rere(\re( \Pi_{\Delta,\beta}(v,x) ))$, problem $\Pi_{\Delta, \beta}(v', x')$ can be solved in $0$ rounds, for some specific value of $v'$ and $x'$. We will exploit the fact that any node configuration of $\rere(\re( \Pi_{\Delta,\beta}(v,x) ))$ can be relaxed to some node configuration of $\Pi_{\Delta, \beta}(v', x')$, and we will additionally relate the edge constraints of the two problems.
\begin{lemma}\label{lem:rerere}
	Let $v, x$ satisfy $\size(v) \cdot (x+1) + 1 \leq \Delta$, and set $v' := [v'_0, \dots, v'_\beta]$, $v'_i := \sum_{j=0}^i v_j$ for all $0 \leq i \leq \beta$, and $x' := \size(v) \cdot (x+1)$.
	Then, given a solution for $\rere(\re( \Pi_{\Delta,\beta}(v,x) ))$, problem $\Pi_{\Delta, \beta}(v', x')$ can be solved in $0$ rounds.
\end{lemma}
\begin{proof}
	Denote the set of colors contained in $\Sigma_{\Delta,\beta}(v,x)$ by $\mathcal C$, and the set of colors contained in $\Sigma_{\Delta,\beta}(v',x')$ by $\mathcal D$.
	W.l.o.g, we can assume that $\mathcal C$ and $\mathcal D$ are disjoint, as we can simply rename colors if they are not disjoint.
	Furthermore, we will use $\A_i$, $\B_i$, and $\X$ as usual when talking about labels from $\Sigma_{\Delta,\beta}(v,x)$, and use $\A'_i$, $\B'_i$, and $\X'$ instead when talking about labels from $\Sigma_{\Delta,\beta}(v',x')$.
	We transform a solution for $\rere(\re( \Pi_{\Delta,\beta}(v,x) ))$ into a solution for $\Pi_{\Delta, \beta}(v', x')$ in $0$ rounds as follows.
	
	Let $f: \{ (\C, i) \mid \C \in \mathcal C, 0 \leq i \leq \beta, g(\C) \leq i \} \to \mathcal D$ be an arbitrary, but fixed, bijection such that $g(f((\C, i))) = i$, i.e., such that any pair $(\C, i)$ is mapped to a color in $\mathcal D$ that is contained in group $i$.
	By the definition of $v'$, such a bijection exists.
	Consider a node $u$ and let $\mathcal U = \U_1 \s \dots \s \U_\Delta$ be the node configuration around $u$ in the given solution for $\rere(\re( \Pi_{\Delta,\beta}(v,x) ))$.
	Node $u$ starts by choosing a configuration $\mathcal U' = \U'_1 \s \dots \s \U'_\Delta \in \mathcal Z_{\Delta, \beta}(v, x)$ such that $\U_k \subseteq \U'_k$ for each $1 \leq k \leq \Delta$, and replacing each $\U_k$ with $\U'_k$ on the respective incident edges.
	By Lemma~\ref{lem:industop}, such a configuration exists.
	Then, $u$ proceeds as follows.
	
	If $\mathcal U' = \gen{\gen{\A_i}} \s \gen{\gen{\B_i}}^{\Delta-1}$ for some $1 \leq i \leq \beta$, then node $u$ simply replaces $\gen{\gen{\A_i}}$ with $\A'_i$ and $\gen{\gen{\B_i}}$ with $\B'_i$ on the respective incident edges to produce its output.
	If, for some $\C \in \mathcal C$ satisfying $g(\C) = 0$, $\mathcal U' = \gen{\gen{\C}}^{\Delta - \size(v) \cdot (x+1)} \s (\Sigma'_{\Delta,\beta}(v,x))^{\size(v) \cdot (x+1)}$,  then $u$ replaces $\gen{\gen{\C}}$ with $f((\C, 0))$ and $(\Sigma'_{\Delta,\beta}(v,x))$ with $\X'$ on the respective incident edges to produce its output.
	If $\mathcal U' = \gen{\gen{\C, \B_i}, \gen{\A_i}}^{\Delta - \size(v) \cdot (x+1)} \s (\Sigma'_{\Delta,\beta}(v,x))^{\size(v) \cdot (x+1)}$ for some $(\C, i)$ with $\max\{ 1, g(\C) \} \leq i \leq \beta$, then $u$ replaces $\gen{\gen{\C, \B_i}, \gen{\A_i}}$ with $f((\C, i))$ and $(\Sigma'_{\Delta,\beta}(v,x))$ with $\X'$ on the respective incident edges to produce its output.
	In the following, we show that this yields a correct output for $\Pi_{\Delta, \beta}(v', x')$.

	As the nodes clearly produce node configurations that are contained in $\nodeconst_{\Delta,\beta}(v',x')$, the only thing to be shown is that the output pair on each edge is an edge configuration contained in $\edgeconst_{\Delta,\beta}(v',x')$.
	Hence, consider an edge $\{u, w\}$, and let $\U \s \W$ be the edge configuration on $\{ u, w \}$ in the solution for $\rere(\re( \Pi_{\Delta,\beta}(v,x) ))$.
	Let $\U'$ be the label with which $u$ replaces $\U$ (in the first step), and $\W'$ the label $w$ replaces $\W$ with, and let $\U'' \s \W''$ be the final output on edge $\{u, w\}$ (after the second replacement).
	We now show that the above $0$-round algorithm does not produce an incorrect edge configuration, by going through all incorrect edge configurations one by one and showing that they do not occur.
	For this, we use that the correctness of $\U \s \W$ (i.e., the fact that $\U \s \W$ is contained in the edge constraint of $\rere(\re( \Pi_{\Delta,\beta}(v,x) ))$) implies that there exist sets $\U^* \in \U'$ and $\W^* \in \W'$ such that for all pairs $(\u^*, \w^*)$ with $\u^* \in \U^*$, $\w^* \in \W^*$ we have that $\u^* \s \w^*$ is an edge configuration contained in $\edgeconst_{\Delta,\beta}(v,x)$.
	This follows by the definitions of $\re(\cdot)$ and $\rere(\cdot)$, and the fact that $\U \subseteq \U'$ and $\W \subseteq \W'$.
	In the following, $\U^*$ and $\W^*$ will always be \emph{arbitrarily chosen} sets contained in $\U'$ and $\W'$, respectively, and $\u^*$ and $\w^*$ will be picked suitably from $\U^*$ and $\W^*$, respectively.

	If $\U'' = \C' = \W''$ for some $\C' \in \mathcal D$ satisfying $g(\C') \geq 1$, then we have $\U' = \gen{\gen{\C, \B_i}, \gen{\A_i}} = \W'$ for some $(\C, i)$ with $\max\{ 1, g(\C) \} \leq i \leq \beta$.
	Hence, by Lemma~\ref{lem:restrong}, we have $\gen{\C, \B_i} \subseteq \U^*$ or $\gen{\A_i} \subseteq \U^*$, which implies that $\{\C, \B_i\} \subseteq \U^*$ or $\A_i \in \U^*$.
	Since analogous statements hold for $\W^*$, we see that we can pick $\u^*$ and $\w^*$ such that $\u^* \s \w^*$ is one of the configurations listed in $\{ \A_i \s \A_i, \C \s \C , \A_i \s \B_i \}$.
	Since neither of the listed configurations is contained in $\edgeconst_{\Delta,\beta}(v,x)$, we obtain the desired contradiction.

	If $\U'' = \C' = \W''$ for some $\C' \in \mathcal D$ satisfying $g(\C') = 0$, then we have $\U' = \gen{\gen{\C}} = \W'$ for some $\C \in \mathcal C$.
	Analogously to the previous case, we obtain $\C \in \U^*$ and $\C \in \W^*$.
	Picking $\u^* = \C = \w^*$ yields the desired contradiction as $\C \s \C \notin \edgeconst_{\Delta,\beta}(v,x)$.

	If $\U'' = \A'_i$, $\W'' = \B'_j$ where $i \leq j$, then $\U' = \gen{\gen{A_i}}$, $\W' = \gen{\gen{B_j}}$.
	Analogously to before, we obtain $\A_i \in \U^*$ and $\B_j \in \W^*$.
	Picking $\u^* = \A_i$ and $\w^* = \B_j$ yields the desired contradiction since $\A_i \s \B_j \notin \edgeconst_{\Delta,\beta}(v,x)$ if $i \leq j$.

	If $\U'' = \A'_i$, $\W'' = \C'$ where $i \leq g(\C')$, then $\U' = \gen{\gen{A_i}}$, $\W' = \gen{\gen{\C, \B_j}, \gen{\A_j}}$ where $i \leq j$ and $g(\C) \leq j$.
	Analogously to before, we obtain $\A_i \in \U^*$, and $\{\C, \B_j\} \subseteq \W^*$ or $\A_j \in \W^*$.
	Hence, we can pick $\u^* = \A_i$ and $\w^* = \A_j$, or $\u^* = \A_i$ and $\w^* = \B_j$; in both cases, we obtain a contradiction as both $\A_i \s \A_j$ and $\A_i \s \B_j$ are not contained in $\edgeconst_{\Delta,\beta}(v,x)$ if $i \leq j$.

	Finally, if $\U'' = \A'_i$, $\W'' = \A'_j$ for some $1 \leq i, j \leq \beta$, then $\U' = \gen{\gen{A_i}}$, $\W' = \gen{\gen{A_j}}$.
	Analogously to before, we obtain $\A_i \in \U^*$ and $\A_j \in \W^*$.
	By picking $\u^* = \A_i$ and $\w^* = \A_j$, we obtain the desired contradiction.

	As these cases cover all the edge configurations that are not contained in $\edgeconst_{\Delta,\beta}(v',x')$ (up to exchanging $\U''$ and $\W''$), and none of the cases can occur, it follows that the output pair on edge $\{ u, w \}$ is an edge configuration contained in $\edgeconst_{\Delta,\beta}(v',x')$.
	Hence, the output produced by our $0$-round algorithm is correct.
\end{proof}

We now prove a lemma that is basically the result of applying Lemma \ref{lem:rerere} multiple times. This is our main result of this section; it essentially shows a lower bound for ruling sets for the port numbering model, with the additional property of providing a family of problems that satisfies the round elimination property, such that all problems of the family have a number of labels that is not too large.
\begin{lemma}\label{lem:lbfamily}
	Let $\beta \geq 1$, and let $t = 1/2 \cdot \log \Delta / (\beta \log \log \Delta)$.
	If $\beta \leq t$, then there exists a sequence $\Pi_0 \rightarrow \Pi_1 \rightarrow \dots \rightarrow \Pi_t$ of problems such that
	\begin{enumerate}[label=(\roman*)]
		\item \label{item:eins} $\Pi_0 = \Pi_{\Delta,\beta}([1, 0, 0, \dots, 0], 0)$,
		\item \label{item:zwei} $\Pi_{j+1}$ can be solved in $0$ rounds given a solution to $\rere(\re(\Pi_j))$, for all $0 \leq j \leq t-1$,
		\item \label{item:drei} $\Pi_{t-1}$ cannot be solved in $0$ rounds in the deterministic PN model.
	\end{enumerate}
\end{lemma}
\begin{proof}
	For all $0 \leq j \leq t$, define $\Pi_j := \Pi_{\Delta,\beta}(v^{(j)}, x^{(j)})$, where $v^{(j)} = [v^{(j)}_0, \dots, v^{(j)}_\beta]$ and $ x^{(j)}$ are recursively defined by $v^{(j+1)}_k := \sum_{i=0}^k v^{(j)}_i$, for all $0 \leq k \leq \beta$, and $x^{(j+1)} :=  \size(v^{(j)}) \cdot (x^{(j)} + 1)$, where we set $v^{(0)} := [1, 0, 0, \dots, 0]$ and $x^{(j)} := 0$.
	This definition immediately ensures Property~$\ref{item:eins}$.
	In order to show Property~$\ref{item:zwei}$, we would like to use Lemma~\ref{lem:rerere}, so we have to make sure that $\size(v^{(j)})$ and $x^{(j)}$ satisfy the condition given in Lemma~\ref{lem:rerere} for all $0 \leq j \leq t-1$.
	To this end, we can apply Lemma \ref{lem:colorgrowth}, that proves that for all $1 \leq j \leq t$ and all $0 \leq k \leq \beta$, we have $v^{(j)}_k = {j+k-1 \choose k}$.
	
	Hence, we can now bound $\size(v^{(j)})$ and $x^{(j)}$ from above.
	Since $\size(v^{(j)})$ and $x^{(j)}$ are monotonically increasing (with increasing $j$), it suffices to bound $\size(v^{(t-1)})$ and $x^{(t-1)}$ to obtain bounds for $\size(v^{(j)})$ and $x^{(j)}$, for all $0 \leq j \leq t-1$.
	As $\size(v^{(t-1)}) = v^{(t)}_\beta$ by definition of $v^{(t)}_\beta$, we have $\size(v^{(t-1)}) = {t+\beta-1 \choose \beta} \leq (2t)^\beta-1$ since $\beta \leq t$.
	This implies that $x^{(j)} \leq (2t)^{j\beta}-1$, for all $0 \leq j \leq t-1$, by induction: clearly, $x^{(0)} = 0 \leq (2t)^{0}-1$, and using the induction hypothesis $x^{(j-1)} \leq (2t)^{(j-1)\beta}-1$, we obtain $x^{(j)} \leq \size(v^{(j-1)}) \cdot (x^{(j-1)} + 1) \leq ((2t)^\beta-1) \cdot (2t)^{(j-1)\beta} \leq (2t)^{j\beta}-1$, for all $1 \leq j \leq t-1$.
	In particular, $x^{(t-1)} \leq (2t)^{(t-1)\beta}-1$.
	Hence, for any $0 \leq j \leq t-1$, we have $\size(v^{(j)}) \cdot (x^{(j)} + 1) + 1 \leq \size(v^{(t-1)}) \cdot (x^{(t-1)} + 1) + 1 \leq ((2t)^\beta-1) \cdot (2t)^{(t-1)\beta} + 1 \leq (2t)^{(t\beta)} \leq (\log \Delta / (\beta \log \log \Delta))^{\log \Delta / (2 \log \log \Delta)} \leq \Delta$.
	Therefore, for all $0 \leq j \leq t-1$, the condition $\size(v) \cdot (x+1) + 1 \leq \Delta$ in Lemma~\ref{lem:rerere} is satisfied by setting $v := v^{(j)}$ and $x := x^{(j)}$, and Property~$\ref{item:zwei}$ follows by applying Lemma~\ref{lem:rerere} inductively.

	Now, we will prove Property~$\ref{item:drei}$.
	For a contradiction, assume that there is a deterministic $0$-round algorithm solving $\Pi_{t-1}$ in the PN model.
	As the algorithm is deterministic, no node has any information about other nodes, and the nodes of the input graph are not distinguished by unique identifiers, the labels that the nodes output must form the same node configuration around every node.
	Let $\mathcal U = \U_1 \s \dots \s U_\Delta$ denote this node configuration, and w.l.o.g., assume that each node $u$ outputs $\U_k$ on the incident edge connected to $u$ via port $k$.
	In order for the algorithm to be correct, we must have $\mathcal U \in \nodeconst_{\Delta,\beta}(v^{(t-1)},x^{(t-1)})$.
	Since, by the above calculations, $x^{(t-1)} \leq \Delta - 1$, we see that the number of wildcards in the problem definition is strictly less than $\Delta$.
	Hence, we have $\U_k = \A_i$ or $\U_k = \C$, for some $1 \leq k \leq \Delta$, some index $1 \leq i \leq \beta$, and some color $\C \in \mathcal C$, according to the definition of $\Pi_{t-1} = \Pi_{\Delta,\beta}(v^{(t-1)},x^{(t-1)})$.
	Now consider a graph, where some edge $e = \{ u, w \}$ is connected to both endpoints via port $k$.
	If $\U_k = \A_i$, then the output produces the edge configuration $\A_i \s \A_i$ on edge $e$; if $\U_k = \C$, then the output produces the edge configuration $\C \s \C$ on $e$.
	Since both configurations are not contained in $\edgeconst_{\Delta,\beta}(v^{(t-1)},x^{(t-1)})$, we obtain the desired contradiction.
	This concludes the proof of Property~$\ref{item:drei}$.

\end{proof}

\section{Lifting results to the LOCAL model}\label{sec:liftlocal}
In the previous sections we obtained upper and lower bounds for the port numbering model. While upper bounds trivially hold also for the LOCAL model, lower bounds need some additional analysis to be lifted to the LOCAL model. The main challenge here is that the round elimination theorem does not tolerate the presence of node identifiers.
Previous works that used the round elimination technique to prove lower bounds \cite{Balliu2019,binary,trulytight} followed a common approach to lift an $f(\Delta)$ lower bound for the port numbering model to a lower bound as a function of $\Delta$ and $n$ for the LOCAL model, and we will follow the same approach to lift our lower bounds as well. We perform such a lifting of our results in the (usual) following way:
\begin{itemize}
	\item We first adapt our lower bound proof to randomized algorithms, showing that if $\Pi_{\Delta,\beta}(v,x)$ can be solved in $t$ rounds with local failure probability $p$, then $\rere(\re(\Pi_{\Delta,\beta}(v,x)))$ can be solved in $t-1$ rounds with some local failure probability $p'$.
	\item We then give a lower bound on the failure probability of any algorithm that runs ``too fast'', that is, in strictly fewer rounds than the $t$ given in Lemma~\ref{lem:lbfamily} (i.e., than the implicit lower bound proved in Section~\ref{sec:lb}). This yields a lower bound for the runtime of any randomized algorithm in the port numbering model as a function of $\Delta$ and $n$. In particular, we obtain a lower bound of $\Omega(\min\{f(\Delta), \log_\Delta \log n\})$.
	\item We make $\Delta$ as large as possible, as a function of $n$, in order to obtain the best possible lower bound (for randomized algorithms in the port numbering model) as a function of $n$. In other words, we choose the value of $\Delta$ that makes $f(\Delta)$ equal to $\log_\Delta \log n$.
	\item Randomized algorithms in the port numbering model can generate unique IDs with high probability and then simulate algorithms for the LOCAL model. Thus we get a lower bound for randomized algorithms in the LOCAL model.
	\item Lower bounds for randomized algorithms in the LOCAL model imply lower bounds for deterministic algorithms as well. We use standard techniques to obtain \emph{better} lower bounds (as a function of $n$) for deterministic algorithms.
\end{itemize}
We will provide lower bounds for high-girth regular graphs, in particular in graphs where the girth is larger than the running time. Such lower bounds directly apply on trees as well. 

\subsection{Evolution of local failure probability}
Balliu et al.~\cite{binary} proved that, given some problem $\Pi$ defined on graphs of degree $\Delta$ using labels from a set $\Sigma$, we can upper bound the local failure probability $p'$ of any algorithm solving $\rere(\re(\Pi))$ in $t-1$ rounds by a function that depends only on $\Delta$, $|\Sigma|$, and $p$, where $p$ is an upper bound on the local failure probability of an algorithm solving $\Pi$ in $t$ rounds. More formally, the authors prove the following result in \cite{binary} (rephrased for our purposes), which is a version of the round elimination theorem that applies to randomized algorithms.

\begin{lemma}[Lemma 41 of \cite{binary}]\label{lem:singlestep}
	Let $A$ be a randomized $t$-round algorithm for $\Pi$ with local failure probability at most $p$ (where $t>0$). Then there exists a randomized $(t-1)$-round algorithm $A'$ for $\rere(\re(\Pi))$ with local failure probability  $p'' \le 2^\frac{1}{\Delta+1} (\Delta |\Sigma'|)^\frac{\Delta}{\Delta+1} {p'}^\frac{1}{\Delta+1} + p'$, where $p' \le 2^\frac{1}{\Delta+1} (\Delta |\Sigma|)^\frac{\Delta}{\Delta+1} p^\frac{1}{\Delta+1} + p$ and $\Sigma'$ is the label set of $\re(\Pi)$.
\end{lemma}
This lemma basically states that if there exists an algorithm that solves some problem $\Pi$ in $t$ rounds with some small failure probability $p$, then there also exists a faster algorithm that solves $\rere(\re(\Pi))$ with some possibly larger but still small enough failure probability. Let $\hat{\re}(\Pi) = \rere(\re(\Pi))$, and let $\hat{\re}^j(\Pi)$ be the function that recursively applies the $\hat{\re}$ function $j$ times. By using Lemma~\ref{lem:singlestep} multiple times, we can give an upper bound on the failure probability of an algorithm solving $\hat{\re}^j(\Pi)$ in $\max\{0,t-j\}$ rounds, as a function of an upper bound on the failure probability of an algorithm solving $\Pi$ in $t$ rounds.

Unfortunately, we cannot directly apply this lemma multiple times to $\Pi_{\Delta,\beta}([1,0,\ldots,0],0)$, as we do not know an upper bound on the number of labels of the problem $\hat{\re}^{j-1}(\Pi_{\Delta,\beta}([1,0,\ldots,0],0))$ and of all the intermediate problems (notice that the number of labels appears in the bound given by Lemma \ref{lem:singlestep}). But what we can do instead is to consider the family of problems $\Pi_i$ defined in Lemma \ref{lem:lbfamily}, where $\Pi_0$ is the locally checkable version of the $(2,\beta)$-ruling set problem, and $\Pi_{i+1}$ is a problem that can be obtained by relaxing the configurations in the node constraint of $\rere(\re(\Pi_{i}))$. By applying Lemma \ref{lem:singlestep} we can conclude that, if we have an algorithm for $\Pi_i$ running in $t>0$ rounds that fails with probability at most $p$, then there exists an algorithm solving $\Pi_{i+1}$ running in $t-1$ rounds that fails with probability at most $p''$.
Moreover, for any problem $\Pi_i$ where index $i$ is at most the $t$ from Lemma~\ref{lem:lbfamily}, we have that the number of colors, and hence the number of distinct configurations in the node constraint of $\Pi_i$, is in $O(\Delta)$.
We now prove a bound on the failure probability obtained by applying Lemma \ref{lem:singlestep} multiple times. We will prove a general statement that not only applies to the family of problems defined in Lemma \ref{lem:lbfamily}, but to all problems for which the number of labels is bounded by $O(\Delta^2)$. Since the number of labels that we consider is bounded by $O(\Delta^2)$, and since the number of labels of the intermediate problems is always at most exponential in the original number of labels, we can now apply the lemma multiple times, assuming $|\Sigma| = O(\Delta^2)$ and $|\Sigma'| = 2^{O(\Delta^2)}$ for the whole family. Hence, we prove the following lemma.

\begin{lemma}\label{lem:multiplesteps}
	Let $\Pi_0 \rightarrow \Pi_1 \rightarrow \dots \rightarrow \Pi_t$ be a sequence of problems such that $\Pi_{i+1}$ can be solved in $0$ rounds given a solution for $\rere(\re(\Pi_i))$, and such that the number of labels of each problem $\Pi_i$ is upper bounded by $O(\Delta^2)$.
	Let $A$ be a randomized $t$-round algorithm for $\Pi_0$ with local failure probability at most $p$. Then there exists a randomized $(t-j)$-round algorithm $A'$ for $\Pi_{j}$ with local failure probability at most $2^{K\Delta^2} p^{1/(\Delta+1)^{2j}}$ for some constant $K$, for all $0<j \le t$, for large enough $\Delta$.
\end{lemma}
\begin{proof}
	For large enough $\Delta$, and for our choice of parameters, by Lemma \ref{lem:singlestep} we have that 
	\begin{align*}
	p'' &\le 2^\frac{1}{\Delta+1} (\Delta |\Sigma'|)^\frac{\Delta}{\Delta+1} {p'}^\frac{1}{\Delta+1} + p' \\
	&\le (\Delta |\Sigma'|)^\frac{1}{\Delta+1} (\Delta |\Sigma'|)^\frac{\Delta}{\Delta+1} {p'}^\frac{1}{\Delta+1} + p'\\
	&\le (\Delta |\Sigma'|) {p'}^\frac{1}{\Delta+1} + p'\\
	&\le (2^{\log\Delta} 2^{|\Sigma|}) {p'}^\frac{1}{\Delta+1} + p'\\
	&\le 2^{O(\Delta^2)} {p'}^\frac{1}{\Delta+1}.
	\end{align*}
	Similarly,
		\begin{align*}
	p' &\le 2^\frac{1}{\Delta+1} (\Delta |\Sigma|)^\frac{\Delta}{\Delta+1} {p}^\frac{1}{\Delta+1} + p \\
	&\le (\Delta |\Sigma'|)^\frac{1}{\Delta+1} (\Delta |\Sigma'|)^\frac{\Delta}{\Delta+1} {p}^\frac{1}{\Delta+1} + p\\
	&\le 2^{O(\Delta^2)} {p}^\frac{1}{\Delta+1}.
	\end{align*}
	Hence, this implies that
	\begin{align*}
	p''&\le 2^{K \Delta^2} p'^{\frac{1}{\Delta+1}}, \text{ where }\\
	p' &\le 2^{K \Delta^2} p^{\frac{1}{\Delta+1}},
	\end{align*}
	for some constant $K$. By recursively applying Lemma \ref{lem:singlestep} we get the following:
	\[
		p_j \le 2^{K \Delta^2} p_{j-1}^{\frac{1}{\Delta+1}},
	\]
	where $p_0=p$ and $p_{2j}$, are, respectively, the local failure probability bounds for $\Pi_{0}$ and $\Pi_{j}$. We prove by induction that for all $j>0$,
	\[
		p_j \le 2^{K \sum_{i=1}^{j} \Delta^{3-i}} p^{\frac{1}{(\Delta+1)^j}}.
	\]
	For the base case where $j=1$, we get that $p_1 \le 2^{K \Delta^2} p^{\frac{1}{\Delta+1}}$, which holds, as we showed above. Let us assume that the claim holds for $j$, and let us prove it for $j+1$.
	\begin{align*}
			p_{j+1} &\le 2^{K \Delta^2} p_{j}^{\frac{1}{\Delta+1}} \le 2^{K \Delta^2} \left(  2^{K \sum_{i=1}^{j} \Delta^{3-i}} p^{\frac{1}{(\Delta+1)^j}} \right)^{\frac{1}{\Delta+1}} \\
			&\le 2^{K \Delta^2} 2^{K \sum_{i=1}^{j} \Delta^{3-i-1}} p^{\frac{1}{(\Delta+1)^{j+1}}} =  2^{K \sum_{i=1}^{j+1} \Delta^{3-i}} p^{\frac{1}{(\Delta+1)^{j+1}}} 
	\end{align*}
	
	Since for $\Delta \ge 2$ and for any $j$, $\sum_{i=1}^{j} \Delta^{3-i} \le \Delta^2 + \Delta + 2$, we get the following:
	\[p_{j} \le 2^{K(\Delta^2 + \Delta + 2)} p^{\frac{1}{(\Delta+1)^j}} \le 2^{K'\Delta^2} p^{\frac{1}{(\Delta+1)^j}},
	\] for some constant $K'$, hence the claim follows.
\end{proof}

We now give a lower bound on the failure probability of any algorithm solving $\Pi_i$ in $0$ rounds, for any $i < t$, where $t$ is the lower bound obtained in Section \ref{sec:lb} for the time complexity of $\Pi_0$. Again, we will prove a stronger result, that applies to any family of problems where the number of allowed configurations is $O(\Delta^2)$, a condition that is satisfied by the family of Lemma \ref{lem:lbfamily}, since there are $\beta = o(\Delta)$ possible pointer configurations and at most $O(\Delta)$ color configurations. 
\begin{lemma}\label{lem:zerorounds}
	Let $\Pi$ be a problem that cannot be solved in $0$ rounds with deterministic algorithms in the port numbering model, such that the node constraint $\nodeconst$ contains $O(\Delta^2)$ allowed configurations. Then any $0$-round algorithm solving $\Pi$ must fail with probability at least $1/\Delta^8$.
\end{lemma}
\begin{proof}
	We follow the same strategy as in \cite[Lemma 6.4]{trulytight}. Since in $0$ rounds of communication all nodes have the same information, we can see any $0$-round algorithm as a probability assignment to each node configuration. That is, for each $\c_i \in \nodeconst$, the algorithm outputs $\c_i$ with probability $p_i$, such that $\sum p_i = 1$. Hence, by the pigeonhole principle, there exists some configuration $\bar{\c}$ that all nodes output with probability at least $\frac{1}{K \Delta^2}$ for some constant $K$. Since by assumption $\Pi$ is not $0$-rounds solvable, then there exists no node configuration such that, for all pairs of choices $\ell_1 \s \ell_2$ over such configuration, $\ell_1 \s \ell_2$ is in $\edgeconst$, i.e., in the edge constraint of $\Pi$ (otherwise the problem would be $0$-round solvable in the port numbering model, by making all nodes output that configuration). Thus, there exist $2$ (possibly the same) labels $\ell_1$ and $\ell_2$ appearing in $\bar{\c}$ such that $\ell_1\s \ell_2$ is not in $\edgeconst$. Moreover, conditioned on the fact that the node is outputting the configuration $\bar{\c}$, since $\ell_1$ and $\ell_2$ appear at least once in such configuration, there must exist some ports $i$ and $j$ where $\ell_1$ and $\ell_2$ are written with probability at least $\frac{1}{\Delta}$. Thus, two neighboring nodes connected through ports $i$ and $j$ will both output configuration $\bar{\c}$ and produce the invalid edge configuration $\ell_1 \s \ell_2$ with probability at least $\frac{1}{(K \Delta^3)^2}$, that for large enough $\Delta$ is at least $1 / \Delta^8$.
\end{proof}

We now bound the failure probability of any algorithm solving $\Pi_0$ in strictly less than $t$ rounds.
\begin{lemma}\label{lem:randomizedpnlb}
	Let $\Pi_0 \rightarrow \Pi_1 \rightarrow \dots \rightarrow \Pi_t$ be a sequence of problems such that $\Pi_{i+1}$ can be solved in $0$ rounds given a solution for $\rere(\re(\Pi_i))$, the number of labels of each problem $\Pi_i$ is upper bounded by $O(\Delta^2)$, and the node constraint of each problem $\Pi_i$ contains $O(\Delta^2)$ allowed configurations. Let $t$ be a number satisfying that, for all $t' < t$, $\Pi_{t'}$ is not $0$-round solvable in the port numbering model using deterministic algorithms. Any algorithm for $\Pi_0$ running in strictly less than $t$ rounds must fail with probability at least $\frac{1}{2^{\Delta^{4t}}}$, if $\Delta$ is large enough.
\end{lemma}
\begin{proof}
	By applying Lemma \ref{lem:multiplesteps} we get that an algorithm solving $\Pi_0$ in $t'<t$ rounds with local failure probability at most $p$ implies an algorithm solving $\Pi_{t'}$ in $0$ rounds with local failure probability at most $2^{K\Delta^2} p^{1/(\Delta+1)^{2t'}}$. Then, since  $\Pi_{t'}$ is not $0$-round solvable in the port numbering model using deterministic algorithms by assumption, by applying Lemma \ref{lem:zerorounds} we get the following:
	\[
		 2^{K\Delta^2} p^{1/(\Delta+1)^{2t'}} \ge \frac{1}{\Delta^8},
	\]
	that for large enough $\Delta$ implies the following:
	\[
		p \ge   \frac{1}{2^{(K\Delta^2 + 8\log\Delta)(\Delta+1)^{2t'}}} \ge \frac{1}{2^{\Delta^3(\Delta+1)^{2t'}}}  \ge \frac{1}{2^{\Delta^{4t'+3}}} \ge \frac{1}{2^{\Delta^{4t}}}
	\]
\end{proof}

\subsection{Making $\Delta$ as large as possible}
By Lemma \ref{lem:lbfamily} we know that, for $t = 1/2 \cdot \log \Delta / (\beta \log \log \Delta)$, if $\beta \le t$, problem $\Pi_{\Delta,\beta}([1,0,\ldots,0],0)$ is not $t$-round solvable in the port numbering model using deterministic algorithms, and each problem of the family used to prove this result uses $O(\Delta)$ labels, and $O(\Delta)$ node configurations. Hence, by combining Lemma \ref{lem:lbfamily} and Lemma \ref{lem:randomizedpnlb} we get that any randomized algorithm for problem $\Pi_{\Delta,\beta}([1,0,\ldots,0],0)$ running in $t = o(\frac{\log \Delta}{\beta \log \log \Delta})$ rounds must fail with probability at least $\frac{1}{2^{\Delta^{4t}}}$, in some $\Delta$-regular neighborhood of girth at least $2t+2$ (condition required by the round elimination theorem, see Theorem \ref{thm:sebastien}), in the randomized port numbering model, provided $\beta \le t$. 

 We are now ready to lift this bound to the LOCAL model. The main idea is that, since in the randomized port numbering model nodes can generate unique IDs with high probability, the existence of an algorithm for the randomized LOCAL model with some failure probability directly implies the existence of an algorithm for the randomized port numbering model with roughly the same failure probability (up to a factor of $1-1/n^c$ for an arbitrary constant $c$, that is the success probability of the randomized ID generation process). Hence, a lower bound for the randomized port numbering model directly applies to the LOCAL model as well. We first prove Theorem \ref{thm:randlb}, and then we obtain Corollary \ref{cor:randlb} by taking a value of $\Delta$, as a function of $n$, that maximizes the result of the theorem.
 \randlb*
 \begin{proof}
 	As mentioned above, any randomized algorithm solving $\Pi_{\Delta,\beta}([1,0,\ldots,0],0)$ in the LOCAL model running in $t = o(\frac{\log \Delta}{\beta \log \log \Delta})$ rounds must fail with probability at least $\frac{1}{2^{\Delta^{4t}}}$, in some $\Delta$-regular neighborhood of girth at least $2t+2$, provided that  $\beta \le t$. Later we will show that, for the choice of parameters stated in the theorem, we have $t-\beta = \Omega(t)$, which implies that $\beta \le t$. Hence, to prove the theorem, we need to show how large we can make $t$ such that:
 	\begin{itemize}
 		\item the failure probability is still too large, that is, $\frac{1}{2^{\Delta^{4t}}} > 1/n$, and
 		\item there exists a $\Delta$-regular graph of girth at least $2t+2$.
 	\end{itemize}
  	The first requirement is satisfied (for all sufficiently large $n$) if $t = o(\log_{\Delta} \log n)$. The second requirement is satisfied if $t = o(\log_\Delta n)$, since from extremal graph theory we know that, for infinite values of $n$, there exist $\Delta$-regular graphs of girth $\Omega(\log_\Delta n)$ (see, e.g., \cite{highgirth}). By combining the obtained lower bounds, we get that any randomized algorithm that succeeds w.h.p.\ requires at least  $t = \Omega\left(\min \left\{  \frac{\log \Delta}{\beta \log \log \Delta}  ,  \log_\Delta \log n \right\} \right)$ rounds.
  	
  	Unfortunately, as discussed in Section \ref{sec:problems}, $\Pi_{\Delta,\beta}([1,0,\ldots,0],0)$ may be (at most) $\beta$ rounds harder than the $(2,\beta)$-ruling set problem, since $\Pi_{\Delta,\beta}([1,0,\ldots,0],0)$ requires to output some additional pointers that make the solution locally checkable. Hence, a $t$-round lower bound for $\Pi_{\Delta,\beta}([1,0,\ldots,0],0)$ only implies a lower bound of $t-\beta$ for $(2,\beta)$-ruling sets. We now prove that, for the range of values of $\beta$ stated in the theorem, the lower bound obtained for $\Pi_{\Delta,\beta}([1,0,\ldots,0],0)$ also applies to $(2,\beta)$-ruling sets, which implies the theorem. For this purpose, we prove that $t - \beta = \Omega(t)$, by considering two cases.
  	\begin{itemize}
  		  \item $\sqrt{\frac{\log \Delta}{\log \log \Delta}} \le  \log_\Delta \log n$. In this case, $\beta \le c \cdot \sqrt{\frac{\log \Delta}{\log \log \Delta}}$, and hence $\frac{\log \Delta}{\beta \log \log \Delta} \ge \frac{1}{c} \cdot\sqrt{\frac{\log \Delta}{\log \log \Delta}}$. We obtain that $t-\beta = \Omega\left( \min\left( \frac{\log \Delta}{\beta \log \log \Delta }-c\cdot\sqrt{\frac{\log \Delta}{\log \log \Delta}}, \log_\Delta \log n -  c \cdot \sqrt{\frac{\log \Delta}{\log \log \Delta}} \right) \right)$, which, for small enough $c$, is $\Omega(t)$.
  		\item $\log_\Delta \log n \le \sqrt{\frac{\log \Delta}{\log \log \Delta}} $. In this case, $\beta \le c \cdot \log_\Delta \log n $. \\We obtain that $t-\beta = \Omega\left( \min\left( \frac{\log \Delta}{\beta \log \log \Delta }-c \cdot\log_\Delta \log n, (1-c) \cdot\log_\Delta \log n \right) \right)$, which, for small enough $c$, is $\Omega(t)$.

  	\end{itemize}
 \end{proof}

 We are now ready to prove Corollary \ref{cor:randlb}.
 \randcor*
 \begin{proof}
 	We apply Theorem \ref{thm:randlb} where we set $\Delta = 2^{\sqrt{\beta \log \log n \log \log \log n}}$. We obtain
 	
 	\[
 	\frac{\log \Delta}{\beta \log \log \Delta} = \frac{\sqrt{\beta \log \log n \log \log \log n}}{\beta \log \sqrt{\beta \log \log n \log \log \log n}} = \Omega\left(  \sqrt  \frac{ \log \log n }{\beta \log \log \log n  }   \right) \enspace, \text{ and}
 	\]
 	\[
 	\log_\Delta \log n = \frac{\log \log n}{ \sqrt{\beta \log \log n \log \log \log n} } =   \sqrt  \frac{ \log \log n }{\beta \log \log \log n  } \enspace.
 	\]
 	Hence, we obtain a lower bound of $t = \Omega\left(\min \left\{\frac{\log \Delta}{\beta \log \log \Delta} , \log_\Delta \log n  \right\} \right) =  \Omega\left(  \sqrt  \frac{ \log \log n }{\beta \log \log \log n  }   \right)$.
Note that by choosing $c$ small enough (depending on the $c$ in Theorem~\ref{thm:randlb}), we can ensure that any $\beta \le c \cdot \sqrt[3]{ \frac{\log \log n}{\log \log \log n}}$ satisfies the condition for $\beta$ given in Theorem~\ref{thm:randlb}.
 \end{proof}

\subsection{Stronger deterministic lower bound}
We will now prove that, assuming that there is a fast deterministic algorithm, we can construct an even faster randomized algorithm, contradicting the results proved in Theorem \ref{thm:randlb}.

\detlb*

\begin{proof}	
	 Assume for a contradiction that an algorithm violating the claimed bound exists and let $t(n)$ be its runtime on $n$-node graphs. We use this algorithm to construct an algorithm running within a time bound that would violate Theorem \ref{thm:randlb}. The idea is to execute such an algorithm by lying about the size of the network, by claiming that it is of size $N = \log n$.
This immediately yields an algorithm with a runtime of $O(t(N)) = O(t(\log n))$, which contradicts the bound given in Theorem~\ref{thm:randlb}.
This algorithm is even deterministic, but, in the context of Theorem~\ref{thm:randlb}, can equivalently be seen as a randomized algorithm with failure probability $0$.

In order to use this idea, we need to make sure that the algorithm does not detect such a lie. Using standard arguments (see e.g.\ \cite[Theorem 4.5]{chang16exponential}), one can show that the following conditions are sufficient for this purpose:
	 \begin{itemize}
	 	\item We need to compute a new ID assignment, with IDs of $O(\log \log n)$ bits, in $O(t(N))$ rounds. This ensures that the algorithm does not detect that IDs are too large. It is sufficient that these IDs are unique in each $(t(N)+1)$-radius neighborhood.
	 	\item Each $t(N)$-radius neighborhood must contain at most $N$ nodes, so that the algorithm does not detect that there are more nodes than the claimed amount.
	 \end{itemize}
	 
	 In order to compute the new ID assignment we can compute a $k = (\Delta^{2(2t(N)+2)}( \log \log N + \log(\Delta^{2t(N)+2})))$-coloring of $G^{2t(N)+2}$, and this can be done in $O(t(N))$ rounds using Linial's coloring algorithm~\cite[Corollary 4.1]{Linial1992}.
Since $t(N) = o(\log_\Delta N)$, we see that each $t(N)$-radius neighborhood contains $O(\Delta^{t(N)}) = o(N)$ nodes and that $k =  o(N)$.
Hence, the supposed algorithm cannot detect the lie, which concludes the proof.
\end{proof}

Corollary \ref{cor:detlb} follows by applying Theorem \ref{thm:detlb} with $\Delta = 2^{\sqrt{\beta \log n \log \log n}}$.

\section{Conclusion and open problems}
In this work, we proved that the deterministic complexity of computing $(2,\beta)$-ruling sets is at least $\poly\log n$, unless $\beta$ is too large. Combined with existing $\poly\log n$ upper bounds, our results imply that the deterministic complexity of ruling sets lies in the $\poly\log n$ region. An interesting open question is how many $\log$ factors are required exactly.

Another open question concerns the techniques that we use: We first prove a lower bound for a constant-radius checkable version of the ruling set problem, and then transform this bound into a lower bound for the original problem, where we lose an additive $\beta$-factor. Hence, currently our technique is not capable to prove lower bounds for $(2,\beta)$-ruling sets where $\beta$ is so large that the problem can be solved in $o(\beta)$ rounds. An open question is to get rid of this restriction. For example, can we show that finding $(2,\log^{1/3} n)$-ruling sets requires $\Omega(\log^{1/3} n / \log^{1/2} \log n)$ rounds with deterministic algorithms?

While we now have a good picture of the dependency of the complexity of $(2,\beta)$-ruling sets on $n$, the dependency on $\Delta$ is far less clear. Even in the case where a $(\Delta+1)$-vertex coloring is provided, the current best upper bound is $O(\beta \Delta^{1/\beta})$ rounds. Note that, for constant values of $\beta$, we get a complexity that is polynomial in $\Delta$, while the lower bound that we provide lies in the $\poly \log \Delta$ region. Hence, there is an exponential gap between the current upper and lower bounds, as a function of $\Delta$.  We know that, on general graphs, any algorithm that solves MIS in time $f(\Delta) + g(n)$, must either satisfy $f(\Delta) = \Omega(\Delta)$, or $g(n) = \Omega(\log \log n / \log \log \log n)$ for randomized algorithms and $g(n) = \Omega(\log n / \log \log n)$ for deterministic ones  \cite{Balliu2019}, and we think that a necessary step for really understanding the $\Delta$-dependency of $(2,\beta)$-ruling sets is to first prove an $\Omega(\Delta)$ lower bound for MIS on trees.

Finally, a number of interesting open questions revolve around our new technique of proving a lower bound via an upper bound. Can we characterize the problems that allow such an approach? What properties does an algorithm have to have to be well-representable as an upper bound sequence? Can the related technique of introducing a coloring component into non-coloring problems be successfully applied to other problems? We believe that finding answers to these and related questions will constitute an important step towards a better understanding of the round elimination technique. 

\section*{Acknowledgments}

We are very thankful to the anonymous reviewers of the conference version of this paper for their fruitful comments. We would also like to thank Mohsen Ghaffari, Fabian Kuhn, Yannic Maus, and Julian Portmann for helpful comments on related works.

\urlstyle{same}
\bibliographystyle{plainnat}
\bibliography{mis-lower-bound}

\end{document}